\documentclass[11pt]{llncs}
\usepackage{amsmath,amssymb}
\usepackage{graphicx}
\usepackage[linesnumbered, vlined, ruled]{algorithm2e}
\usepackage{fullpage}
\usepackage{lineno}

\usepackage[top=0.8in, bottom=0.9in, left=0.9in, right=0.9in]{geometry}

\def\calI{\mathcal{I}}

\def\calL{\mathcal{L}}
\def\calQ{\mathcal{Q}}
\def\calP{\mathcal{H}}
\newcommand{\CH}{\mbox{$C\!H$}}
\newtheorem{observation}{Observation}

\begin{document}

\title{On the Planar Two-Center Problem and Circular Hulls\thanks{A preliminary version of this paper will appear in the Proceedings of the 36th International Symposium on Computational Geometry (SoCG 2020).}}
\author{
Haitao Wang
}
\institute{
Department of Computer Science\\
Utah State University, Logan, UT 84322, USA\\
\email{haitao.wang@usu.edu}\\
}

\maketitle

\pagestyle{plain}
\pagenumbering{arabic}
\setcounter{page}{1}

\vspace{-0.2in}
\begin{abstract}
Given a set $S$ of $n$ points in the Euclidean plane, the two-center problem is to find two congruent disks of smallest radius whose union covers all points of $S$. Previously, Eppstein [SODA'97] gave a randomized algorithm of $O(n\log^2n)$ expected time and Chan [CGTA'99] presented a deterministic algorithm of $O(n\log^2 n\log^2\log n)$ time. 
In this paper, we propose an $O(n\log^2 n)$ time deterministic algorithm, which improves Chan's deterministic algorithm and matches the randomized bound of Eppstein. If $S$ is in convex position, then we solve the problem in $O(n\log n\log\log n)$ deterministic time. Our results rely on new techniques for dynamically maintaining circular hulls under point insertions and deletions, which are of independent interest.
\end{abstract}


\section{Introduction}
\label{sec:intro}
In this paper, we consider the planar 2-center problem. Given a set $S$ of $n$ points in the Euclidean plane, we wish to find two congruent disks of smallest radius whose union covers all points of $S$.

The classical 1-center problem for a set of points is to find the smallest disk covering all points, and the problem can be solved in linear time in any fixed dimensional space~\cite{ref:ChazelleOn96,ref:DyerOn86,ref:MegiddoLi83}. As a natural generalization, the 2-center problem has attracted much attention. Hershberger and Suri~\cite{ref:HershbergerFi91} first solved the decision version of the problem in $O(n^2\log n)$ time, which was later improved to $O(n^2)$ time~\cite{ref:HershbergerA93}. Using this result and parametric search~\cite{ref:MegiddoAp83}, Agarwal and Sharir~\cite{ref:AgarwalPl94} gave an $O(n^2\log^3 n)$ time algorithm for the $2$-center problem. Katz and Sharir~\cite{ref:KatzAn97} achieved the same running time by using expanders instead of parametric search. Eppstein~\cite{ref:EppsteinDy92} presented a randomized algorithm of $O(n^2\log^2n \log\log n)$ expected time.
Later, Jaromczyk and Kowaluk~\cite{ref:JaromczykAn94} proposed an $O(n^2)$ time algorithm. A breakthrough was achieved by Sharir~\cite{ref:SharirA97}, who proposed the first subquadratic algorithm for the problem, and the running time is $O(n\log^9 n)$. Afterwards, following Sharir's algorithmic scheme, Eppstein~\cite{ref:EppsteinFa97} derived a randomized algorithm of $O(n\log^2 n)$ expected time, and then Chan~\cite{ref:ChanMo99} developed an $O(n\log^2 n\log^2 \log n)$ time deterministic algorithm and a randomized algorithm of $O(n\log^2 n)$ time with high probability.
Recently, Tan and Jiang~\cite{ref:TanSi17} proposed a simple algorithm of
$O(n\log^2 n)$ time based on binary search, but unfortunately, the algorithm is not correct (see the appendix for details).
The problem has an $\Omega(n\log n)$ time lower bound in the algebraic decision tree model~\cite{ref:EppsteinFa97}, by a reduction from the max-gap problem.

In this paper, we present a new deterministic algorithm of $O(n\log^2 n)$ time, which improves the $O(n\log^2 n\log^2\log n)$ time deterministic algorithm by Chan~\cite{ref:ChanMo99} and
matches the randomized bound of $O(n\log^2n)$~\cite{ref:ChanMo99,ref:EppsteinFa97}.
This is the first progress on the problem since Chan's work~\cite{ref:ChanMo99} was published twenty years ago.
Further, if $S$ is in convex position (i.e., every point of $S$ is a vertex of the convex hull of $S$), then our technique can solve the 2-center problem on $S$ in $O(n\log n\log\log n)$ time. Previously, Kim and Shin~\cite{ref:KimEf00} announced an $O(n\log^2 n)$ time algorithm for this convex position case, but Tan and Jiang~\cite{ref:TanSi17} found errors in their time analysis.

Some variations of the 2-center problem have also been considered in the literature.
Agarwal et al.~\cite{ref:AgarwalTh98} studied the discrete 2-center problem where the centers of the two disks must be in $S$, and they solved the problem in $O(n^{4/3}\log^5 n)$ time. Agarwal and Phillips~\cite{ref:AgarwalAn08} considered an outlier version of the problem where $k$ points of $S$ are allowed to be outside the two disks, and they presented a randomized algorithm of $O(nk^7\log^3 n)$ expected time.
Arkin et al.~\cite{ref:ArkinBi15} studied a bichromatic 2-center problem for a set of $n$ pairs of points in the plane, and the goal is to assign a red color to a point and a blue color to the other point for every pair, such that $\max\{r_1,r_2\}$ is minimized, where $r_1$ (resp., $r_2$) is the radius of the smallest disk covering all red (resp., blue) points. Arkin et al.~\cite{ref:ArkinBi15} gave an $O(n^3\log^2 n)$ time algorithm, which was recently improved to $O(n^2\log^2 n)$ time by Wang and Xue~\cite{ref:WangIm19}.
The more general $k$-center problem is NP-hard if $k$ is part of the input~\cite{ref:MegiddoOn84}.


\subsection{Our Techniques}

Let $D_1^*$ and $D_2^*$ be two congruent disks in an optimal solution such that the distance of their centers is minimized.
Let $r^*$ be their radius and $\delta^*$ the distance of their centers.
If $\delta^*\geq r^*$, we call it the {\em distant case}; otherwise, it is the {\em nearby case}.

Eppstein~\cite{ref:EppsteinFa97} already solved the distant case in $O(n\log^2 n)$ deterministic time.
Solving the nearby case turns out to be the bottleneck in all previous three sub-quadratic time algorithms~\cite{ref:ChanMo99,ref:EppsteinFa97,ref:SharirA97}. Specifically, Sharir~\cite{ref:SharirA97} first solved it in $O(n\log^9n)$ deterministic time. Eppstein~\cite{ref:EppsteinFa97} gave a randomized algorithm of $O(n\log n\log\log n)$ expected time. Chan~\cite{ref:EppsteinFa97} proposed a randomized algorithm of $O(n\log n)$ time with high probability and another deterministic algorithm of $O(n\log^2 n\log^2\log n)$ time.
Our contribution is an $O(n\log n\log\log n)$ time deterministic algorithm for the nearby case,
which improves Chan's algorithm by a factor of $\log n\log\log n$.
Combining with the $O(n\log^2 n)$ time deterministic algorithm of
Eppstein~\cite{ref:EppsteinFa97} for the distant case, the 2-center
problem can now be solved in $O(n\log^2 n)$ deterministic time. Interestingly,
solving the distant case now becomes the bottleneck of the problem.


Our algorithm (for the nearby case) is based on the framework of Chan~\cite{ref:ChanMo99}.
Our improvement is twofold.
First, Chan~\cite{ref:ChanMo99} derived an $O(n\log n)$ time algorithm
for the {\em decision problem}, i.e., given  $r$, decide whether
$r^*\leq r$. We improve the algorithm to $O(n)$ time, after
	$O(n\log n)$ time preprocessing. Second, Chan~\cite{ref:ChanMo99}
	solved the optimization problem (i.e., the original 2-center
	problem) by parametric search. To this end, Chan developed a
	parallel algorithm for the decision problem and the algorithm runs
	in $O(\log n\log^2 \log n)$ parallel steps using $O(n\log n)$
	processors. By applying Cole's parametric
	search~\cite{ref:ColeSl87} and using his $O(n\log n)$ time
	decision algorithm, Chan solved the optimization problem in
	$O(n\log^2 n\log^2\log n)$ time. We first notice that simply
	replacing Chan's $O(n\log n)$ time decision algorithm with our new
	$O(n)$ time algorithm does not lead to any improvement. Indeed, in
	Chan's parallel algorithm, the number of processors times the
	number of parallel steps is $O(n\log^2 n\log^2\log n)$. We further
	design another parallel algorithm for the decision problem, which
	runs in $O(\log n\log\log n)$ parallel steps using $O(n)$
	processors. Consequently, by applying Cole's parametric search
	with our $O(n)$ time decision algorithm, we solve the optimization
	problem in $O(n\log n\log\log n)$ time. Note that although Cole's
	parametric search is used, our algorithm mainly involves
	independent binary searches and no sorting networks are needed.

In	addition, we show that our algorithm can be easily applied to solving the convex position case in $O(n\log n\log\log n)$ time.

\paragraph{Circular hulls.}
To obtain our algorithm for the decision problem, we develop new techniques for {\em circular hulls}~\cite{ref:HershbergerFi91} (also known as {\em $\alpha$-hulls} with $\alpha=1$~\cite{ref:EdelsbrunnerOn83}). A circular hull of radius $r$ for a set $Q$ of points is the common intersection of all disks of radius $r$ containing $Q$ (to see how circular hulls are related to the two-center problem, notice that there exists a disk of radius $r$ covering all points of $Q$ if and only if the circular hull of $Q$ of radius $r$ exists).
Although circular hulls have been studied before, our result needs more efficient algorithms for certain operations. For example, two algorithms~\cite{ref:EdelsbrunnerOn83,ref:HershbergerFi91} were known for constructing the circular hull for a set of $n$ points; both algorithms run in $O(n\log n)$ time, even if the points are given sorted. We instead present a linear-time algorithm once the points are sorted. Also, Hershberger and Suri~\cite{ref:HershbergerFi91} gave a linear-time algorithm to find the common tangents of two circular hulls separated by a line, and we design a new algorithm of $O(\log n)$ time. We also need to maintain a dynamic circular hull for a set of points under point insertions and deletions. Hershberger and Suri~\cite{ref:HershbergerFi91} gave a semi-dynamic data structure that can support deletions in $O(\log n)$ amortized time each. In our problem, we need to handle both insertions and deletions but with the following special properties: the point in each insertion must be to the right of all points of $Q$ and the point in each deletion must be the leftmost point of $Q$. Our data structure can handle each update in $O(1)$ amortized time (which leads to the linear time decision algorithm for the 2-center problem\footnote{As will be clear later, the points processed in our dynamic circular hull problem are actually sorted radially around a point; we can extend the result for the left-right sorted case to the radically sorted case.}).
We believe that these results on circular hulls are interesting in their own right
and undoubtedly have other applications.

\paragraph{Outline.}
The rest of the paper is organized as follows. We introduce notation
and review some previous work in Section~\ref{sec:pre}. In
Section~\ref{sec:decision}, we present our decision algorithm, and the
algorithm needs a data structure to maintain circular hulls
dynamically, which is given in Section~\ref{sec:circularhull}.
Section~\ref{sec:optimization} solves the optimization problem.
Section~\ref{sec:convexposition} is concerned with
the convex position case.
Section~\ref{sec:commontangent} is devoted to proving a lemma 
that is used in Section~\ref{sec:optimization}.

\section{Preliminaries}
\label{sec:pre}

We begin with some notation, some of which is borrowed from~\cite{ref:ChanMo99}.
It suffices to solve the nearby case. Thus, we assume that $\delta^*<r^*$ in the rest of the paper.
In the nearby case, it is possible to find in $O(n)$ time a constant number of points such that at least one of them, denoted by $o$, is in $D_1^*\cap D_2^*$~\cite{ref:EppsteinFa97}.
We assume that $o$ is the origin of the plane.
We make a general position assumption:
no two points of $S$ are collinear with $o$ and no two points of $S$ have the same $x$-coordinate. This assumption does not affect the running time of the algorithm, but simplifies the presentation.

For any set $P$ of points in the plane, let $\tau(P)$ denote the radius of the smallest enclosing disk of $P$.
For a connected region $B$ in the plane, let $\partial B$ denote the boundary of $B$.

The boundaries of the two disks $D^*_1$ and $D^*_2$ have exactly two intersections, and let $\rho_1$ and $\rho_2$ be the two rays through these intersections emanating from $o$ (e.g., see Fig.~\ref{fig:casenear}). As argued in~\cite{ref:ChanMo99}, one of the two coordinate axes must separate $\rho_1$ and $\rho_2$ since the angle between the two rays lies in $[\pi/2,3\pi/2]$, and without loss of generality, we assume it is the $x$-axis.

\begin{figure}[t]
\begin{minipage}[t]{0.49\textwidth}
\begin{center}
\includegraphics[height=1.5in]{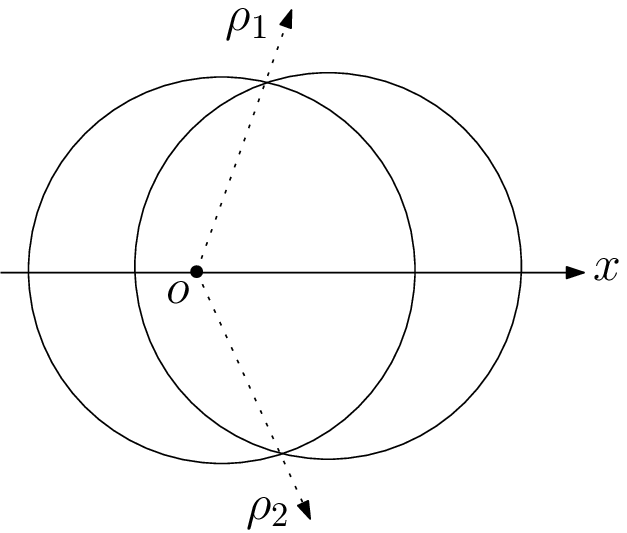}
\caption{\footnotesize Illustrating the nearby case.}
\label{fig:casenear}
\end{center}
\end{minipage}
\hspace{0.02in}
\begin{minipage}[t]{0.49\textwidth}
\begin{center}
\includegraphics[height=1.5in]{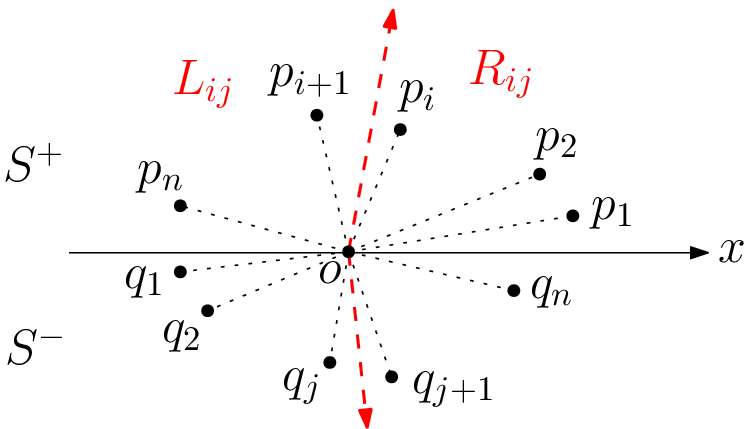}
\caption{\footnotesize Illustrating the points of $S^+$ and $S^-$.}
\label{fig:sort}
\end{center}
\end{minipage}
\vspace{-0.15in}
\end{figure}

Let $S^+$ denote the subset of points of $S$ above the $x$-axis, and $S^-=S\setminus S^+$. For notational simplicity, let $|S^+|=|S^-|=n$. Let $p_1,p_2,\ldots, p_n$ be the sorted list of $S^+$ counterclockwise around $o$, and $q_1,q_2,\ldots, q_n$ the sorted list of $S^-$ also counterclockwise around $o$ (e.g., see Fig.~\ref{fig:sort}).
For each $i=0,1,\ldots, n$ and $j=0,1,\ldots,n$, define $L_{ij}=\{p_{i+1}\ldots, p_n, q_1,\ldots,q_j\}$ and $R_{ij}=\{q_{j+1},\ldots,q_n,p_1,\ldots,p_i\}$. Note that if $i=n$, then $L_{ij}=\{q_1,\ldots,q_j\}$, and if $j=n$, then $R_{ij}=\{p_1,\ldots,p_i\}$. In words, if we consider a ray emanating from $o$ and between $p_i$ and $p_{i+1}$, and another ray emanating from $o$ and between $q_j$ and $q_{j+1}$, then $L_{ij}$ (resp., $R_{ij}$) consist of all points to the left (resp., right) of the two rays (e.g., see Fig.~\ref{fig:sort}).


Note that the partition of $S$ by the two rays $\rho_1\cup \rho_2$ is $\{L_{ij},R_{ij}\}$ for some $i$ and $j$, and thus $r^*=\max\{\tau(L_{ij}),\tau(R_{ij})\}$. Define $A[i,j]=\tau(L_{ij})$ and $B[i,j]=\tau(R_{ij})$, for all $0\leq i, j\leq n$. Then, $r^*=\min_{0\leq i, j\leq n}\max\{A[i,j],B[i,j]\}$.
If we consider $A$ and $B$ as $(n+1)\times (n+1)$ matrices, then each row of $A$ (resp., $B$) is monotonically increasing (resp., decreasing) and each column of $A$ (resp., $B$) is monotonically decreasing (resp., increasing).
For each $i\in [0,n]$, define $r^*_i=\min_{0\leq j\leq n}\max\{A[i,j],B[i,j]\}$. Thus, $r^*=\min_{0\leq i\leq n}r_i^*$.



\subsection{Circular Hulls}

For any point $c$ in the plane and a value $r$, we use $D_r(c)$ to denote the disk centered at $c$ with radius $r$. For a set $Q$ of points in the plane, define $\calI_r(Q)=\bigcap_{c\in Q}D_r(c)$, i.e., the common intersection of the disks $D_r(c)$ for all points $c\in Q$. Note that $\calI_r(Q)$ is convex.
A dual concept of $\calI_r(Q)$ is the {\em circular hull}~\cite{ref:HershbergerFi91} (also known as {\em $\alpha$-hull} with $\alpha=1$~\cite{ref:EdelsbrunnerOn83}; e.g., see Fig~\ref{fig:circularhull}), denoted by $\alpha_r(Q)$, which is the common intersection of all disks of radius $r$ containing $Q$. $\alpha_r(Q)$ is convex and unique. The vertices of $\alpha_r(Q)$ is a subset of $Q$ and the edges are arcs of circles of radius $r$. $\calI_r(Q)$ and $\alpha_r(Q)$ are dual to each other: Every arc of $\alpha_r(Q)$ is on the circle of radius $r$ centered at a vertex of $\calI_r(Q)$ and every arc of $\calI_r(Q)$ is on the circle of radius $r$ centered at a vertex of $\alpha_r(Q)$. Note that $\alpha_r(Q)$ exists if and only if $\calI_r(Q)\neq\emptyset$, which is true if and only $\tau(Q)\leq r$.
For brevity, we often drop the subscript $r$ from $\calI_r(Q)$ and $\alpha_r(Q)$ if it is clear from the context.

\begin{figure}[t]
\begin{minipage}[t]{0.49\textwidth}
\begin{center}
\includegraphics[height=1.2in]{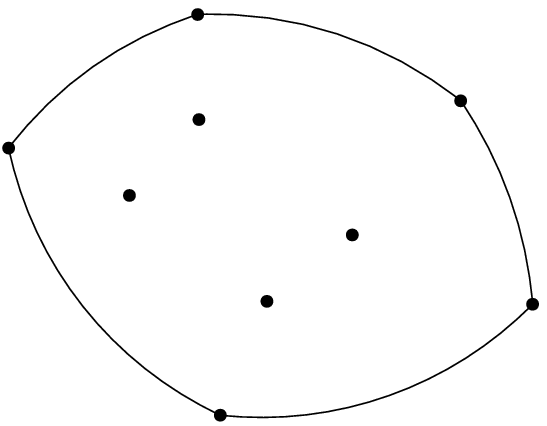}
\caption{\footnotesize Illustrating the circular hull of a set of points.}
\label{fig:circularhull}
\end{center}
\end{minipage}
\hspace{0.02in}
\begin{minipage}[t]{0.49\textwidth}
\begin{center}
\includegraphics[height=1.2in]{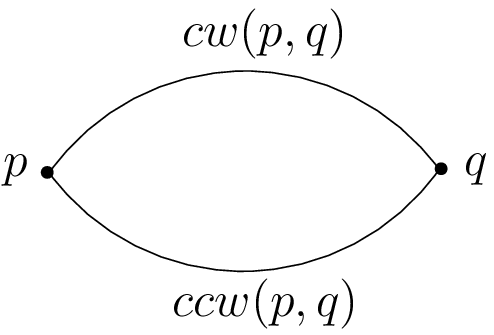}
\caption{\footnotesize Illustrating two minor arcs of $p$ and $q$.}
\label{fig:minorarc}
\end{center}
\end{minipage}
\vspace{-0.15in}
\end{figure}

Circular hulls will play a very important role in our algorithm.
As discussed in~\cite{ref:HershbergerFi91}, circular hulls have many
properties similar to convex hulls. However, circular hulls also have
special properties that convex hulls do not possess. For example, the
circular hull for a set of points may not exist. Also, the leftmost
point of a set $Q$ of points must be a vertex of the convex hull of $Q$, but this may not be the case for the circular hull. Due to these special properties, extending algorithms on convex hulls to circular hulls sometimes is not trivial, as will be seen later.
In the following, we introduce some concepts on circular hulls that will be needed later.

We assume that $r=1$ and thus a disk of radius $r$ is a {\em unit disk} (whose boundary is a {\em unit circle}). We use $\alpha(Q)$ to refer to $\alpha_r(Q)$. We assume that $\alpha(Q)$ exists.

For any arc of a circle, the circle is called the {\em supporting circle} of the arc, and the disk enclosed in the circle is called the {\em supporting disk} of the arc.
If a disk $D$ contains all points of a set $P$, then we say that $D$ {\em covers} $P$.
We say that a set $P$ of points in the plane is {\em unit disk coverable} if there is
a unit disk that contains all points of $P$, which is true if and only if $\alpha(P)$ exists.

Consider two points $p$ and $q$ that are unit disk coverable. There
must be a unit circle with $p$ and $q$ on it, and we call the arc of
the circle subtending an angle of at most $180^{\circ}$ a {\em minor
arc}~\cite{ref:HershbergerFi91}. Note that there are two minor arcs connecting $p$ and $q$, we
use $cw(p,q)$ to refer to the one clockwise around the center of the supporting circle of the arc
from $p$ to $q$, and use $ccw(p,q)$ to refer to the other one (e.g., see Fig.~\ref{fig:minorarc}).
Note that $cw(p,q)=ccw(q,p)$ and $ccw(p,q)=cw(q,p)$. For any minor arc
$w$, we use $D(w)$ to denote the supporting disk of $w$, i.e., the disk whose boundary contains $w$.
Note that all arcs of $\alpha(Q)$ are minor arcs.
We make a general position assumption that no point of $Q$ is on a minor arc of two other points of $Q$. The following observation has already been discovered previously~\cite{ref:EdelsbrunnerOn83,ref:HershbergerFi91}.

\begin{observation}{\em \cite{ref:EdelsbrunnerOn83,ref:HershbergerFi91}}\label{obser:basic}
\begin{enumerate}
\item
A point $p$ of $Q$ is a vertex of $\alpha(Q)$ iff there is a unit disk covering $Q$ and with $p$ on the boundary.
\item
A minor arc connecting two points of $Q$ is an arc of $\alpha(Q)$ iff its supporting disk covers $Q$.
\item
$\alpha(Q)$ is the common intersection of the supporting disks of all arcs of $\alpha(Q)$.
\item
A unit disk covers $Q$ iff it contains $\alpha(Q)$.
\item
For any subset $Q'$ of $Q$, $\alpha(Q')\subseteq \alpha(Q)$.
\end{enumerate}
\end{observation}

For any vertex $v$ of $\alpha(Q)$, we refer to the clockwise neighboring vertex of $v$ on $\alpha(Q)$ the {\em clockwise neighbor} of $v$, and the {\em counterclockwise} neighbor is defined analogously. We use $cw(v)$ and $ccw(v)$ to denote $v$'s clockwise and counterclockwise neighbors, respectively.
%

\paragraph{Tangents.}
Consider a vertex $v$ in the circular hull $\alpha(Q)$. Consider the arc $cw(ccw(v),v)$ of $\alpha(Q)$.
Let $D$ be the disk $D(cw(ccw(v),v))$. By Observation~\ref{obser:basic}(2) and (4), $D$ contains $\alpha(Q)$. Observe that if we rotate $D$ around $v$ clockwise until $\partial D$ contains the arc $cw(v,cw(v))$, $D$ always contains $\alpha(Q)$, and in fact, this continuum of disks $D$ are the only unit disks that contain $\alpha(Q)$ and have $v$ on the boundaries. For each of such disk $D$, we say that $D$ (and any part of $\partial D$ containing $v$) is {\em tangent} to $\alpha(Q)$ at $v$. We have the following observation.

\begin{observation}\label{obser:tangent}
A unit disk $D$ that contains a vertex $v$ of $\alpha(Q)$ on its boundary is tangent to $\alpha(Q)$ at $v$ if and only if $D$ contains both $cw(v)$ and $ccw(v)$.
\end{observation}


\begin{figure}[t]
\begin{minipage}[t]{0.47\textwidth}
\begin{center}
\includegraphics[height=1.2in]{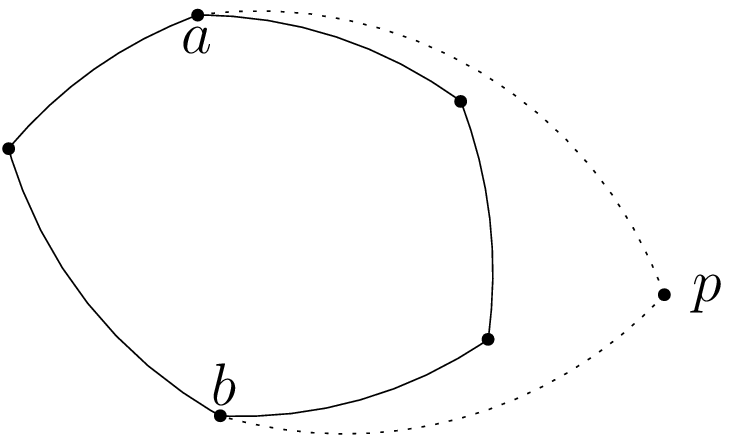}
\caption{\footnotesize Illustrating the two tangents from $p$ to $\alpha(Q)$: $cw(a,p)$ and $ccw(b,p)$.}
\label{fig:tangent}
\end{center}
\end{minipage}
\hspace{0.02in}
\begin{minipage}[t]{0.52\textwidth}
\begin{center}
\includegraphics[height=1.5in]{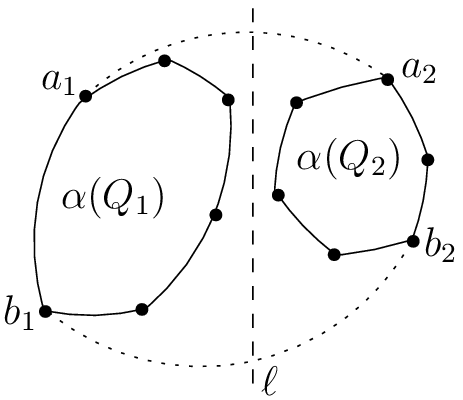}
\caption{\footnotesize Illustrating the upper common tangent $cw(a_1,a_2)$ and the lower common tangent $ccw(b_1,b_2)$ of $\alpha(Q_1)$ and $\alpha(Q_2)$.}
\label{fig:commontangent}
\end{center}
\end{minipage}
\vspace{-0.15in}
\end{figure}

Let $p$ be a point outside $\alpha(Q)$. If there is a vertex $a$ on $\alpha(Q)$ such that $D(cw(a,p))$ is tangent to $\alpha(Q)$ at $a$, then the arc $cw(a,p)$ is an {\em upper tangent} from $p$ to $\alpha(Q)$; e.g., see Fig~\ref{fig:tangent}.
If there is a vertex $b$ on $\alpha(Q)$ such that $D(ccw(b,p))$ is tangent to $\alpha(Q)$ at $b$, then the arc $ccw(b,p)$ is a {\em lower tangent} from $p$ to $\alpha(Q)$.
By replacing the arcs of $\alpha(Q)$ clockwise from $a$ to $b$ with the two tangents from $p$, we obtain $\alpha(Q\cup \{p\})$.
This also shows that $p$ has tangents to $\alpha(Q)$ if and only if $Q\cup\{p\}$ is unit disk coverable and $p$ is outside $\alpha(Q)$.
Note that $a=b$ is possible, in which case $\alpha(Q\cup \{p\})=\alpha(\{a,p\})$.

\paragraph{Common tangents of two circular hulls.}
Let $Q_1$ and $Q_2$ be two sets of points in the plane such that all points of $Q_1$ (resp., $Q_2$) are to the left (resp., right) of a vertical line $\ell$. Let $Q=Q_1\cup Q_2$.
A unit disk $D$ that is tangent to $\alpha(Q_1)$, say at a vertex $a$, and is also tangent to $\alpha(Q_2)$, say at a vertex $b$, is said to be {\em commonly tangent} to $\alpha(Q_1)$ and $\alpha(Q_2)$. The minor arc of $D$ connecting $a$ and $b$ is called a {\em common tangent} of the two circular hulls. It is an {\em upper} (resp, {\em lower}) tangent if it is clockwise (resp., counterclockwise) from $a$ to $b$ along the minor arc (e.g., see Fig.~\ref{fig:commontangent}).
The common tangents of $\alpha(Q_1)$ and $\alpha(Q_2)$ may not exist.
Indeed, if $\alpha(Q)$ does not exist, then the common tangents do not
exist. Otherwise the common tangents do not exist either if all points
of $Q_2$ are contained in $\alpha(Q_1)$, which happens only if $Q_2$
is covered by $D(w)$ for the rightmost arc $w$ of $\alpha(Q_1)$ and we
call it the {\em $Q_1$-dominating case}, or if all points of $Q_1$ are
contained in $\alpha(Q_2)$, which happens only if $Q_1$ is covered by $D(w')$ for the leftmost arc $w'$ of $\alpha(Q_2)$ and we call it the {\em $Q_2$-dominating case}.
If none of the above cases happens, then there are exactly two common tangents between the two hulls.
Each tangent intersects the vertical line $\ell$, which separates
$Q_1$ and $Q_2$, and the upper tangent intersects $\ell$ higher than the lower tangent does.

Suppose $\calL$ is a sequence of points and $p$ and $q$ are two points of $\calL$. We will adhere to the convention that a subsequence of $\calL$ {\em from $p$ to $q$} includes both $p$ and $q$, but a subsequence of $\calL$ {\em strictly} from $p$ to $q$ does not include either one. In many cases, $\calL$ is a cyclic sequence of points, e.g., vertices on a circular hull, and we often say points of $\calL$ clockwise/counterclockwise (strictly) from $p$ to $q$.

\section{The Decision Problem}
\label{sec:decision}

This section is concerned with the decision problem: Given a value $r$, decide whether $r^*\leq r$. Previously, Chan~\cite{ref:ChanMo99} solved the problem in $O(n\log n)$ time (Chan actually considered a slightly different problem: decide whether $r^*<r$, but the idea is similar). We present an $O(n)$ time algorithm, after $O(n\log n)$ time preprocessing to sort all points of $S^+$ and $S^-$ to obtain the sorted lists $p_1,\ldots, p_n$ and $q_1,\ldots, q_n$.

Given $r$, we use the following algorithmic framework in Algorithm~\ref{algo:case3} from~\cite{ref:ChanMo99} (see Theorem 3.3), which can decide whether $r^*\leq r$, and if yes, report all indices $i$ with $r_i^*\leq r$.

\begin{algorithm}[h]
	\caption{The decision algorithm of Chan~\cite{ref:ChanMo99}}
	\label{algo:case3}
	\SetAlgoNoLine
	$j\leftarrow -1$\;
	\For{$i\leftarrow 0$ \KwTo $n$}
	{
        \lWhile{$A[i,j+1]\leq  r$}
        {
         $j++$
        }
		\lIf{$B[i,j]\leq r$}
		{
			{\bf report} $i$
	    }
	}
\end{algorithm}


The algorithm is simple, but the technical crux is in how to decide
whether $A[i,j+1]\leq r$ and whether $B[i,j]\leq r$.
Chan~\cite{ref:ChanMo99} built a data structure in $O(n\log n)$ time
so that each of these two steps can be done in $O(\log n)$ time, which
leads to an overall $O(n\log n)$ time for his decision algorithm.
Our innovation is a new data structure that can perform each of the two
steps in $O(1)$ amortized time, resulting in an $O(n)$ time
algorithm.
Our idea is motivated by the following observation.

\begin{observation}
All such elements $A[i,j+1]$ that are checked in the algorithms (i.e., Line 3) are in a path of the matrix $A$ from $A[0,0]$ to an element in the bottom row and the path only goes rightwards or downwards. The same holds for the elements of $B$ that are checked in the algorithms (i.e., Line 4).
\end{observation}

We call such a path in $A$ as specified in the observation a {\em monotone path}, which contains at most $2(n+1)$ elements of $A$. We show that we can determine in $O(n)$ time whether $A[i,j]\leq r$ for all elements $A[i,j]$ in a monotone path of $A$. The same algorithm works for $B$ as well.

Let $\pi$ be a monotone path of $A$, starting from $A[0,0]$. Consider any element $A[i,j]$ on
$\pi$. Recall that $A[i,j]=\tau(L_{ij})$. The next value of $\pi$ after
$A[i,j]$ is either $A[i,j+1]$ or $A[i+1,j]$, i.e., either
$\tau(L_{i,j+1})$ or $\tau(L_{i+1,j})$. Note that $L_{i,j+1}$ can be
obtained from $L_{ij}$ by inserting $q_{j+1}$ and $L_{i+1,j}$ can be
obtained from $L_{ij}$ by deleting $p_{i+1}$. Because the points
$p_1,p_2,\ldots,p_n,q_1,q_2\ldots, q_n$ are ordered around $o$
counterclockwise, our problem becomes the following. Maintain a
sublist $Q$ of the above sorted list of $S$, with
$Q=S^+$ initially, to determine whether $\tau(Q)\leq
r$ (or equivalently whether $\alpha_r(Q)$ exists) under deletions and
insertions, such that a deletion operation deletes the first point
of $Q$ and an insertion operation inserts the point of $S$ following
the last point of $Q$. Further, deletions only happen to
points of $S^+$ (i.e., once $p_n$ is deleted from $Q$, no deletions
will happen). We refer to the problem as the {\em dynamic circular
hull problem}. We will show in Section~\ref{sec:circularhull} that the problem can be
solved in $O(n)$ time, i.e., each update takes $O(1)$ amortized time.
This leads to the following result.

\begin{theorem}\label{theo:decisionserial}
Assume that points of $S$ are sorted cyclically around $o$. Given
any $r$, whether $r^*\leq r$ can be decided in $O(n)$ time.
\end{theorem}


\paragraph{Remark.} For the nearby case, Chan proposed (in Theorem 3.4~\cite{ref:EppsteinFa97}) a randomized algorithm of $O(n\log n)$ time with high probability (i.e., $1-2^{-\Omega(n/\log^{12}n)}$) by using his $O(n\log n)$ time decision algorithm. Applying our linear time decision algorithm and following Chan's algorithm (specifically, setting $m$ to $\lfloor n/\log^7n \rfloor$ instead of $\lfloor n/\log^6n\rfloor$ in the algorithm of Theorem 3.4 in~\cite{ref:EppsteinFa97}), we can obtain the following result: After $O(n\log n)$ deterministic time preprocessing, we can compute $r^*$ for the nearby case in $O(n)$ time with high probability (i.e., $1-2^{-\Omega(n/\log^{14}n)}$).

\section{The Optimization Problem}
\label{sec:optimization}

With Theorem~\ref{theo:decisionserial}, we solve the optimization problem by parametric search~\cite{ref:ColeSl87,ref:MegiddoAp83}. As Chan's algorithm~\cite{ref:ChanMo99}, because our decision algorithm
is inherently sequential, we need to design a parallel decision algorithm. Chan~\cite{ref:ChanMo99} gave a parallel decision algorithm that runs in $O(\log n\log^2\log n)$ parallel steps using $O(n\log n)$ processors. Consequently, by using his $O(n\log n)$ time decision algorithm and applying Cole's parametric search~\cite{ref:ColeSl87}, Chan~\cite{ref:ChanMo99} solved the optimization problem in $O(n\log^2 n\log^2\log n)$ time. By following Chan's algorithmic scheme, we develop a new parallel decision algorithm that runs in $O(\log n\log \log n)$ parallel steps using $O(n)$ processors. Then, with the serial decision algorithm in Theorem~\ref{theo:decisionserial} and applying Cole's parametric search~\cite{ref:ColeSl87} on our new parallel decision algorithm,
we solve the optimization problem in $O(n\log n\log\log n)$ time.

Our algorithm relies on the following lemma, whose proof is quite independent of the remainder of this section and will be given in Section~\ref{sec:commontangent}. Note that Hershberger and Suri~\cite{ref:HershbergerFi91} gave a linear-time algorithm to achieve the same result as Lemma~\ref{lem:commontangent}, which suffices for their purpose.

\begin{lemma}\label{lem:commontangent}
Given the circular hull (with respect to a radius $r$) of a set $L$ of points and the circular hull of another set $R$ of points such that the points of $L$ and $R$ are separated by a line, one can do the following in $O(\log (|L|+|R|))$ time (assuming that the vertices of each circular hull are stored in a data structure that supports binary search):
determine whether the circular hull of $L\cup R$ (with respect to $r$) exists; if yes, either determine which dominating case happens (i.e., all points of a set are contained in the circular hull of the other set) or compute the two common tangents between the circular hulls of $L$ and $R$.
\end{lemma}

For any $0\leq i\leq j\leq n$, let $S^+[i,j]=\{p_{i},p_{i+1},\ldots,p_{j}\}$ and $S^-[i,j]=\{q_{i},q_{i+1},\ldots,q_{j}\}$.
By using Lemma~\ref{lem:commontangent}, we have the following lemma.

\begin{lemma}\label{lem:parallelquery}
We can preprocess $S$ and compute an interval $(r_1,r_2]$ containing
$r^*$ in $O(n\log n)$ time so that given any $r\in (r_1,r_2)$ and any pair $(i,j)$ with $1\leq i\leq j\leq n$,
we can determine whether $\alpha_r(S^+[i,j])$ (resp., $\alpha_r(S^-[i,j])$) exists, and if yes, return the root of a balanced binary search tree representing the circular hull, in
$O(\log k\log\log k)$ parallel steps using $O(\log k)$ processors, or in $O(\log^2k)$ time using one processor, where $k=j-i+1$.
\end{lemma}
\begin{proof}
As in~\cite{ref:ChanMo99,ref:EppsteinFa97}, we use the following geometric transformation.
For any point $p=(a,b)$, let $h(p)$ denote the halfspace $\{(x,y,z):
z\geq a^2+b^2-2ax-aby\}$. Then, for any set $P$ of points in the
plane, $(\tau(P))^2$ is the minimum of $x^2+y^2+z$ over all points $(x,y,z)$ in the polyhedron $\calP(P)=\bigcap_{p\in P} h(p)$.

\paragraph{Preprocessing.}
We build a complete binary search tree $T^+$ on the set
$S^+=\{p_1,p_2,\ldots,p_n\}$ such that the leaves of $T^+$ from left
to right storing the
points of $S^+$ in their index order. Each internal node $v$ of
$T^+$ stores a hierarchical representation~\cite{ref:DobkinDe90}
of the polyhedron $\calP(P)$, where $P$ is the set of points stored in
the leaves of the subtree rooted at $v$
($P$ is called a {\em canonical subset}). Computing the polyhedrons of
all internal nodes of $T^+$ can be done in $O(n\log n)$ time in a
bottom-up manner using linear time polyhedra intersection
algorithms~\cite{ref:ChanA16,ref:ChazelleAn92}.
Similarly, we build a tree $T^-$ on the set $S^-=\{q_1,q_2,\ldots,q_n\}$.

Consider a vertex $v=(x,y,z)$ of $\calP(P)$ for a canonical subset $P$
of $T^+$. Define $r(v)=\sqrt{x^2+y^2+z}$.
Let $C$ be the set of the values $r(v)$ of all vertices $v$ of $\calP(P)$ for all canonical subsets $P$ of $T^+$.
Note that $|C|=O(n\log n)$. We find the smallest value $r(v)\in C$ such that $r^*\leq r(v)$, and let $r_2$ denote such $r(v)$. The value $r_2$ can be found in $O(n\log n)$ using our linear time decision algorithm and doing binary search on $C$ using the linear time selection algorithm~\cite{ref:BlumTi73}.
Next, we find the largest value in $C$ that is smaller than $r_2$, and let $r_1$ denote that value. By definition, $r^*\in (r_1,r_2]$ and $(r_1,r_2)$ does not contain any element of $C$.

Consider a canonical subset $P$ of $T^+$ and any $r\in (r_1,r_2)$. We construct $\calI_r(P)$ for each canonical subset $P$ of $T^+$ by intersecting the facets of $\calP(P)$ with the paraboloid $W(r)=\{(x,y,z): x^2+y^2+z=r^2\}$ and projecting them vertically to the $xy$-plane.
By the definitions of $r_1$ and $r_2$, the paraboloid $W(r)$
intersects the same set of edges of $\calP(P)$ for all $r\in
(r_1,r_2)$; this implies that $\calI_r(P)$ is combinatorially the same
for all $r\in (r_1,r_2)$. Hence, we can consider $\alpha_r(P)$, which
	is the dual of $\calI_r(P)$, as a parameterized circular hull of $P$. We
	store the (parameterized) vertices of $\alpha_r(P)$ in a balanced
	binary search tree. Since $\calP(P)$ is convex, we can obtain
	$\calI_r(P)$ and thus the balanced binary search tree for
	$\alpha_r(P)$ in $O(|P|)$ time; we associate the tree at the node
	of $T^+$ for $P$.  Because the total size of $\calP(P)$ for all
	canonical subsets $P$ in $T^+$ is $O(n\log n)$, we can obtain the
	balanced binary search trees for $\alpha_r(P)$ of all canonical
	subsets $P$ in $T^+$ in $O(n\log n)$ time.

We do the same for $T^-$ as above. The processing on $T^-$
will obtain two values $r_1'$ and $r_2'$ correspondingly as the above
$r_1$ and $r_2$. We update $r_1=\max\{r_1,r_1'\}$ and
$r_2=\min\{r_2,r_2'\}$; so $r^*\in (r_1,r_2]$ still holds.
This finishes our processing on $S$, which takes $O(n\log n)$ time and is independent of $r$.


\paragraph{Queries.}
Given any $r\in (r_1,r_2)$ and any pair $(i,j)$ with $i<j$, we
determine whether $\alpha_r(S^+[i,j])$ exists, and if yes, return the root of a balanced binary search tree representing it, as follows (the case for $S^-[i,j]$ is similar). Let $k=j-i+1$ and let $P=S^+[i,j]$.

By the standard method, we first find $O(\log k)$ canonical subsets of
$T^+$ whose union is exactly $S^+[i,j]$.
Our following computation procedure can be described as a complete binary tree $T$ where the
leaves corresponding to the above $O(\log k)$ canonical subsets.
So $T$ has $O(\log k)$ leaves, and its height is $O(\log\log k)$.
For each leave of $T$, its circular hull is already available due to the
preprocessing. For each internal node $v$ that is the parent of two
leaves, we compute the circular hull of the union of the two subsets
$P_1$ and $P_2$ of the two leaves. As the points of $S^+$ are ordered radially by $o$,
the two subsets are separated by a line through $o$. Hence, we can
find the common tangents (if exist) using
Lemma~\ref{lem:commontangent} in $O(\log k)$ time because the size of each subset is no more than $k$.
Recall that the circular hull of each canonical subset is represented
by a balanced binary search tree.
After having the common tangents, we split and merge the two balanced binary
search trees to obtain a balanced binary search tree for
$\alpha_r(P_1\cup P_2)$. In addition, we keep unaltered the two
original trees for $\alpha_r(P_1)$ and $\alpha_r(P_2)$ respectively, and this can be done by using persistent data structures (e.g., using the copy-path technique~\cite{ref:DriscollMa89,ref:SarnakPl86}) in $O(\log k)$ time. In this way, the original trees for $\alpha_r(P_1)$ and $\alpha_r(P_2)$ can be used in parallel for other computations. If the algorithm detects that $\alpha_r(P_1\cup P_2)$ does not
exist, then we simply halt the algorithm and report that $\alpha_r(S^+[i,j])$ does not exist.
Also, if the algorithm finds that a dominating case happens, e.g., the
$P_1$-dominating case, then $\alpha_r(P_1\cup P_2)=\alpha_r(P_1)$ and
thus we simply return the root of the tree for $\alpha_r(P_1)$.

We do this for all internal nodes in the second level of $T$ (i.e., the
level above the leaves) in parallel by assigning a processor for each node. In this way, as $T$ has $O(\log k)$ leaves, we can compute the circular hulls for the second level in $O(\log k)$ parallel steps using
$O(\log k)$ processors. Then, we proceed on the third level in the
same way. At the root of $T$, we will have the root of a balanced
binary search tree for $\alpha_r(P)$.
Using $O(\log k)$ processors, this takes $O(\log k\log\log k)$ parallel steps
because each level needs $O(\log k)$ parallel steps and the height of $T$
is $O(\log\log k)$.

Alternatively, if we only use one processor, then since $T$ has $O(\log k)$ nodes and we spend $O(\log k)$ time on each node, the total time is $O(\log^2 k)$.
\qed
\end{proof}

Armed with Lemma~\ref{lem:parallelquery}, to determine whether $r^*\leq r$, we use the algorithm framework in Theorem~4.2 of Chan~\cite{ref:ChanMo99}, but we provide a more efficient implementation, as follows.

Recall the definitions of the matrices $A$ and $B$ in Section~\ref{sec:pre}, and in particular, each row of $A$ (resp., $B$) is monotonically increasing while each column of $A$ (resp., $B$) is monotonically decreasing.
For convenience, let $A[i,-1]=0$ and $A[i,n+1]=B[i,-1]=\infty$ for all $0\leq i\leq n$.
Let $m=\lfloor n/\log^6n\rfloor$. Let $j_t=t\cdot \lfloor n/m \rfloor$ for $t=1,2,\ldots,m-1$. Set $j_0=-1$ and $j_{m}=n$. For each $t\in [0,m]$, find the largest $i_t\in [0,n]$ with $A[i_t,j_t]\geq B[i_t,j_t]$ (set $i_t=-1$ if no such index exists; note that $i_0=-1$). Observe that $i_0\leq i_1\leq \cdots\leq i_m$. Each $i_t$ can be found in $O(\log^7 n)$ time by binary search using Lemma~\ref{lem:querymatrix}. Hence, computing all $i_t$'s takes $O(n\log n)$ time. This is part of our preprocessing, independent of $r$.

\begin{lemma}{\em \cite{ref:ChanMo99,ref:EppsteinFa97}}\label{lem:querymatrix}
After $O(n\log n)$ time preprocessing, $A[i,j]$ and $B[i,j]$ can be computed in $O(\log^6 n)$ time for any given pair $(i,j)$.
\end{lemma}

Given $r>0$, our goal is to decide whether $r^*\leq r$. Let $(r_1,r_2]$ be the interval obtained by the preprocessing of Lemma~\ref{lem:parallelquery}. Since $r^*\in (r_1,r_2]$, if $r\leq r_1$, then $r^*>r$; if $r\geq r_2$, then $r^*\leq r$. It remains to resolve the case $r\in (r_1,r_2)$, as follows. In this case the result of Lemma~\ref{lem:parallelquery} applies.

We will decide whether $r_i^*\leq r$ for all $i=0,1,\ldots,n$ (recall that $r^*\leq r$ iff some $r_i^*\leq r$), as follows. Let $t\in [0,m-1]$ such that $i_t<i\leq i_{t+1}$. If $A[i,j_t]>r$,
then return $r_i^*>r$. Otherwise, find (by binary search) the largest
$j\in [j_t,j_{t+1}]$ with $A[i,j]\leq r$, and return  $r_i^*\leq
r$ if and only if $B[i,j]\leq r$. Algorithm~\ref{algo:parallel} gives
the pseudocode. See Theorem 4.2 of~\cite{ref:ChanMo99} for the algorithm correctness.

\begin{algorithm}[h]\label{algo:parallel}
	\caption{The decision algorithm of Theorem 4.2 by Chan~\cite{ref:ChanMo99}}
	\label{algo:decision}
	\SetAlgoNoLine
	Let $t\in [0,m-1]$ such that $i_t<i\leq i_{t+1}$\;
    \lIf{$A[i,j_t]> r$}
		{
			{ return} $r_i^*>r$
	    }
	find the largest $j\in [j_t,j_{t+1}]$ with $A[i,j]\leq r$\;
    return $r_i^*\leq r$ iff $B[i,j]\leq r$\;
\end{algorithm}

Chan~\cite{ref:ChanMo99} implemented the algorithm in $O(\log n\log^2\log n)$ parallel steps using $O(n\log n)$ processors. In what follows, with the help of Lemma~\ref{lem:parallelquery}, we provide a more efficient implementation of $O(\log n\log\log n)$ parallel steps using $O(n)$ processors.
Line 1 can be done in $O(n)$ time as part of the preprocessing, independent of $r$.
We first discuss how to implement Line 3 for all indices $i$, and we will show later that
Lines 2 and 4 can be implemented in a similar (and faster) way.

For each $t=0,1,\ldots,m-1$, if $i_{t+1}-i_t\leq \log^6 n$, then we
form a group of at most $\log^6n$ indices:
$i_t+1,i_t+2,\ldots,i_{t+1}$. Otherwise, starting from $i_t+1$, we
form a group for every consecutive $\log^6 n$ indices until
$i_{t+1}$, so every group has exactly $\log^6 n$ indices except that
the last group may have less than $\log^6 n$ indices. In this way, we
have at most $2m$ groups, each of which consists of at most $\log^6 n$
consecutive indices in $(i_t,i_{t+1}]$ for some $t\in [0,m-1]$.

Consider a group $G=\{a,a+1,\ldots,a+b\}$ of indices in $(i_t,i_{t+1}]$. Note that $b<\log^6 n$.
For each $i\in [a,a+b]$ such that $A[i,j_t]\leq r$, we need to perform binary search on $[j_t,j_{t+1}]$ to find the largest index $j$ with $A[i,j]\leq r$. To this end, we do the following.
We compute the two circular hulls $\alpha(S^+[a+b,n])$ and
$\alpha(S^-[1,j_t])$, in $O(\log n\log\log n)$ parallel steps using $O(\log n)$ processors by Lemma~\ref{lem:parallelquery}.
Note that by ``compute the two circular hulls'', we mean that
the two circular hulls are computed implicitly in the sense that each of them is represented by a balanced binary search tree and we have the access of its root. If $\alpha(S^+[a+b,n])$ (resp., $\alpha(S^-[1,j_t])$) does not exist, then we set it to $null$.
We do this for all $2m$ groups in parallel, which takes $O(\log
n\log\log n)$ parallel steps using $O(m\log n)\in O(n)$ processors.

Consider the group $G$ defined above again. For each $i\in [a,a+b]$, we need
to do binary search on $[j_t,j_{t+1}]$ for
$O(\log(j_{t+1}-j_t))=O(\log\log n)$ iterations. In each iteration,
the goal is to determine whether $A[i,j]\leq r$ for an index $j\in
[j_t,j_{t+1}]$. To this end, it suffices to determine whether
$\alpha(U_{ij})$ exists. Notice that $U_{ij}=S^+[i+1,a+b-1]\cup S^+[a+b,n]\cup S^-[1,j_t] \cup S^-[j_t+1,j]$.
$\alpha(S^+[a+b,n])$ and $\alpha(S^-[1,j_t])$ are already
computed above. If one of them does not exist, then $\alpha(U_{ij})$ does not
exist and thus $A[i,j]>r$. Otherwise, we compute the circular hull
$\alpha(S^+[i+1,a+b-1])$, which can be done in $O(\log^2\log n)$ time
using one processor by Lemma~\ref{lem:parallelquery} because
$a+b-1-i\leq b-1 \leq  \log^6 n$. We also compute $\alpha(S^-[j_t+1,j])$ in
$O(\log^2\log n)$ time using one processor. Then, we compute the
common tangents of $\alpha(S^+[i+1,a+b-1])$ and $\alpha(S^+[a+b,n])$
by Lemma~\ref{lem:commontangent} (note that $S^+[i+1,a+b-1]$ and
$S^+[a+b,n]$ are separated by a line through $o$), in $O(\log n)$ time
using one processor. Then, we merge the two hulls with the two common
tangents to obtain a balanced binary search tree for
$\alpha(S^+[i+1,n])$. Because we want to keep the tree for
$\alpha(S^+[a+b,n])$ unaltered so that it can participate in other
computations in parallel, we use a persistent tree to represent it. Similarly,
we obtain a tree for $\alpha(S^-[1,j])$, in $O(\log n)$ time using one
processor.
If one of $\alpha(S^+[i+1,n])$ and $\alpha(S^-[1,j])$ does not exist, then we return $A[i,j]>r$.
Note that $S^+[a+b,n]$ and $S^-[1,j]$ are separated by $\ell$ and $U_{ij}=S^+[a+b,n]\cup S^-[1,j]$. By applying Lemma~\ref{lem:commontangent}, we can determine whether $\alpha(U_{ij})$ exists in $O(\log n)$ time using one processor and consequently determine whether $A[i,j]\leq r$.
Therefore, the above algorithm determines whether $A[i,j]\leq r$ in $O(\log n)$ time using one processor.

If we do the above for all $i$'s in parallel, then we can determine
whether  $A[i,j]\leq r$ in $O(\log n)$ time using $n+1$ processors, for
each iteration of the binary search. As there are $O(\log\log n)$
iterations, the binary search procedure (i.e., Line 3) for all
$i=0,1,\ldots,n$ runs in $O(\log n\log\log n)$ parallel steps using $n+1$ processors.

For implementing Line 2, we can use the same approach as above by grouping the indices $i$ into $2m$ groups. The difference is that now each $i$ has a specific index $j$, i.e., $j=j_t$, for deciding whether $A[i,j]\leq r$, and thus we do not have to do binary search. Hence, using $n+1$ processors, we can implement Line~2 for all $i=0,1,\ldots, n$ in $O(\log n)$ parallel steps. We can do the same for Line 4.

As a summary, we have the following theorem.

\begin{theorem}\label{theo:paralleldecision}
After $O(n\log n)$ time preprocessing on $S$, given any $r$, we can decide whether $r^*\leq r$ in $O(\log n\log\log n)$ parallel steps using $O(n)$ processors.
\end{theorem}

With the serial decision algorithm in Theorem~\ref{theo:decisionserial} and applying Cole's parametric search~\cite{ref:ColeSl87} on the parallel decision algorithm in Theorem~\ref{theo:paralleldecision}, the following result follows.

\begin{theorem}\label{theo:optimization}
The value $r^*$ can be computed in $O(n\log n\log\log n)$ time.
\end{theorem}
\begin{proof}
Suppose there is a serial decision algorithm of time $T_S$ and another parallel decision algorithm that runs in $T_p$ parallel steps using $P$ processors. Then, Megiddo's parametric search~\cite{ref:MegiddoAp83} can compute $r^*$ in $O(PT_p+T_sT_p\log P)$ time by simulating the parallel decision algorithm on $r^*$ and using the serial decision algorithm to resolve comparisons with $r^*$. If the parallel decision algorithm has a ``bounded fan-in or bounded fan-out'' property, then Cole's technique~\cite{ref:ColeSl87} can reduce the time complexity to $O(PT_p+T_s(T_p+\log P))$.
Like Chan's algorithm~\cite{ref:ChanMo99}, our algorithm has this property because it mainly consists of $O(\log\log n)$ rounds of independent binary search (i.e., the algorithm of Lemma~\ref{lem:commontangent}). In our case, $T_s=O(n)$, $T_p=O(\log n\log\log n)$, and $P=O(n)$. Thus, applying Cole's technique, $r^*$ can be computed in $O(n\log n\log\log n)$ time.
\qed
\end{proof}


Note that once $r^*$ is computed, we can apply the seriel decision algorithm to obtain in additional $O(n)$ time a pair of congruent disks of radius $r^*$ covering $S$.

\begin{corollary}
The planar two-center problem can be solved in $O(n\log^2 n)$ time.
\end{corollary}
\begin{proof}
This follows by combining Theorem~\ref{theo:optimization}, which is for the nearby case, with the $O(n\log^2 n)$ time algorithm for the distant case~\cite{ref:EppsteinFa97}. \qed
\end{proof}

\section{The Convex Position Case}
\label{sec:convexposition}

In this section, we consider the case where $S$ is in convex position (i.e., every point of $S$ is a vertex of the convex hull of $S$). We show that our above $O(n\log n\log\log n)$ time algorithm can be applied to solving this case in the same time asymptotically.

We first compute the convex hull  $\CH(S)$ of $S$ and order all vertices clockwise as $p_1,p_2,\ldots, p_n$. A key observation~\cite{ref:KimEf00} is that there is an optimal solution with two congruent disks $D_1^*$ and $D_2^*$ of radius $r^*$ such that $D_1^*$ covers the points of $S$ in a chain of $\partial \CH(S)$ and $D_2^*$ covers the rest of the points. In other words, the cyclic list of $p_1,p_2,\ldots, p_n$ can be cut into two lists such that one list is covered by $D_1^*$ and the other list is covered by $D_2^*$.

Let $o$ be any point in the interior of $\CH(S)$. By the above observation, there exists a pair of rays $\rho_1$ and $\rho_2$ emanating from $o$ such that $D_1^*$ covers all points of $S$ on one side of the two rays and $D_2^*$ covers the points of $S$ in the other side. In order to apply our previous algorithm, we need to find a line $\ell$ that separates the two rays. For this, we propose the following approach.

For any $i,j\in [1,n]$, let $S_{cw}[i,j]$ denote the subset of
vertices on $\CH(S)$ clockwise from $p_i$ to $p_j$, and
$S_{cw}[i,j]=\{p_i\}$ if $i=j$.
Due to the above key observation, $r^*=\min_{i,j\in [1,n]}\max\{\tau(S_{cw}[i,j]),\tau(S_{cw}[j+1,i-1])\}$, with indices modulo $n$.
For each $i\in [1,n]$, define $r(i)=\min_{h\in [i,i+n-1]}\max\{\tau(S_{cw}[i,h]),\tau(S_{cw}[h+1,i-1])\}$.
Notice that as $h$ increases in $[1,n-1]$, $\tau(S_{cw}[1,h])$ is monotonically increasing while $\tau(S_{cw}[h+1,n])$ is monotonically decreasing.
Define $k$ to be the largest index in $[1,n-1]$ such that $\tau(S_{cw}[1,k])\leq \tau(S_{cw}[k+1,n])$.
We have the following lemma.


\begin{lemma}\label{lem:partition}
$r^*$ is equal to the minimum of the following four values: $r(1)$, $r(k+1)$, $r(k+2)$, and $\max\{\tau(S_{cw}[i,j]),\tau(S_{cw}[j+1,i-1])$ for all indices $i$ and $j$ with $i\in [1,k]$ and $j\in [k+2,n]$.
\end{lemma}
\begin{proof}
Observe that $r^*=\min_{i,j\in [1,n]}\max\{\tau(S_{cw}[i,j]),\tau(S_{cw}[j+1,i-1])\}=\min_{1\leq h\leq n}r(h)$. Hence, $r^*$ is no larger than any of the values specified in the lemma statement.

Let $i$ and $j$ be two indices such that $r^*=\max\{\tau(S_{cw}[i,j]),\tau(S_{cw}[j+1,i-1])\}$ with $1\leq i\leq j\leq n$. We claim that $r^*=r(i)$. Indeed, since $r^*=\min_{1\leq h\leq n}r(h)$, we have $r^*\leq r(i)$. On the other hand, as $r(i)\leq \max\{\tau(S_{cw}[i,j]),\tau(S_{cw}[j+1,i-1])\}=r^*$, we obtain $r(i)= r^*$. By a similar argument, $r^*=r(j+1)$ also holds.

Without loss of generality, we assume that $r^*=\tau(S_{cw}[i,j])\geq \tau(S_{cw}[j+1,i-1])$.

If $i\in [1,k]$ and $j\in [k+2,n]$, then the lemma follows. Otherwise, one of the following four cases must hold: $i=k+1$, $j=k+1$, $[i,j]\subseteq [1,k]$, and $[i,j]\subseteq[k+2,n]$.
If $i=k+1$, then $r^*=r(k+1)$. If $j=k+1$, then $r^*=r(k+2)$.
Below we show that $r^*=r(1)$ if $[i,j]\subseteq [1,k]$ and we also show that the case $[i,j]\subseteq[k+2,n]$ cannot happen, which will prove the lemma.


If $[i,j]\subseteq [1,k]$, then $\tau(S_{cw}[j+1,i-1])\geq \tau(S_{cw}[k+1,n])$, for $S_{cw}[k+1,n]\subseteq S_{cw}[j+1,i-1]$. By the definition of $k$, we have $\tau(S_{cw}[k+1,n])\geq \tau(S_{cw}[1,k])$. Because $[i,j]\subseteq [1,k]$, $\tau(S_{cw}[1,k])\geq \tau(S_{cw}[i,j])$. Combining the above three inequalities leads to the following: $\tau(S_{cw}[j+1,i-1])\geq \tau(S_{cw}[k+1,n])\geq \tau(S_{cw}[1,k])\geq \tau(S_{cw}[i,j])$. Because $r^*=\tau(S_{cw}[i,j])\geq \tau(S_{cw}[j+1,i-1])$, we obtain $r^*=\tau(S_{cw}[j+1,i-1])= \tau(S_{cw}[k+1,n])= \tau(S_{cw}[1,k])= \tau(S_{cw}[i,j])$. Notice that $r(1)\leq \max\{\tau(S_{cw}[1,k]),\tau(S_{cw}[k+1,n])\}$. Thus, we derive $r(1)\leq r^*$. Since $r^*\leq r(1)$, we finally have $r^*=r(1)$.

If $[i,j]\subseteq [k+2,n]$, then $\tau(S_{cw}[j+1,i-1])\geq \tau(S_{cw}[1,k+1])$. By the definition of $k$, we have $\tau(S_{cw}[1,k+1])> \tau(S_{cw}[k+2,n])$. Also, since $[i,j]\subseteq [k+2,n]$, $\tau(S_{cw}[k+2,n])\geq \tau(S_{cw}[i,j])$ holds. Therefore, we obtain $\tau(S_{cw}[j+1,i-1])\geq \tau(S_{cw}[1,k+1])> \tau(S_{cw}[k+2,n])\geq \tau(S_{cw}[i,j])$, which incurs contradiction since $r^*=\tau(S_{cw}[i,j])\geq \tau(S_{cw}[j+1,i-1])$. Thus, the case $[i,j]\subseteq [k+2,n]$ cannot happen.
\qed
\end{proof}

Based on the above lemma, our algorithm works as follows.

We first compute $r(1)$ and the index $k$. This can be easily done in $O(n\log n)$ time. Indeed, as $h$ increases in $[1,n-1]$, $\tau(S_{cw}[1,h])$ is monotonically increasing while $\tau(S_{cw}[h+1,n])$ is monotonically decreasing. Therefore, $r^*_1$ and $k$ can be found by binary search on $[1,n-1]$. As both $\tau(S_{cw}[1,h])$ and $\tau(S_{cw}[h+1,n])$ can be computed in $O(n)$ time, the binary search takes $O(n\log n)$ time. For the same reason, we can compute $r(k+1)$ and $r(k+2)$ in $O(n\log n)$ time.

If $r^*\not\in \{r(1),r(k+1),r(k+2)\}$, then $r^*=\max\{\tau(S_{cw}[i,j]),\tau(S_{cw}[j+1,i-1])$ for two indices $i$ and $j$ with $i\in [1,k]$ and $j\in [k+2,n]$. We can compute it as follows.
Let $\ell$ be a line through $v_{k+1}$ and intersecting the interior of $\overline{p_np_1}$.
Let $o$ be any point on $\ell$ in the interior of $\CH(S)$.
Lemma~\ref{lem:partition} implies $\ell$ and $o$ satisfy the property discussed above, i.e., $\ell$ separates the two rays $\rho_1$ and $\rho_2$. Consequently, we can apply our algorithm for Theorem~\ref{theo:optimization} to compute $r^*$ in $O(n\log n\log\log n)$ time.

\begin{theorem}
The planar two-center problem for a set of $n$ points in convex position can be solved in $O(n\log n\log\log n)$ time.
\end{theorem}

\paragraph{Remark.} The randomized result remarked after Theorem~\ref{theo:decisionserial} also applies to the convex position case, i.e., after $O(n\log n)$ deterministic time preprocessing, we can compute $r^*$ in $O(n)$ time with high probability (i.e., $1-2^{-\Omega(n/\log^{14}n)}$).

\section{The Dynamic Circular Hull Problem}
\label{sec:circularhull}


In this section, we give an $O(n)$ time algorithm for the dynamic circular hull problem needed in our decision algorithm in Section~\ref{sec:decision}.

Recall that the points of $S$ are ordered around $o$ cyclically. To simplify the
exposition, we first work on a slightly different problem setting in which points are sorted by their $x$-coordinates; we will show later that the algorithm can be easily adapted to the original problem setting.

Specifically, let $L=\{p_1,p_2,\ldots,p_n\}$  be a set of $n$ points sorted
from left to right and $R=\{q_1,q_2,\ldots,q_n\}$  be a set of $n$ points sorted
from left to right, such that all points of $L$ are strictly to the
left of a vertical line $\ell$ and all points of $R$ are strictly to
the right of $\ell$.  The problem is to maintain a
sublist $Q$ of the sorted list of $L\cup R$, with
$Q=L$ initially, to determine whether $\alpha_r(Q)$ exists under deletions and
insertions, such that a deletion operation deletes the leftmost point
of $Q$ and an insertion operation inserts the point of $R$ after
the rightmost point of $Q$. Further, deletion operations only happen to
points of $L$. In the following, we build a data
structure in $O(n)$ time that can handle each update in $O(1)$
amortized time (i.e., after each update, we know whether
$\alpha_r(Q)$ exists).
We make a general position assumption that no two points of $L\cup R$ have the same $x$-coordinate.

Since initially $Q=L$, we need to compute $\alpha_r(Q)$. Hershberger and Suri~\cite{ref:HershbergerFi91} gave an $O(n\log n)$ time algorithm using divide-and-conquer.
The algorithm of Edelsbrunner et al.~\cite{ref:EdelsbrunnerOn83} can
also compute $\alpha_r(Q)$ in $O(n\log n)$ time by first computing the
farthest Delaunay triangulation of $Q$. Both algorithms still take
$\Theta(n\log n)$ time even if points of $Q$ are sorted (indeed,
the algorithm of \cite{ref:HershbergerFi91} spends $O(n)$ time for
each combine/merge step and the algorithm of
\cite{ref:EdelsbrunnerOn83} needs to compute the farthest Delaunay
triangulation first). We instead exhibit an $O(n)$ time incremental algorithm, which can be considered an extension of Graham's scan for convex hulls, although the extension is not straightforward at all. Before we are able to describe the algorithm, we need to discuss some properties of the circular hulls.

The remainder of this section is organized as follows. In Section~\ref{sec:obser}, we show some properties of circular hulls that will be useful for our algorithm. In Section~\ref{sec:static}, we present our linear-time algorithm for constructing the circular hull for a set of sorted points. In Section~\ref{sec:datastructure}, we elaborate on our data structure for maintaining $\alpha_r(Q)$ for a dynamic set $Q$. Section~\ref{sec:initial} sets up the data structure initially when $Q=L$ (e.g., invokes the algorithm given in Section~\ref{sec:static}). Our algorithms for processing deletions and insertions will be described in Sections~\ref{sec:delete} and \ref{sec:insert}, respectively. Finally in Section~\ref{sec:radial} we adapt the algorithm to our original problem setting where points are sorted radially around the origin $o$.

\subsection{Observations and Properties of Circular Hulls}
\label{sec:obser}

From now on, we assume $r=1$ and thus a disk of radius $r$ is a {\em unit disk} (whose boundary is a {\em unit circle}). We use $\alpha(Q)$ to refer to $\alpha_r(Q)$.
We assume that $Q$ is a subset of $L\cup R$ and $\alpha(Q)$ exists.

Recall that every arc of $\alpha(Q)$ is a minor arc. In the following, unless otherwise stated, an arc refers to a minor arc and a disk refers to a unit disk. For ease of exposition, we make a general position assumption that no point of $L\cup R$ is on a minor arc of two other points of $L\cup R$.

We define the {\em upper hull} of $\alpha(Q)$ as the boundary of $\alpha(Q)$ from the
leftmost vertex to the rightmost vertex. The remaining arcs of $\alpha(Q)$ comprise
the {\em lower hull}. Unlike convex hulls, the upper hull (resp., the
lower hull) of $\alpha(Q)$ may not be $x$-monotone due to that the
leftmost/rightmost arc may not be $x$-monotone.
If the rightmost point $p$ of $\alpha(Q)$ is in the interior of an
arc, then we refer to the arc as the {\em rightmost arc} of $\alpha(Q)$;
otherwise, the rightmost arc is $null$ (and its supporting disk is defined to be $\emptyset$).
We define the {\em leftmost arc} of $\alpha(Q)$ likewise.


\begin{figure}[t]
\begin{minipage}[t]{0.49\textwidth}
\begin{center}
\includegraphics[height=1.5in]{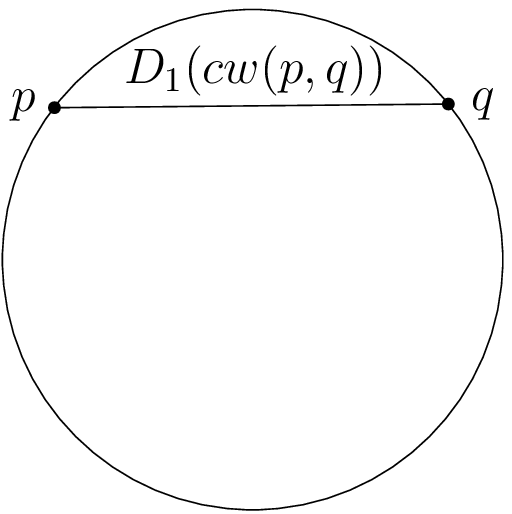}
\caption{\footnotesize Illustrating $D_1(cw(p,q))$.}
\label{fig:chord}
\end{center}
\end{minipage}
\hspace{0.02in}
\begin{minipage}[t]{0.49\textwidth}
\begin{center}
\includegraphics[height=1.5in]{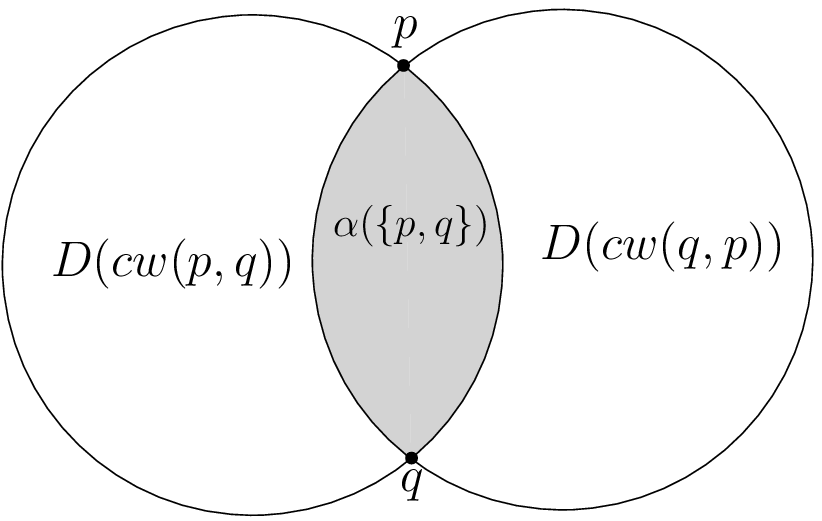}
\caption{\footnotesize Illustrating $\alpha(\{p,q\})$.}
\label{fig:twopointhull}
\end{center}
\end{minipage}
\end{figure}

For a minor arc $w$, recall that $D(w)$ is the supporting disk of $w$.
We further use $D_1(w)$ to denote the region of $D(w)$ bounded by $w$ and the chord of $D(w)$ connecting the two endpoints of $w$ (e.g., see Fig.~\ref{fig:chord}).
Observe that $\alpha(\{p,q\})=D_1(cw(p,q))\cup D_1(ccw(p,q))=D(cw(p,q))\cap D(ccw(p,q))$; e.g., see Fig.~\ref{fig:twopointhull}.
For notational simplicity, we use $\alpha(p,q)$ to refer to $\alpha(\{p,q\})$.
The following observation, which is due to the convexity of the circular hull, was already shown in~\cite{ref:HershbergerFi91}.

\begin{observation}\label{obser:red}{\em \cite{ref:HershbergerFi91}}
Suppose $p$ is a point to the right (resp., left) of all points of $Q$ and $\alpha(\{p\} \cup Q)$ exists. Then, $p$ is not a vertex of $\alpha(\{p\} \cup Q)$ if and only if $p$ is in $D_1(w)$, where $w$ is the rightmost (reps., leftmost) arc of $\alpha(Q)$. We say that $p$ is {\em redundant} (with respect to $\alpha(Q)$) if $p\in D_1(w)$.
\end{observation}


Recall that in Graham's scan for computing convex hulls, the algorithm uses ``left turn'' and ``right turn''. Here instead we find it more informative to use {\em inner turn} and {\em outer turn}, defined as follows. Note that these concepts are new.
Suppose two points $p$ and $q$ are unit disk coverable.
Consider the minor arc $cw(p,q)$, and a point $t$. We say that $cw(p,q)$ and $t$ form an {\em inner turn} if
$t\in D(cw(p,q))$ and {\em outer turn} otherwise.
The following observation will help prove the correctness of our algorithm.


\begin{figure}[t]
\begin{minipage}[t]{0.49\textwidth}
\begin{center}
\includegraphics[height=1.5in]{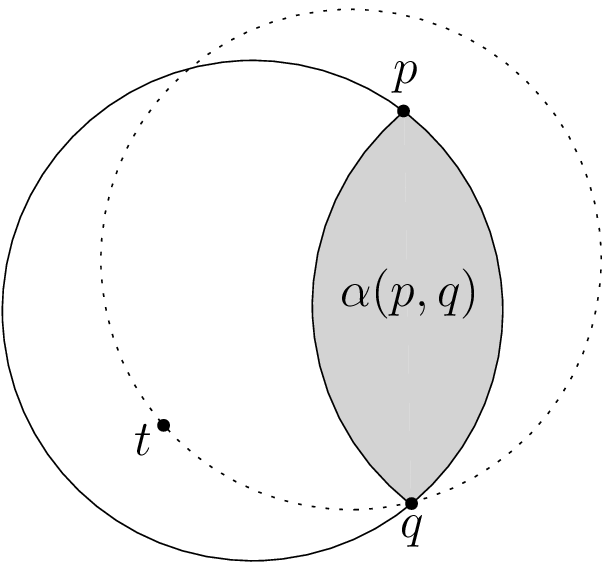}
\caption{\footnotesize Illustrating Observation~\ref{obser:turn}(1). The dotted circle depicts $D(cw(q,t))$.}
\label{fig:innerturn}
\end{center}
\end{minipage}
\hspace{0.02in}
\begin{minipage}[t]{0.49\textwidth}
\begin{center}
\includegraphics[height=1.5in]{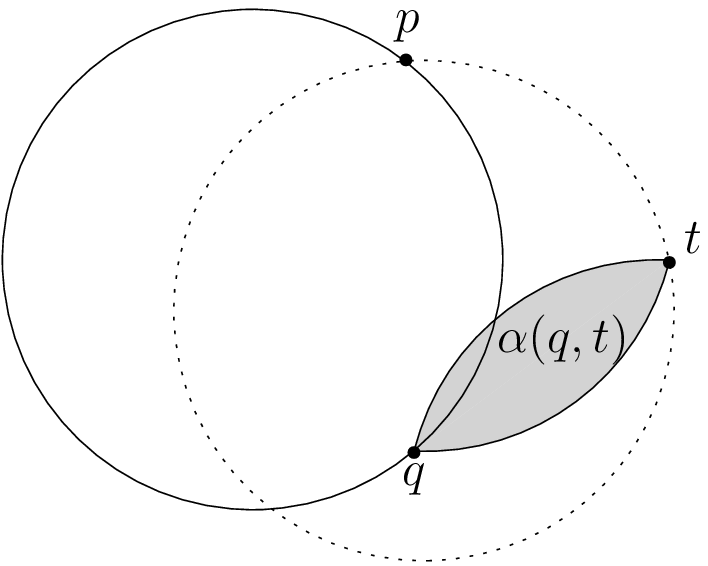}
\caption{\footnotesize Illustrating Observation~\ref{obser:turn}(2). The dotted circle depicts $D(cw(p,t))$.}
\label{fig:outerturn}
\end{center}
\end{minipage}
\end{figure}

\begin{observation}\label{obser:turn}
Consider a minor arc $cw(p,q)$ and a point $t$.
\begin{enumerate}
\item
Suppose $cw(p,q)$ and $t$ form an inner turn. If $t$ is not in the interior of $\alpha(p,q)$, then $p$ is contained in the disk $D(cw(q,t))$; e.g., see Fig.~\ref{fig:innerturn}.
\item
Suppose $cw(p,q)$ and $t$ form an outer turn. If $\{p,q,t\}$ is unit disk coverable and $p$ is not in the interior of $\alpha(q,t)$, then $q$ is contained in the disk $D(cw(p,t))$; e.g., see Fig.~\ref{fig:outerturn}.
\end{enumerate}
\end{observation}
\begin{proof}
For the first statement, since $cw(p,q)$ and $t$ form an inner turn, $t\in D(cw(p,q))$. As $t$ is not in the interior of $\alpha(p,q)$, one can verify from Fig.~\ref{fig:innerturn} that $D(cw(q,t))$ must contain $p$.

We next prove the second statement. Because $\{p,q,t\}$ is unit disk coverable, $\alpha(\{p,q,t\})$ exists.
As $p$ is not in the interior of $\alpha(q,t)$, $p$ must be a vertex of $\alpha(\{p,q,t\})$. Let $a$ be the clockwise neighbor of $p$ on $\alpha(\{p,q,t\})$. Hence, $cw(p,a)$ is an arc of $\alpha(\{p,q,t\})$ and $a$ is either $q$ or $t$. Also, $D(cw(p,a))$ covers $\{p,q,t\}$ by Observation~\ref{obser:basic}(2).
If $a=q$, then $D(cw(p,q))$ contains $t$, which contradicts with that $cw(p,q)$ and $t$ form an outer turn. Thus, $a=t$, and $D(cw(p,t))$ contains $q$.
\qed
\end{proof}


For any two vertices $v_1$ and $v_2$ on $\alpha(Q)$, we use $\partial{\alpha(Q)}[v_1,v_2]$ to denote the set of vertices of $\alpha(Q)$ clockwise from $v_1$ to $v_2$. In particular, if $v_1=v_2$, then we let $\partial_{\alpha(Q)}[v_1,v_2]$ consist of all vertices of $\alpha(Q)$.
Define $\partial_{\alpha(Q)}(v_1,v_2)=\partial_{\alpha(Q)}[v_1,v_2]\setminus\{v_1,v_2\}$.
We use $\overline{\partial_{\alpha(Q)}[v_1,v_2]}$ to refer to the subset of vertices of $\alpha(Q)$ not in $\partial_{\alpha(Q)}[v_1,v_2]$, and define $\overline{\partial_{\alpha(Q)}(v_1,v_2)}$ similarly.


Let $p$ be a point outside $\alpha(Q)$, and $cw(a,p)$ and $ccw(b,p)$ are the upper and lower tangents from $p$ to $\alpha(Q)$, respectively; e.g., see Fig~\ref{fig:tangent}.
Recall that by replacing the arcs of $\alpha(Q)$ clockwise from $a$ to $b$ with the two tangents, we can obtain $\alpha(Q\cup \{p\})$. Hence, $\partial_{\alpha(Q)}(a,b)$ consists of exactly those vertices of $\alpha(Q)$ that are not vertices of $\alpha(Q\cup \{p\})$. We further have the following observation.


\begin{observation}\label{obser:tangentpos}
Suppose $cw(a,p)$ and $ccw(b,p)$ are the upper and lower tangents from a point $p$ to $\alpha(Q)$, respectively; e.g., see Fig.~\ref{fig:tangent}.
\begin{enumerate}
\item
For any vertex $c$ in $\partial_{\alpha(Q)}(a,b)$, there is no disk with $c$ on the boundary that contains $Q\cup\{p\}$.

\item
For any vertex $c$ in $\overline{\partial_{\alpha(Q)}[a,b]}$, any disk tangent to $\alpha(Q)$ at $c$ covers $Q\cup\{p\}$.

\item
If $p$ is strictly to the right of all points of $Q$, then the rightmost vertex of $\alpha(Q)$ must be in $\partial_{\alpha(Q)}[a,b]$.

\item
If there is a line $l$ through a vertex $v$ of $\alpha(Q)$ such that
all vertices of $Q$ are on the same side of $l$ while $p$ is on the other side, then $v$ must be in $\partial_{\alpha(Q)}[a,b]$.
\end{enumerate}
\end{observation}
\begin{proof}
The first two statements can be easily seen by knowing that $\alpha(Q\cup \{p\})$ can be obtained by replacing the arcs of $\alpha(Q)$ clockwise from $a$ to $b$ by the two tangents $cw(a,p)$ and $ccw(b,p)$.

For the third statement, assume to the contrary that $v\not\in \partial_{\alpha(Q)}[a,b]$, where $v$ is the rightmost vertex of $\alpha(Q)$. Then, $v\in \overline{\partial_{\alpha(Q)}[a,b]}$, and by the second statement, any disk tangent to $\alpha(Q)$ at $v$ covers $Q\cup\{p\}$. Let $v_1=cw(v)$ and $v_2=ccw(v)$. Since $D(cw(v,v_1))$ and $D(ccw(v,v_2))$ are both tangent to $\alpha(Q)$ at $v$, both disks cover $Q\cup\{p\}$. Hence, $Z=D(cw(v,v_1))\cap D(ccw(v,v_2))$ contains $p$. Since $D(cw(v,v_1))$ covers $Q$, it contains $v_2$. Since $D(ccw(v,v_2))$ covers $Q$, it contains $v_1$. Let $l_v$ be the vertical line through $v$. We claim that $l_v$ must intersect one of $cw(v,v_1)$ and $ccw(v,v_2)$ twice. Indeed, since $l_v$ contains $v$, it intersects each of the two arcs at least once. If $l_v$ does not intersect either arc twice, then since $D(cw(v,v_1))$ contains $v_2$ and $D(ccw(v,v_2))$ contains $v_1$, and both $v_1$ and $v_2$ are to the left of $v$, $Z$ must be to the left of $l_v$. As $p$ is strictly to the right of $l_v$, $p$ cannot be in $Z$, incurring contradiction.
Hence, $l_v$ intersects one of $cw(v,v_1)$ and $ccw(v,v_2)$ twice.
Assume without loss of generality that $l_v$ intersects $cw(v,v_1)$ twice. This implies that the region of $D(cw(v,v_1))$ to the right of $l_v$ is a subset of $D_1(cw(v,v_1))$. Since $p$ is to the right of $l_v$ and $p$ is in $D(cw(v,v_1))$, $p$ must be in the region of  $D(cw(v,v_1))$ to the right of $l_v$ and thus is in $D_1(cw(v,v_1))$. Because $D_1(cw(v,v_1))\subseteq \alpha(Q)$, $p$ is in $\alpha(Q)$. But this means that there are no tangents from $p$ to $\alpha(Q)$, incurring contradiction.

The fourth statement can be proved in the same way as above by rotating the coordinate system so that $l$ is vertical and $p$ is on its right side. \qed
\end{proof}

Let $Q_1$ be the subset of $Q$ to the left of the vertical line $\ell$ and $Q_2=Q\setminus Q_1$.
Let $cw(a_1,a_2)$ and $ccw(b_1,b_2)$ be the upper and lower common tangents of $\alpha(Q_1)$ and $\alpha(Q_2)$, respectively, i.e., $a_1$ and $b_1$ are the tangent points on $\alpha(Q_1)$ and $a_2$ and $b_2$ are the tangent points on $\alpha(Q_2)$; e.g., see Fig.~\ref{fig:commontangent}. Then, the following arcs constitute the boundary of $\alpha(Q)$ in clockwise order: arcs of $\alpha(Q_1)$ clockwise from $b_1$ to $a_1$, $cw(a_1,a_2)$, arcs of $\alpha(Q_2)$ clockwise from $a_2$ to $b_2$, and $cw(b_2,b_1)$.
The following observation generalizes Observation~\ref{obser:tangentpos}.

\begin{observation}\label{obser:tangentposgeneral}
Suppose $cw(a_1,a_2)$ and $ccw(b_1,b_2)$ are the upper and lower common tangents of $\alpha(Q_1)$ and $\alpha(Q_2)$, respectively; e.g., see Fig.~\ref{fig:commontangent}.
\begin{enumerate}
\item
For any vertex $c$ in $\partial_{\alpha(Q_1)}(a_1,b_1)\cup \partial_{\alpha(Q_2)}(b_2,a_2)$, there is no disk with $c$ on the boundary that contains $Q$.

\item
For any vertex $c$ in $\overline{\partial_{\alpha(Q_1)}[a_1,b_1]}$, any disk tangent to $\alpha(Q_1)$ at $c$ contains $Q$. For any vertex $c$ in $\overline{\partial_{\alpha(Q_2)}[b_2,a_2]}$, any disk tangent to $\alpha(Q_2)$ at $c$ contains $Q$.

\item
The rightmost vertex of $\alpha(Q_1)$ must be in $\partial_{\alpha(Q_1)}[a_1,b_1]$.
The leftmost vertex of $\alpha(Q_2)$ must be in $\partial_{\alpha(Q_2)}[b_2,a_2]$.

\end{enumerate}
\end{observation}
\begin{proof}
The first two statements simply follow from how we can obtain $\alpha(Q)$ from $\alpha(Q_1)$ and $\alpha(Q_2)$ using the two common tangents.

For the third statement, we only show the case for the rightmost vertex of $\alpha(Q_1)$ and the other case can be treated likewise. The proof is similar to that for Observation~\ref{obser:tangentpos} and we briefly discuss it. Let $v$ be the rightmost vertex of $\alpha(Q_1)$. Assume to the contrary that $v$ is not in $\partial_{\alpha(Q_1)}[a_1,b_1]$. Then, by the second statement, both $D(cw(v,v_1))$ and $D(ccw(v,v_2))$ cover $Q$, where $v_1=cw(v)$ and $v_2=ccw(v)$. Hence, $Z=D(cw(v,v_1))\cap D(ccw(v,v_2))$ covers $Q$. Since $Q_1$ is to the left of $\ell$ while $Q_2$ is to the right of $\ell$, by the same analysis as that for Observation~\ref{obser:tangentpos} we can show that $\ell$ must intersect one of $cw(v,v_1)$ and $ccw(v,v_2)$ twice. Assume without loss of generality that $l$ intersects $cw(v,v_1)$ twice. This implies $D_1(cw(v,v_1))$ contains all points of $Q_2$. Since $D_1(cw(v,v_1))\subseteq \alpha(Q_1)$, we obtain that $\alpha(Q_1)$ contains all points of $Q_2$.
But this means that there are no common tangents between $\alpha(Q_1)$ and $\alpha(Q_2)$, incurring contradiction.
\qed
\end{proof}

\subsection{The Static Algorithm}
\label{sec:static}

In this subsection, we assume that $Q=L=\{p_1,p_2,\ldots,p_n\}$ and we provide an $O(n)$ time algorithm for computing $\alpha(Q)$.
The algorithm incrementally processes the points from $p_1$ to $p_n$. Hence, one may either consider it as a static algorithm or a semi-dynamic algorithm for point insertions only.
The algorithm will determine whether $\alpha(Q)$ exists, and if yes, compute and store the vertices of $\alpha(Q)$ in a circular doubly linked list.

The algorithm is similar in spirit to Graham's scan for computing convex hulls.
However, unlike the convex hull case, where it is possible to compute the upper and lower hulls separately, here we need to compute $\alpha(Q)$ as a whole because updating either the upper or the lower hull may end up with updating the other hull. Our algorithm relies on the following lemma.

\begin{lemma}\label{lem:hullexist}
Suppose $p$ is a point outside the circular hull $\alpha(P)$ of a point set $P$. Then, $\{p\}\cup P$ is unit disk coverable if and only if one of the following is true.
\begin{enumerate}
\item
$p$ is in the supporting disk of an arc of $\alpha(P)$.
\item
$\alpha(P)$ has a vertex $v$ such that $\alpha(P)$ is contained in $\alpha(v,p)$. Further, this is true if and only if both $cw(v)$ and $ccw(v)$ are in $\alpha(v,p)$.
 \end{enumerate}
\end{lemma}
\begin{proof}
The ``if'' direction is easy. If $p$ is in the supporting disk $D$ of an arc of $\alpha(P)$, then since $D$ also covers $P$, $D$ covers $P\cup \{p\}$. If $\alpha(P)$ has a vertex $v$ such that $\alpha(P)$ is contained in $\alpha(v,p)$, then $D(cw(v,p))$ contains $\alpha(v,p)$ and thus contains $\alpha(P)$. Hence, $D(cw(v,p))$ covers $P\cup \{p\}$.
In the following, we prove the ``only if'' direction.

Let $D$ be a disk that contains $P\cup \{p\}$. Clearly, it is possible to move $D$ such that $D$ covers $P\cup \{p\}$ and $\partial D$ contains a point $v$ of $P$. By Observation~\ref{obser:basic}(1), $v$ is a vertex of $\alpha(P)$. Now we rotate $D$ around $v$ clockwise (so that $v$ is always on $\partial D$) and keep $D$ covering $P\cup \{p\}$ until $\partial D$ meets another point $z\in P\cup\{p\}$.
If $z\in P$, then $z$ must be the clockwise neighbor of $v$ on $\alpha(P)$ and now $D=D(cw(v,z))$. Since $p$ is in $D$, the first lemma statement holds. Below we assume that $z\not\in P$, i.e., $z=p$.

Since $z=p$, $D$ is $D(cw(v,p))$, and thus $D(cw(v,p))$ covers $P$. By Observation~\ref{obser:basic}(4), $D(cw(v,p))$ also contains $\alpha(P)$. Now, we rotate $D$ around $v$ counterclockwise and keep $D$ containing $P\cup \{p\}$ until $\partial D$ meets another point $z'\in P\cup\{p\}$. Depending on whether $z'\in P$, there are two cases.
If $z'\in P$, then by the same analysis as above, the first lemma statement follows. Otherwise, as above, we can obtain that $D(ccw(v,p))$ contains $\alpha(P)$. Because $\alpha(v,p)=D(cw(v,p))\cap D(ccw(v,p))$ and both $D(cw(v,p))$ and $D(ccw(v,p))$ contain $\alpha(P)$, we obtain that $\alpha(v,p)$ contains $\alpha(P)$. Therefore, the second lemma statement holds.

It remains to show that $\alpha(P)\subseteq \alpha(v,p)$ if and only if both $cw(v)$ and $ccw(v)$ are in $\alpha(v,p)$. If $\alpha(P)$ is contained in $\alpha(v,p)$, then it is obviously true that both $cw(v)$ and $ccw(v)$ are in $\alpha(v,p)$. Now assume that both $cw(v)$ and $ccw(v)$ are in $\alpha(v,p)$.
Since $\alpha(v,p)=D(cw(v,p))\cap D(ccw(v,p))$, both $cw(v)$ and $ccw(v)$ are in $D(cw(v,p))$ and also in  $D(ccw(v,p))$. By Observation~\ref{obser:tangent}, both $D(cw(v,p))$ and $D(ccw(v,p))$ are tangent to $\alpha(P)$ at $v$ and thus both disks contain $\alpha(P)$. Therefore, $\alpha(P)\subseteq D(cw(v,p))\cap D(ccw(v,p))=\alpha(v,p)$. \qed
\end{proof}

We process the vertices of $Q=\{p_1,p_2,\ldots, p_n\}$ incrementally from $p_1$ to $p_n$.
We use a circular doubly linked list $\calL$ to maintain the vertices of the current circular hull that has been computed. Each vertex in the list has a $cw$ pointer and a $ccw$ pointer to refer to its clockwise and counterclockwise neighbors on the current hull, respectively. In addition, we maintain the rightmost vertex $v^*$ of  the current hull, which is also the access we have to $\calL$.
Initially we directly compute $\alpha(q_1,q_2)$ and set up the list $\calL$, with $v^*=q_2$.
For each $i=1,\ldots,n$, let $Q_i=\{p_1,p_2,\ldots,p_{i}\}$.

Consider a general step for processing a new vertex $p_i$ with $i\geq 3$, and suppose $\calL$ now stores the circular hull of $Q_{i-1}$. With $v^*$, we can find the rightmost arc $w$ of the current hull. If $p_i$ is in $D_1(w)$, then $p_i$ is ``redundant'' by Observation~\ref{obser:red}, i.e., $p_i$ does not affect the current circular hull, so we do not need to do anything (i.e., no need to update $\calL$). Otherwise, our goal is to find the two tangents from $p_i$ to the current hull, or decide that they do not exist. Starting from $v^*$, we first run a {\em counterclockwise scanning procedure} to search the upper tangent, as follows (see Algorithm~\ref{algo:ccwscan} for the pseudocode).
Starting with $v=v^*$, we check $v$ in the following way. We first check whether both $cw(v)$ and $ccw(v)$ are in $\alpha(v,p_i)$. If yes, then we stop the procedure and return $cw(v,p_i)$ as the upper tangent. Otherwise, we check whether $cw(ccw(v),v)$ and $p_i$ form an inner turn. If yes, then we stop the procedure and return $cw(v,p_i)$ as the upper tangent. Assume that they form an outer turn. Then, if $ccw(v)\neq v^*$, then we set $v=ccw(v)$ and proceed as above; otherwise, we stop the procedure and conclude that $Q_{i}$ (and thus $Q$) is not unit disk coverable.

\begin{algorithm}[h]
	\caption{The counterclockwise scanning procedure searching the upper tangent}
	\label{algo:ccwscan}
	\SetAlgoNoLine
    $v\leftarrow v^*$\;
    \While{\em true }
    {
         \eIf{\em both $cw(v)$ and $ccw(v)$ are in $\alpha(v,p_i)$}
         {
            return $cw(v,p_i)$ as the upper tangent\;
         }
         {
            \eIf{\em $cw(ccw(v),v)$ and $p_i$ form an inner turn}
            {
                return $cw(v,p_i)$ as the upper tangent\;
            }
            {
                \eIf{\em $ccw(v)\neq v^*$}
                {
                    $v\leftarrow ccw(v)$\;
                }
                {
                    return null and conclude that $\alpha(Q_i)$ (and thus $\alpha(Q)$) does not exist\;
                }
            }
         }
     }		
\end{algorithm}

It is not difficult to see that the algorithm will eventually stop.
The following lemma proves the correctness of the algorithm.

\begin{lemma}
The counterclockwise scanning procedure will decide whether $\alpha(Q_{i})$ exists, and if yes, find the upper tangent from $p_i$ to $\alpha(Q_{i-1})$ unless $p_i$ is redundant.
\end{lemma}
\begin{proof}
First of all, if $p_i$ is redundant, then our algorithm correctly determines it. Below we assume that $p_i$ is not redundant. Suppose the procedure is checking the vertex $v$. There are three cases for the procedure to stop: $cw(v)$ and $ccw(v)$ are in $\alpha(v,p_i)$; $cw(ccw(v),v)$ and $p_i$ form an inner turn; $cw(ccw(v),v)$ and $p_i$ form an outer turn and $v^*=ccw(v)$. In the first two cases, we will show that $cw(v,p_i)$ is the upper tangent. In the third case, we will show that $Q_{i}$ is not unit disk coverable.

If $cw(v)$ and $ccw(v)$ are in $\alpha(v,p_i)$, then by Lemma~\ref{lem:hullexist}(2), $\alpha(Q_{i-1})\subseteq \alpha(v,p_i)$. Hence, $\alpha(v,p_i)=\alpha(Q_{i})$. Since $cw(v,p_i)$ is an arc of $\alpha(v,p_i)$,  $D(cw(v,p_i))$ contains $\alpha(v,p_i)$ and thus $\alpha(Q_{i-1})$.
Therefore, $cw(v,p_i)$ is the upper tangent from $p_i$ to $\alpha(Q_{i-1})$.

If $cw(ccw(v),v)$ and $p_i$ form an inner turn, to show that $cw(v,p_i)$ is tangent to $\alpha(Q_{i-1})$ at $v$, by Observation~\ref{obser:tangent} it suffices to show that $D(cw(v,p_i))$ contains both $cw(v)$ and $ccw(v)$.
Since $p_i$ is not redundant and $p_i$ is to the right of both $ccw(v)$ and $v$, $p_i$ is not in $\alpha(ccw(v),v)$.
Because $cw(ccw(v),v)$ and $p_i$ form an inner turn, by Observation~\ref{obser:turn}(1), $D(cw(v,p_i))$ contains $ccw(v)$. Next we prove $cw(v)\in D(cw(v,p_i))$. Depending on whether $v=v^*$, there are two subcases.

\begin{itemize}
\item
If $v\neq v^*$, then according to our algorithm, $cw(v,cw(v))$ and $p_i$ form an outer turn.
Because $cw(ccw(v),v)$ and $p_i$ form an inner turn, $p_i\in D(cw(ccw(v),v))$. Since $cw(ccw(v),v)$ is an arc of $\alpha(Q_{i-1})$, $D(cw(ccw(v),v))$ contains $Q_{i-1}$ and thus $cw(v)$. Hence, $D(cw(ccw(v),v))$ contains $\{v,cw(v),p_i\}$, and thus, $\{v,cw(v),p_i\}$ is unit disk coverable.

We claim that $v$ is not in the interior of $\alpha(p_i,cw(v))$. Indeed, assume to the contrary this is not true. Then, since $v$ is on the boundary of $D(cw(ccw(v),v))$, one of $p_i$ and $cw(v)$, as two vertices of $\alpha(p_i,cw(v))$ must be outside $D(cw(ccw(v),v))$. However, we have proved above that $D(cw(ccw(v),v))$ contains both $p_i$ and $cw(v)$, incurring contradiction.

Since $v$ is not in the interior of $\alpha(p_i,cw(v))$, by Observation~\ref{obser:turn}(2), $D(cw(v,p_i))$ contains $cw(v)$.

\item
If $v=v^*$, then in the same way as the above case we can show that $D(cw(ccw(v),v))$ contains $\{v,cw(v),p_i\}$, and thus, $\{v,cw(v),p_i\}$ is unit disk coverable.


We claim that $cw(v,cw(v))$ and $p_i$ form an outer turn.
Assume to the contrary that they form an inner turn. Then, $p_i\in D(cw(v,cw(v)))$. As $p_i\in D(cw(ccw(v),v))$, we obtain that $p_i\in D(cw(v,cw(v)))\cap D(cw(ccw(v),v))$. Since $cw(v)$ and $ccw(v)$ are to the left of $v$ and $p_i$ is to the right of $v$, by a similar argument as in the proof of Observation~\ref{obser:tangentpos}(3), we can show that $p_i$ is inside $\alpha(Q_{i-1})$, implying that $p_i$ is redundant, which incurs contradiction because $p_i$ is not redundant.


Further, using the same analysis as the above subcase $v\neq v^*$, we can show that $v$ is not in the interior of $\alpha(p_i,cw(v))$. Consequently, by Observation~\ref{obser:turn}(2), $cw(v)$ is in $D(cw(v,p_i))$.

\end{itemize}

It remains to discuss the third case where
$cw(ccw(v),v)$ and $p_i$ form an outer turn and $v^*=ccw(v)$. According to our algorithm, this case happens only if both of the followings are true: (1) for each vertex $v$ of $\alpha(Q_{i-1})$, $\alpha(v,p_i)$ does not contain both $cw(v)$ and $ccw(v)$;
(2) for each arc $cw(ccw(v'),v')$ of $\alpha(Q_{i-1})$, it does not form an inner turn with $p_i$ (i.e., $p_i\not\in D(cw(ccw(v'),v'))$), implying that $p_i$ is not in the supporting disk of any arc of $\alpha(Q_{i-1})$. According to Lemma~\ref{lem:hullexist}, $Q_{i}$ is not unit disk coverable.
\qed
\end{proof}

If the above procedure finds the upper tangent, then we run a symmetric {\em clockwise scanning procedure} to find the lower tangent (which guarantees to exist, for the upper tangent exists).
Next, we replace the vertices in $\calL$ clockwise strictly from the upper tangent point to the lower tangent point by $p_i$, and then reset $v^*$ to $p_i$. The runtime of the two procedures is $O(1+k)$, where $k$ is the number of vertices removed from $\calL$. After a point is removed from $\calL$, it will never appear in $\calL$ again. Hence the total time of the algorithm for processing all points $\{p_1,\ldots,p_n\}$ is $O(n)$.

\begin{theorem}\label{theo:insertions}
We can maintain the circular hull of a set $Q$ of points such that if a new point to the right of all points of $Q$ is inserted, in $O(1)$ amortized time we can decide whether $\alpha(Q)$ exists, and if yes, update $\alpha(Q)$.
\end{theorem}

\begin{corollary}
Given a set of points in the plane sorted by $x$-coordinates, there exists a linear time algorithm that can decide whether its circular hull exists, and if yes, compute the circular hull.
\end{corollary}

\subsection{The Data Structure for Dynamically Maintaining  $\alpha(Q)$}
\label{sec:datastructure}

In this subsection, we explain our data structure for maintaining $\alpha(Q)$ under both insertions and deletions on $Q$. Recall that $Q$ is a subset of $L\cup R$ and the vertical line $\ell$ separates $L$ and $R$. Let $Q_1=Q\cap L$ and $Q_2=Q\cap R$. Our data structure will maintain $\alpha(Q_1)$ and $\alpha(Q_2)$ separately. Recall that each insertion happens to a point in $R$ and each deletion happens to a point in $L$. Our goal is determine whether $\alpha(Q)$ exists after each update.

For $Q_2$, we use a circular doubly linked list to maintain
$\alpha(Q_2)$, in the same way as in the static algorithm. As such, from
any vertex $v$ of $\alpha(Q_2)$, we can visit its two neighbors $cw(v)$
and $ccw(v)$ in constant time.
If a point is inserted, then we update $\alpha(Q_2)$ as in the static
algorithm. In addition, we also store explicitly the leftmost arc of
$\alpha(Q_2)$ whenever it is updated, which introduces only a constant
time to the previous algorithm. If $\alpha(Q_2)$ does not exist after
an insertion, then since $Q_2\subseteq Q$ and no point from $Q_2$ will
be deleted, $\alpha(Q)$ will not exist after any update in future and
thus we can halt the entire algorithm. Without loss of generality, we assume that $\alpha(R)$ exists and thus $\alpha(Q_2)$ always exists.

For $Q_1$, because points of $Q_1$ are deleted in order from left to right, initially when $Q_1=L$, we build the circular doubly linked list by processing points of $L$ {\em from right to left}, i.e., from $p_n$ to $p_1$.
Further, in order to maintain some historical information,
we have each vertex $v$ of $\alpha(Q_2)$ associated with two stacks $S_{cw}(v)$ and $S_{ccw}(v)$, which are empty initially. Specifically, initially we process the points of $L$ incrementally from $p_n$ to $p_1$. Consider a general step of the algorithm processing a point $p_i$. Suppose $cw(v_1,p_i)$ and $ccw(v_2,p_i)$ are the two tangents found by using our static algorithm. Then, in addition to the processing in the static algorithm, we push $v_1$ into $S_{ccw}(p_i)$, push $v_2$ into $S_{cw}(p_i)$, and push $p_i$ into both $S_{cw}(v_1)$ and $S_{ccw}(v_2)$.
Note that this does not change the time complexity of our previous static algorithm asymptotically.
Later when $p_i$ is deleted, we simply pop $p_i$ out of both $S_{cw}(v_1)$ and $S_{ccw}(v_2)$.
In this way, at any moment during processing the deletions of $Q_1$, for any vertex $v$ in the current circular hull $\alpha(Q_1)$,
the top of $S_{cw}(v)$ (resp., $S_{ccw}(v)$) is always the clockwise (resp., counterclockwise) neighbor of $v$ on $\alpha(Q_1)$, which can be accessed in constant time from the vertex $v$. So we can use these stacks to replace the circular doubly linked list, and we call it {\em the stack data structure}. In addition, for handling insertions, we also explicitly store, say in an array $A$, the rightmost arc of the current circular hull after processing each point of $L$ (i.e., given $i$, $A[i]$ stores the rightmost arc of the circular hull of $\{p_i,p_{i+1},\ldots,p_n\}$).
These only introduces constant time to our original static algorithm. If during processing a new point $p_i$ we find that the circular hull of $\{p_i,\ldots,p_n\}$ does not exist, then we stop the algorithm and set $start=i$. In this way, whenever we process a deletion on $L$, if the index of the deleted point is smaller than or equal to $start$, then we know that $\alpha(Q_1)$ and thus $\alpha(Q)$ does not exist and we do not need to do anything. Without loss of generality, we assume that $\alpha(L)$ exists and thus $\alpha(Q_1)$ always exists (so the variable $start$ is not needed any more).

The above describes our data structure for maintaining $\alpha(Q_1)$ and $\alpha(Q_2)$.
We also need to maintain other information. To explain them,
we first show a property, as follows.

Although $Q_1$ is to the left of $\ell$, $\alpha(Q_1)$ may cover some region of the plane to the right of $\ell$, denoted by $\alpha'(Q_1)$, and if $w$ is the rightmost arc of $\alpha(Q_1)$, then $\alpha'(Q_1)$ is exactly the portion of $D_1(w)$ to the right of $\ell$ due to the convexity of $\alpha(Q_1)$~\cite{ref:HershbergerFi91}. Symmetrically, we define $\alpha'(Q_2)$ as the region of $\alpha(Q_2)$ to the left of $\ell$. The following lemma shows that as points are deleted from $Q_1$, $\alpha'(Q_1)$ becomes monotonically smaller, and as points are inserted into $Q_2$, $\alpha'(Q_2)$ becomes monotonically larger.

\begin{figure}[t]
\begin{minipage}[t]{\textwidth}
\begin{center}
\includegraphics[height=2.0in]{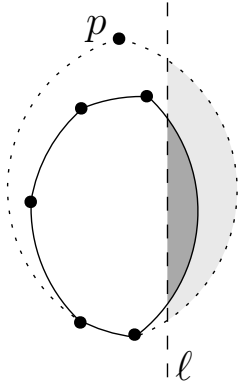}
\caption{\footnotesize Illustrating Lemma~\ref{lem:extremearcs}, where $Q_1=Q_1'\cup \{p\}$. The light (resp., dark) gray region is $\alpha'(Q_1)$ (resp., $\alpha'(Q_1')$).}
\label{fig:contain}
\end{center}
\end{minipage}
\vspace{-0.15in}
\end{figure}

\begin{lemma}\label{lem:extremearcs}
If $Q'_1\subseteq Q_1$, then $\alpha'(Q_1')\subseteq \alpha'(Q_1)$; e.g., see Fig.~\ref{fig:contain}. Similarly, if $Q'_2\subseteq Q_2$, then $\alpha'(Q_2')\subseteq \alpha'(Q_2)$.
\end{lemma}
\begin{proof}
We only prove the case for $Q_1$, and the other case for $Q_2$ can be treated likewise.
Indeed, let $w$ and $w'$ be the rightmost arcs of $Q_1$ and $Q_1'$, respectively. If $w=null$, then $w'$ must be $null$ due to $Q_1'\subseteq Q_1$, and thus we have $\alpha'(Q_1')=\alpha'(Q_1)=\emptyset$. Assume that $w\neq null$. If $w'=null$, then since $\alpha'(Q_1')=\emptyset$ and $\alpha'(Q_1)\neq \emptyset$, $\alpha'(Q_1')\subseteq \alpha'(Q_1)$ holds. Assume that $w'\neq null$ (e.g., see Fig.~\ref{fig:contain}).
Since $w$ is an arc of $\alpha(Q_1)$, $D(w)$ contains $Q_1$ and thus $Q_1'$.
By Observation~\ref{obser:basic}(4), $D(w)$ contains $\alpha(Q_1')$, and thus, $D(w)$ contains the arc $w'$.
Note that $\alpha'(Q_1')$ is bounded from the left by $\ell$ and bounded from the right by the portion of $w'$ to the right of $\ell$. Since $\alpha'(Q_1)$ is the region of $D(w)$ to the right of $\ell$ and $D(w)$ contains $w'$, it must hold that $\alpha'(Q_1')\subseteq \alpha'(Q_1)$.\qed
\end{proof}

In addition to the data structures for $\alpha(Q_1)$ and $\alpha(Q_2)$ described above, our dynamic algorithm also maintains the following information. Recall that based on our assumption both $\alpha(Q_1)$ and $\alpha(Q_2)$ always exist.

\begin{enumerate}
\item
If $Q_2$ is contained in $\alpha(Q_1)$, i.e., the {$Q_1$-dominating case}, then our algorithm will detect it, and in this case $\alpha(Q)=\alpha(Q_1)$ and $\alpha(Q)$ exists.
\item
If $Q_1$ is contained in $\alpha(Q_2)$, i.e., the {$Q_2$-dominating case}, then our algorithm will detect it, and in this case $\alpha(Q)=\alpha(Q_2)$ and $\alpha(Q)$ exists. Further, because in future deletions will only happen to $Q_1$ and insertions will only happen to $Q_2$,
Lemma~\ref{lem:extremearcs} implies that $\alpha(Q)=\alpha(Q_2)$ always holds. Therefore, in future we can ignore all deletions and only handle insertions, which can be done by simply applying the static algorithm on $Q_2$.
\item
If neither of the above cases happens, then our algorithm will detect whether $\alpha(Q)$ exists, and if yes, the two common tangents of $\alpha(Q_1)$ and $\alpha(Q_2)$ will be explicitly maintained.
\end{enumerate}

\subsection{Initialization}
\label{sec:initial}

Initially, $Q=Q_1=L$, so we build the data structure for $\alpha(Q_1)$
as discussed before. This takes $O(n)$ time. Since there are
$2n$ update operations, the amortized cost is $O(1)$.

One annoying issue is to check whether $Q_1$- or $Q_2$-dominating case
will happen after each update. We show how to resolve the issue. We discuss the $Q_1$-dominating case first.

Recall $R=\{q_1,q_2,\ldots,q_n\}$ is sorted from left to right.
When $q_1$ is inserted into $Q$ (i.e., this is the first insertion), it is quite trivial to determine whether the $Q_1$-dominating case happens, which can be done in constant time by checking whether $q_1$ is contained in the supporting circle of the rightmost arc of $\alpha(Q_1)$ (which is maintained after each deletion). However, the problem becomes challenging after more points are inserted. We use the following strategy to resolve the problem ``once for all''.

An easy observation is that once the $Q_1$-dominating case does not
happen for the first time after an update (which may be either an insertion or a
deletion), in light of Lemma~\ref{lem:extremearcs}, it will not never
happen in future, because $Q_1$ will become smaller while $Q_2$ will
become larger. Also, before that particular update,
$\alpha(Q)=\alpha(Q_1)$ holds and thus $\alpha(Q)$ exists.
Lemma~\ref{lem:q1dominating} gives an $O(n)$ time algorithm to find
that particular update. Note that this procedure is only performed once in the entire algorithm.

\begin{lemma}\label{lem:q1dominating}
The first update after which the $Q_1$-dominating case does not happen can be determined in $O(n)$ time.
\end{lemma}
\begin{proof}
For each $i=1,2,\ldots,n$, we use $\alpha'[i,n]$ to refer to $\alpha'(\{p_i,p_{i+1},\ldots,p_n\})$. As discussed before, each $\alpha'[i,n]$ is the part of a unit disk on the right side of the line $\ell$. By Lemma~\ref{lem:extremearcs}, it holds that $\alpha'[i,n]\subseteq \alpha'[i-1,n]$ for all $i=2,3,\ldots,n$.
Recall that the rightmost arc is maintained by our algorithm after each deletion of $L$. Thus, given $i$, $\alpha'[i,n]$ can be obtained in $O(1)$ time.

%
%

From the outset, we process insertions and deletions as follows. During the algorithm, we maintain a variable $i^*$, which is the first deletion after which the $Q_1$-dominating case will not happen for the points in the current set $Q_2$. Initially before any deletion or insertion, $Q_1=L$ and $Q_2=\emptyset$, and thus we set $i^*=n$. For each deletion of a point $p_i$, if $i<i^*$, then we proceed on the next update; otherwise we return the deletion of $p_i$ as the answer to the problem. Consider an insertion of a point $q_j$. We first check whether $q_j$ is in $\alpha'[i^*,n]$. If yes, we proceed on the next update. Otherwise, we keep decrementing $i^*$ until $q_j\in \alpha'[i^*,n]$ or $i^*=0$.
Then we check whether $i^*<i$, where $i$ is the index of the leftmost point of the current set $Q_1$ (i.e.,
$Q_1=\{p_i,\ldots,p_n\}$). If $i^*< i$, then we return the insertion of $q_j$ as the answer to the problem. Otherwise, we proceed on the next update.

The correctness of the algorithm is based on Lemma~\ref{lem:extremearcs}.
It is not difficult to see that the algorithm runs in $O(n)$ time.
\qed
\end{proof}

Lemma~\ref{lem:q1dominating} finds the update after which the $Q_1$-dominating case does not happen for the first time. Regardless of whether it is an insertion or a deletion, let $Q_1$ and $Q_2$ be the two subsets right after the update. So we know that both $\alpha(Q_1)$ and $\alpha(Q_2)$ exist, and the $Q_1$-dominating case does not happen.

Next, we discuss how to detect whether the $Q_2$-dominating case
happens after each update in future (starting from the update found in Lemma~\ref{lem:q1dominating}), by
a {\em $Q_2$-dominating case detection procedure}, as follows. As discussed before, once we find the $Q_2$-dominating case happens for the first time after an update, we can simply use our static algorithm to handle the deletions only in future.
Starting from $j^*=n$, we check whether $p_{j^*}$ is in the supporting
disk $D$ of the leftmost arc of the current $\alpha(Q_2)$.
Recall that the leftmost arc of $\alpha(Q_2)$ is explicitly stored (and if it is $null$, then its supporting disk is $\emptyset$). If yes,
we decrement $j^*$ until $j^*=0$ or $p_{j^*}\not\in D$ (thus all points of $L$ from $p_{j^*+1}$ to $p_n$ are in $D$).
Now consider an insertion to $Q_2$. If the leftmost
arc of $\alpha(Q_2)$ gets updated, then by
Lemma~\ref{lem:extremearcs}, all points of $L$ from $p_{j^*+1}$ to $p_n$ are
still contained in the supporting disk $D$ of the new leftmost arc. We
further check whether $p_{j^*}$ is in $D$.
If yes, we decrement $j^*$ until $j^*=0$ or $p_{j^*}\not\in D$. Let $i^*$
be the index of the leftmost point of the current set $Q_1$. Whenever $j^*$
decrements as above, if $i^*>j^*$, then we know the $Q_2$-dominating case
happens and then we only need to process the insertions using the static algorithm in future.
Similarly, when $p_{i^*}$ is deleted, we increment $i^*$ by one, and if
$i^*>j^*$, and we again run into the $Q_2$-dominating case.

In the following discussion on processing updates,
before actually processing each update, we run the above procedure to
check whether the $Q_2$-dominating case happens. If yes, then the rest
of the algorithm is trivial. Otherwise, we will perform the
corresponding algorithm (to be discussed below) for processing the update. Hence, the $Q_2$-dominating case detection procedure is actually part of the update processing algorithm. In the following discussion whenever
we process an insertion or a deletion, we assume that the
$Q_2$-dominating case will not happen after the operation.
It is easy to see that the procedure takes $O(n)$ time in the entire algorithm for processing all $2n$ updates.

According to the above discussion, we start from the update found by Lemma~\ref{lem:q1dominating}, and neither dominating case will happen. This implies that the common tangents of $\alpha(Q_1)$ and $\alpha(Q_2)$ exist if and only if $\alpha(Q)$ exists. Next, we present an $O(n)$ time procedure to decide whether $\alpha(Q)$ exists, and if yes, find the two common tangents. Note that this procedure is performed only once, e.g., after the update of Lemma~\ref{lem:q1dominating}, which does not affect the $O(1)$ amortized time performance per update.

Because we do not know whether $\alpha(Q)$ exists, we apply our static algorithm processing the points of $Q$ from right to left. If during processing a point we determine the current circular hull does not exist, then
we stop the algorithm and let $p$ refer to the point; otherwise let $p=null$.
If $p=null$, then $\alpha(Q)$ exists and we compute the common tangents of $\alpha(Q_1)$ and $\alpha(Q_2)$ by an algorithm given below.
Assume that $p\neq null$. Since $\alpha(Q_2)$ exists, $p$ must be from $Q_1$. Observe that before $p$ is deleted, $\alpha(Q)$ cannot exist. Suppose we consider the next update. If it is a deletion of a point to the left of $p$, then we do nothing but we know $\alpha(Q)$ does not exist. If it is an insertion of a point $q_j$, then we know that $\alpha(Q)$ does not exist, but instead of immediately inserting $q_j$ to our data structure for $Q_2$,
we hold $q_j$ in a first-in-first-out queue $\calQ$, which is $\emptyset$ initially.
If it is the deletion of $p$, then we know that $\alpha(Q)$ exist, where $Q$ does not include the points held in $\calQ$. In this case (and also the case $p=null$),
we find the two common tangents of $\alpha(Q_1)$ and $\alpha(Q_2)$, as follows.

The algorithm is similar to that for finding common tangents of two convex hulls.
Hershberger and Suri gave a linear time algorithm for that~\cite{ref:HershbergerFi91} (see Lemma 4.12). To make the paper self-contained, we sketch a slightly different algorithm.
We first find the upper common tangent as follows. Starting from the leftmost vertex, we consider the vertices of $\alpha(Q_2)$ in the clockwise order. For each vertex, we find its upper tangent to $\alpha(Q_1)$ by using the counterclockwise scanning procedure in our static algorithm. Once we find the upper tangent, we check whether it is also tangent to $\alpha(Q_2)$. If yes, we have found the upper common tangent. Otherwise, we consider the next vertex of $\alpha(Q_2)$, but start the counterclockwise scanning procedure from the current tangent point on $\alpha(Q_1)$. As the upper common tangent exists, the algorithm will eventually find it. We find the lower common tangent in a similar way using the clockwise scanning procedure of our static algorithm. The time is linear in the total number of vertices of $\alpha(Q_1)$ and $\alpha(Q_2)$.

%

After the common tangents are found, if $\calQ\neq \emptyset$ (which only happens if $p\neq null$), then we need to process the insertions on the points in $\calQ$ in order to know whether $\alpha(Q)$ exists after the deletion of $p$. For this, we will apply on these points the insertion algorithm to be given below.

The above describes our initialization procedure, which takes $O(n)$
time. In the following, we present our algorithm for handling
future insertions (including those in $\calQ$)
and deletions. Our algorithm maintains an invariant that is
stated in the following observation.

\begin{observation}\label{obser:invariant}
Suppose the algorithm is about to process an update.
\begin{enumerate}
  \item Before the update, the $Q_1$-dominating case does not happen,
  \item Before the update, the two common tangents of $\alpha(Q_1)$ and $\alpha(Q_2)$ exist and are explicitly computed.
  \item After the update, the $Q_2$-dominating case does not happen.
\end{enumerate}
\end{observation}

The first invariant is established due to that we always process updates after the update computed in Lemma~\ref{lem:q1dominating}.
The third invariant is established by our $Q_2$-dominating case
detection procedure. More precisely, once the procedure detects that
the $Q_2$-dominating case happens after an update, then we will apply
our static algorithm on $\alpha(Q_2)$ with insertions only. The second
invariant has been established above for the moment, and we will show later
that it will be re-established after each future update is processed. We
are able to do so because our insertion processing algorithm may also
involve performing point deletions. For this reason, in the
following we discuss the deletion processing algorithm first.


\subsection{Deletions}
\label{sec:delete}

Suppose a point $p_i$ is deleted from $Q_1$. Let $Q_1'=Q_1\setminus\{p_i\}$ and let $Q_1$ still be the original set before the deletion. Let $Q=Q_1\cup Q_2$ and $Q'=Q_1'\cup Q_2$. Since $\alpha(Q)$ exists (due to Observation~\ref{obser:invariant}(2)), $\alpha(Q')$ exists. Thus, our task is to update the common tangents if they are changed. We show that we can do so in $O(1)$ amortized time.
Let $cw(a_1,a_2)$ and $cw(b_1,b_2)$ denote
the upper and lower common tangents of $\alpha(Q_1)$ and
$\alpha(Q_2)$, respectively, which have been computed by Observation~\ref{obser:invariant}(2).

First of all, if $p_i$ is not the leftmost vertex of $\alpha(Q_1)$ (which has been explicitly stored when we build the data structure for $Q_1=L$ initially), then $p_i$ is in the interior of $\alpha(Q_1)$ and thus
nothing needs to be done (i.e., the common tangents do not change).
Otherwise, let $p=cw(p_i)$, which can be accessed in $O(1)$ time using our stack data structure for $Q_1$.
According to our stack data structure, $p_i$ is at the top of the stack $S_{ccw}(p)$.
We pop $p_i$ out of $S_{ccw}(p)$. We also pop $p_i$ out of $S_{cw}(p')$, where $p'=ccw(p_i)$.
If $p_i\not\in \{a_1,b_1\}$, then the common tangents do not change and thus we are done with the deletion.
Otherwise, we assume that $p_i=a_1$ (the other case can be treated likewise).
Depending on whether $a_1=b_1$, there are two cases.

If $b_1\neq a_1$, then after $a_1$ is deleted, $b_1$ is still a vertex of $\alpha(Q_1')$ and thus $ccw(b_1,b_2)$ is still the lower common tangent. To find the new upper common tangent, we move $p$ counterclockwise on $\alpha(Q_1')$ and simultaneously move $a_2$ counterclockwise on $\alpha(Q_2)$.
This procedure is similar in spirit to finding common
tangent for two convex hulls separated by a vertical line, and we sketch it below (e.g., see Fig.~\ref{fig:tangentupdate}).

We first check whether $cw(p,a_2)$ is tangent to $\alpha(Q'_1)$ at $p$. Recall that by Observation~\ref{obser:tangent} this is can be done by checking whether $D(cw(p,a_2))$ contains $ccw(p)$ and $cw(p)$ (which can be accessed from $p$ in constant time using our stack data structure). If not, then we move $p$ counterclockwise on $\alpha(Q_1')$ until $cw(p,a_2)$ is tangent to $\alpha(Q'_1)$ at $p$. Then, we check whether $cw(p,a_2)$ is tangent to $\alpha(Q_2)$ at $a_2$. If not, then we move $a_2$ counterclockwise on $\alpha(Q_2)$ until $cw(p,a_2)$ is tangent to $\alpha(Q_2)$ at $a_2$. If the new $cw(p,a_2)$ is not tangent to $\alpha(Q'_1)$ at $p$, then we move $p$ counterclockwise again. We repeat the algorithm until $cw(p,a_2)$ is both tangent to $\alpha(Q'_1)$ at $p$ and tangent to $\alpha(Q_2)$ at $a_2$. As the upper common tangent exists, the procedure will eventually find it.

\begin{figure}[t]
\begin{minipage}[t]{\textwidth}
\begin{center}
\includegraphics[height=2.0in]{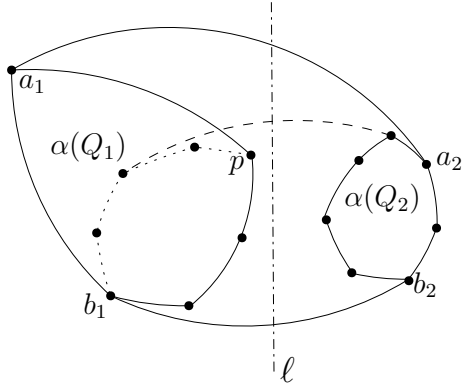}
\caption{\footnotesize Illustrating the new upper common tangent (the dashed one) after $a_1$ is deleted. The dotted curves are arcs on $\alpha(Q_1')$ but not on $\alpha(Q)$. To find the new upper common tangent, one can simultaneously rotate $p$ counterclockwise on $\alpha(Q_1')$ and rotate $a_2$ counterclockwise on $\alpha(Q_2)$.}
\label{fig:tangentupdate}
\end{center}
\end{minipage}
\vspace{-0.15in}
\end{figure}

We then consider the case where $a_1= b_1$. In this case, the lower common tangent is also changed and we need to compute it as well.
As the $Q_2$-dominating case does not happen, both upper and lower common tangents exist.
Thus, we can use the same algorithm as above to find the upper common tangent and use a symmetric algorithm to find the lower common tangent.

In either case above, we call the procedure for finding the upper common tangent the {\em deletion-type upper common tangent searching procedure}, which takes $O(1+k_1+k_2)$ time, where $k_1$ is the number of vertices
of $\alpha(Q'_1)$ strictly counterclockwise from the original $p$ to its new position when the algorithm finishes and $k_2$ is the number of vertices of $\alpha(Q_2)$ strictly counterclockwise from the original $a_2$ to its new position after the algorithm finishes (we say that these vertices are {\em involved} in the procedure). If the lower common tangent is also updated, we call it the {\em deletion-type lower common tangent searching procedure}. Lemma~\ref{lem:cost05} shows that each point can involve in at most one such procedure in the entire algorithm, and thus the amortized cost of the two procedures is $O(1)$.

\begin{lemma}\label{lem:cost05}
Each point of $L\cup R$ can involve in at most one deletion-type upper
tangent searching procedure and at most one deletion-type lower
tangent searching procedure in the entire algorithm (for processing all $2n$ updates).
\end{lemma}
\begin{proof}
We only discuss the upper tangent case, for the lower tangent case is similar.
Let $v$ be a vertex on $\alpha(Q'_1)$ involved in the procedure.
We show that $v$ cannot involve in the procedure again.
Indeed, $v$ was not a vertex of $\alpha(Q_1)$ before $p_i$
is deleted. After $p_i$ is delete, since $v$ is involved in the
procedure, $v$ must be a vertex of $\alpha(Q_1')$.
As only deletions will happen on $Q_1$, $v$ will always be a vertex of the
circular hull of $Q_1$ until it is deleted.
Hence, $v$ will never be involved in the procedure again (because
to involve in the procedure, $v$ cannot be a vertex of the circular hull
of $Q_1$).

Let $q$ be a vertex on $\alpha(Q_2)$ involved in the procedure. Let $a_2$ and $a_2'$ be the old and new upper common tangent points on $\alpha(Q_2)$, respectively. Let $b_2$ and $b_2'$ be the old and new lower common tangent points on $\alpha(Q_2)$, respectively.
Then, $q\in \partial_{\alpha(Q_2)}(a_2',a_2)$. Notice that $\partial_{\alpha(Q_2)}(a_2',a_2)\subseteq \partial_{\alpha(Q_2)}(b_2,a_2)$.
By Observation~\ref{obser:tangentposgeneral}(1), any disk tangent to $\alpha(Q_2)$ at $q$ does not contain $Q_1$. On the other hand, since $q$ is involved in the procedure, we have
$q\in \partial_{\alpha(Q_2)}(a_2',b_2')$ because $a_2$ is moving counterclockwise to $a_2'$ while $b_2$ is moving clockwise to $b_2'$ according to our algorithm. Thus, any disk tangent to $\alpha(Q_2)$ at $q$ must contain the new set $Q_1'$ after the deletion.

Now consider another deletion operation later. We argue that $q$ will not be involved in the same procedure for the deletion. Let $Q''$ be the set of $Q$ right before the deletion, and let $Q_1''=Q''\cap L$ and  $Q_2''=Q''\cap R$. Clearly, $Q_1''\subseteq Q_1'$ and $Q_2\subseteq Q_2''$. Assume to the contrary that $q$ involves in the procedure again. Then, $q$ is a vertex of $\alpha(Q_2'')$. Let $D$ be a disk tangent to $\alpha(Q_2'')$ at $q$. Hence, $D$ covers $Q_2''$ and thus $Q_2$. This implies that $D$ is also tangent to $\alpha(Q_2)$ at $q$. Thus, $D$ contains $Q_1'$. On the other hand, because $q$ is involved in the procedure, as discussed above, any disk tangent to $\alpha(Q_2'')$ at $q$ does not contain $Q_1''$. Hence, $D$ does not contain $Q_1''$. Because $Q_1''\subseteq Q_1'$, we obtain that $D$ does not contain $Q_1'$, incurring contradiction. \qed
\end{proof}

This finishes the description of our deletion algorithm, which takes $O(1)$ amortized time.
Note that the second invariant in Observation~\ref{obser:invariant} is established.

\subsection{Insertions}
\label{sec:insert}

Consider an insertion of a point $q_j$ into $Q_2$.
We first update the hull $\alpha(Q_2)$ as in our static algorithm. If $q_j$ is redundant, then we are done for the insertion because $\alpha(Q)$ still exists (by Observation~\ref{obser:invariant}(2)) and the common tangents do not change.
Otherwise, $q_j$ appears as the rightmost vertex in the new $\alpha(Q_2)$ (recall that we have assumed that $\alpha(R)$ exists and thus $\alpha(Q_2)$ always exists).
Let $Q_2'$ be the set of $Q_2$ before $q_j$ is inserted and $Q_2$ the set after the insertion.
Let $Q'=Q_1\cup Q_2'$ and $Q=Q_1\cup Q_2$.
Let $z_1$ and $z_2$ be the counterclockwise and clockwise neighbors of
$q_j$ in the $\alpha(Q_2)$, or equivalently, they are the upper and lower
tangent points from $q_j$ to $\alpha(Q'_2)$.
For a purpose that will be clear later, we temporarily
keep the circular hull of $\alpha(Q'_2)$ unaltered.

Since the $Q_2$-dominating case does not happen, one of the following two
cases will happen: (1) the common tangents of $\alpha(Q_1)$ and
$\alpha(Q_2)$ exist; (2) $\alpha(Q)$ does not exist. Our algorithm
will detect which case happens. In the first case, the algorithm will
find the new common tangents. In the second case,
some further processing that involves deleting points from $Q_1$ will follow (the deletion processing algorithm in Section~\ref{sec:delete} will be invoked).
Before describing our algorithm, we give two lemmas that will help demonstrate the correctness of our algorithm. 
Let $cw(a_1,a_2)$ and $ccw(b_1,b_2)$ be the upper and lower common
tangents of $\alpha(Q_1)$ and $\alpha(Q_2')$, respectively, which are already
known by Observation~\ref{obser:invariant}(2).
We use $\beta(a_2,b_2)$ denote the subset of vertices of $\alpha(Q_2')$ clockwise from $a_2$ to $b_2$ excluding $a_2$ and $b_2$, and $\beta(a_2,b_2)=\emptyset$ if $a_2=b_2$. In fact,  $\beta(a_2,b_2)=\overline{\partial_{\alpha(Q_2')}[b_2,a_2]}$. Let $\beta[a_2,b_2]=\beta(a_2,b_2)\cup\{a_2,b_2\}$.

\begin{lemma}\label{lem:rightmostvertex}
\begin{enumerate}
\item
The rightmost vertex of $\alpha(Q')$ is also the rightmost vertex of $\alpha(Q_2')$, which must be in $\beta[a_2,b_2]$.
\item
The rightmost arc of $\alpha(Q')$ is one of the following three arcs: the rightmost arc of $\alpha(Q_2')$, $cw(a_1,a_2)$, and $ccw(b_1,b_2)$.
\end{enumerate}
\end{lemma}
\begin{proof}
Let $v$ be the rightmost vertex of $\alpha(Q')$. We first show that $v$ must be in $Q_2'$. Assume to the contrary that this is not true. Then, $v\in Q_1$.
Since all points of $Q_2'$ are to the right of $\ell$ and all points of $Q_1$ are to the left of $\ell$, none of the points of $Q_2'$ is a vertex of $\alpha(Q')$, which implies that all points of $Q_2'$ are in $\alpha(Q')$, and thus $\alpha(Q')=\alpha(Q_1)$. Therefore, we obtain that all points of $Q_2'$ are in $\alpha(Q_1)$, which is the $Q_1$-dominating case. This contradicts Observation~\ref{obser:invariant}(1) that the $Q_1$-dominating case does not happen. Hence, $v$ is in $Q_2'$.

Since $Q_2'\subseteq Q'$, it is not difficult to see that $v$ is also the rightmost vertex of $\alpha(Q_2')$.
Since $\beta[a_2,b_2]$ consists of all vertices of $\alpha(Q_2')$ that are also vertices of $\alpha(Q')$, $v$ must be in $\beta[a_2,b_2]$.

The above proves the first statement of the lemma. The second statement follows from $v\in \beta[a_2,b_2]$, which consists of all vertices of $\alpha(Q_2')$ that are also vertices of $\alpha(Q')$.
\qed
\end{proof}


If $q_j$ is in the supporting disk of the rightmost arc of $\alpha(Q')$, i.e., $q_j$ is redundant with respect to $\alpha(Q')$, then $\alpha(Q)$ exists and $cw(a_1,a_2)$ and $ccw(b_1,b_2)$ are still the common tangents of $\alpha(Q_1)$ and $\alpha(Q_2)$. Otherwise, $\alpha(Q)$ exists if and only if the tangents from $q_j$ to $\alpha(Q')$ exist. If $\alpha(Q)$ exists, we use $a$ and $b$ to denote the upper and lower tangent points from $q_j$ to $\alpha(Q')$, respectively.


%
%
%

\begin{figure}[t]
\begin{minipage}[t]{0.49\textwidth}
\begin{center}
\includegraphics[height=2.0in]{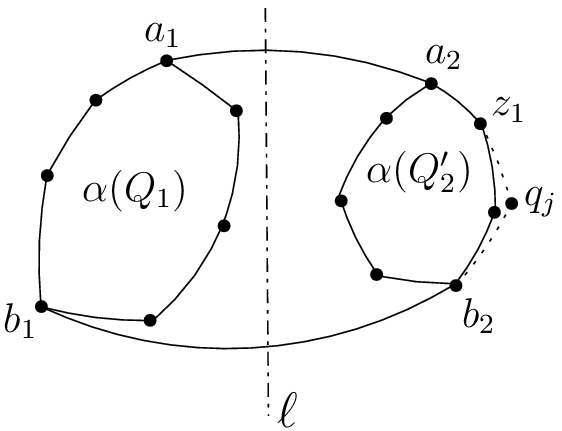}
\caption{\footnotesize Illustrating Lemma~\ref{lem:correct}(1).}
\label{fig:insertionlemma}
\end{center}
\end{minipage}
\hspace{0.05in}
\begin{minipage}[t]{0.49\textwidth}
\begin{center}
\includegraphics[height=2.0in]{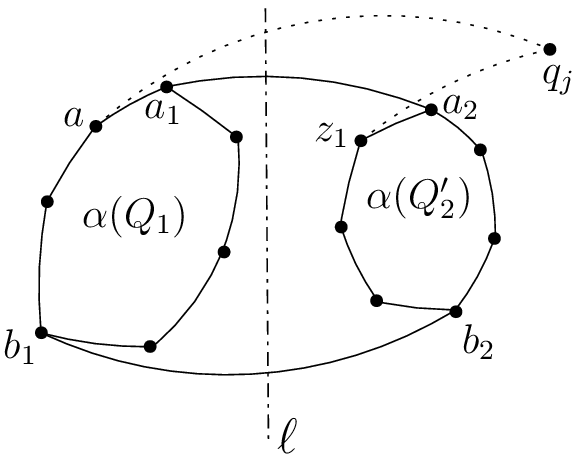}
\caption{\footnotesize Illustrating Lemma~\ref{lem:correct}(2).}
\label{fig:insertionlemma10}
\end{center}
\end{minipage}
\vspace{-0.15in}
\end{figure}

\begin{lemma}\label{lem:correct}
Assume that $q_j$ is not in the supporting disk of the rightmost arc of $\alpha(Q')$ and $\alpha(Q)$ exists. 
\begin{enumerate}
    \item If $z_1\in \beta(a_2,b_2)$, or if $z_1=a_2$ and $cw(a_1,a_2)$ and $q_j$ form an inner turn, then $cw(a_1,a_2)$ is still the upper tangent of $\alpha(Q_1)$ and $\alpha(Q_2)$; e.g., see Fig.~\ref{fig:insertionlemma}.

  \item
  If $z_1\not\in \beta[a_2,b_2]$, or $z_1=a_2$ and $cw(a_1,a_2)$ and $q_j$ form an outer turn, then
     $cw(a,q_j)$ is the new upper common tangent of $\alpha(Q_1)$ and $\alpha(Q_2)$ as well as the upper tangent from $q_j$ to $\alpha(Q_1)$, and further, $a\in \overline{\partial_{\alpha(Q_1)}(a_1,b_1)}$;  e.g., see Fig.~\ref{fig:insertionlemma10}.

  \item If $z_2$ is in $\beta(a_2,b_2)$, or if $z_2=b_2$ and $ccw(b_1,b_2)$ and $q_j$ form an inner turn, then 
      $ccw(b_1,b_2)$ is still the upper tangent of $\alpha(Q_1)$ and $\alpha(Q_2)$.

  \item
  If $z_2\not\in \beta[a_2,b_2]$, or $z_2=b_2$ and $ccw(b_1,b_2)$ and $q_j$ form an outer turn, then
    $ccw(b,q_j)$ is the new lower common tangent of $\alpha(Q_1)$ and $\alpha(Q_2)$ as well as the lower tangent from $q_j$ to $\alpha(Q_1)$, and further, $b\in \overline{\partial_{\alpha(Q_1)}(a_1,b_1)}$.
\end{enumerate}
\end{lemma}
\begin{proof}
We only prove (1) and (2), since (3) and (4) can be proved analogously.

Assume that $z_1\in \beta(a_2,b_2)$. Then, by the definition of $\beta(a_2,b_2)$,
$D(cw(z_1,q_j))$ is tangent to $\alpha(Q')$ at $z_1$, and thus
$cw(z_1,q_j)$ is also the upper tangent from $q_j$ to $\alpha(Q')$ and $z_1=a$. To
show that $cw(a_1,a_2)$ is still the upper common tangent, it suffices
to show that both $a_1$ and $a_2$ are still vertices of $\alpha(Q)$. Assume to
the contrary this is not true. Then, because $z_1\in \beta(a_2,b_2)$,
$cw(z_1,q_j)$ is the upper tangent from $q_j$ to $\alpha(Q')$, and the
rightmost vertex of $\alpha(Q')$ is in $\beta[a_2,b_2]$ by
Lemma~\ref{lem:rightmostvertex}, if we apply the clockwise scanning
procedure on $\alpha(Q')$ to search the lower tangent $ccw(b,q_j)$,
then at least one of $a_1$ and $a_2$ will be removed from the vertex list of $\alpha(Q)$
during procedure. As at least one of $a_1$ and $a_2$ is not a vertex of
$\alpha(Q)$ and the scanning procedure starts from the rightmost vertex of $\alpha(Q_2')$,
$a_1$ cannot be a vertex of $\alpha(Q)$ and $b$ must be in $Q_2'$, and further, $ccw(b,q_j)$ must cross the vertical line $\ell$ twice because both $b$ and $q_j$ are to the right of $\ell$ while $a_1$ is to the left of $\ell$. Hence, $ccw(b,q_j)$ is the leftmost arc of $\alpha(Q)$. In addition, since $b\in Q_2'$, $\alpha(Q)$ is actually $\alpha(Q_2)$, implying that all points of $Q_1$ are in $\alpha(Q_2)$.
Therefore, we obtain that this is the $Q_2$-dominating case, contradicting with  Observation~\ref{obser:invariant}(3) that the $Q_2$-dominating case does not happen after $q_j$ is inserted. Hence, $cw(a_1,a_2)$ is still the upper tangent of $\alpha(Q_1)$ and $\alpha(Q_2)$.

Assume that $z_1=a_2$ and $cw(a_1,a_2)$ and $q_j$ form an inner turn. Then, because by Lemma~\ref{lem:rightmostvertex} the rightmost vertex of $\alpha(Q')$ is also the rightmost vertex of $\alpha(Q_2')$, which is in $\beta[a_2,b_2]$, if we apply the counterclockwise scanning procedure on $\alpha(Q')$ to search the upper tangent from $q_j$ to $\alpha(Q')$, then the procedure will return $cw(z_1,q_j)$, and thus $z_1=a$. Consequently, following the same proof as above, we can show that $cw(a_1,a_2)$ is still the upper tangent of $\alpha(Q_1)$ and $\alpha(Q_2)$.
The proves the lemma statement (1).

Next we prove the lemma statement (2).


Assume $z_1\not\in \beta[a_2,b_2]$.
Consider the counterclockwise scanning procedure on $\alpha(Q'_2)$ for searching $cw(z_1,q_j)$ and the counterclockwise scanning procedure on $\alpha(Q')$ for searching $cw(a,q_j)$. As the rightmost vertex $v$ of $\alpha(Q')$ is also the rightmost vertex of $\alpha(Q_2')$, the two procedures both start from $v$. Further, since $v\in \beta[a_2,b_2]$ and $z_1\not\in \beta[a_2,b_2]$, the counterclockwise scanning procedure on $\alpha(Q'_2)$ for $cw(z_1,q_j)$ will process vertices of $\beta[a_2,b_2]$ counterclockwise from $v$ to $a_2$, after which the counterclockwise neighbor of $a_2$ on $\alpha(Q_2')$ will be processed. This means that the counterclockwise scanning procedure on $\alpha(Q')$ for  $cw(a,q_j)$ will also process vertices of $\beta[a_2,b_2]$ counterclockwise from $v$ to $a_2$, after which the counterclockwise neighbor of $a_2$ on $\alpha(Q')$ will be processed, which is $a_1$.
We claim that $a$ is not in $Q_2$, since otherwise by the similar analysis as above the $Q_2$-dominating case would happen, incurring contradiction. Hence, $a$ is a vertex on $\alpha(Q_1)$. As $cw(a,q_j)$ is the upper tangent of from $q_j$ to $\alpha(Q')$, $D(cw(a,q_j))$ contains $Q'$ and thus $Q_1$. Hence, $D(cw(a,q_j))$ is tangent to $\alpha(Q_1)$ at $a$, and thus $cw(a,q_j)$ is an upper tangent from $q_j$ to $\alpha(Q_1)$. On the other hand, since $D(cw(a,q_j))$ contains $Q'$ and also $q_j$, $cw(a,q_j)$ is an arc of $\alpha(Q)$. Since $a\in Q_1$ and $q_j\in Q_2$, $cw(a,q_j)$ is the upper common tangent of $\alpha(Q_1)$ and $\alpha(Q_2)$.

We next discuss the case where $z_1=a_2$ and $cw(a_1,a_2)$ and $q_j$ form an outer turn. As above, we consider the two counterclockwise scanning procedures. Since $z_1=a_2$, the two procedures will both process vertices on $\beta[a_2,b_2]$ from $v$ until $a_2$. As $cw(a_1,a_2)$ and $q_j$ form an outer turn, according to our counterclockwise searching procedure on $\alpha(Q')$ for $cw(a,q_j)$, when we process $a_2$, we need to further check whether the two neighbors of $a_2$ in $\alpha(Q')$ are both in $\alpha(a_2,q_j)$. We claim that this is not true. Indeed, assume to the contrary that this is true. Then, we obtain that $\alpha(a_2,q_j)=\alpha(Q)$. But this means that the $Q_2$-dominating case happens since both $a_2$ and $q_j$ are in $Q_2$, incurring contradiction. Because the two neighbors of $a_2$ in $\alpha(Q')$ are not both in $\alpha(a_2,q_j)$, according to our counterclockwise searching procedure, we will proceed on processing the counterclockwise neighbor of $a_2$ on $\alpha(Q')$, which is $a_1$. Then, following the same analysis as the above case, we can show that $cw(a,q_j)$ is the upper tangent from $q_j$ to $\alpha(Q_1)$ and also the upper common tangent of $\alpha(Q_1)$ and $\alpha(Q_2)$.


It remains to show that $a\in \overline{\partial_{\alpha(Q_1)}(a_1,b_1)}$.
Since $cw(a,q_j)$ is the upper tangent from $q_j$ to $\alpha(Q')$ and also the upper tangent from $q_j$ to $\alpha(Q_1)$, $a$ must be a vertex of both $\alpha(Q')$ and $\alpha(Q_1)$. Because  $\overline{\partial_{\alpha(Q_1)}(a_1,b_1)}$ consists of all points that are vertices of both $\alpha(Q')$ and $\alpha(Q_1)$, it must contain $a$. This proves the lemma statement (2)
\qed
\end{proof}

In light of Lemma~\ref{lem:correct}, our algorithm works as follows.
We first check whether $q_j$ is in the supporting circle of the rightmost arc of $\alpha(Q')$. By Lemma~\ref{lem:rightmostvertex}, this can be done in constant time. If yes, then $cw(a_1,a_2)$ and $ccw(b_1,b_2)$ are still the common tangents of $\alpha(Q_1)$ and $\alpha(Q_2)$, and we are done with the insertion. In the following, we assume otherwise. Depending on whether $z_1$ satisfies the condition in Lemma~\ref{lem:correct}(1) or Lemma~\ref{lem:correct}(2), and whether $z_2$ satisfies the condition in Lemma~\ref{lem:correct}(3) or Lemma~\ref{lem:correct}(4), there are four cases.

If $z_1$ satisfies Lemma~\ref{lem:correct}(1) and $z_2$ satisfies Lemma~\ref{lem:correct}(3), then $cw(a_1,a_2)$ and $ccw(b_1,b_2)$ are still the common tangents of $\alpha(Q_1)$ and $\alpha(Q_2)$ . So $\alpha(Q)$ exists and we are done with the insertion.

If $z_1$ satisfies Lemma~\ref{lem:correct}(2) and $z_2$ satisfies
Lemma~\ref{lem:correct}(3), then $ccw(b_1,b_2)$ is still the lower
common tangent but $cw(a_1,a_2)$ is not the upper common tangent any
more. This also implies that $\alpha(Q)$ exists.
Next, we find the new upper common
tangent, as follows.
We apply the counterclockwise scanning procedure on $\alpha(Q_1)$ as
in the static algorithm, but it is sufficient for the scanning
procedure to start from $a_1$ (as discussed in the proof of
Lemma~\ref{lem:correct}). As the upper common tangent exists,
this procedure will find it. We call the procedure {\em the
insertion-type upper common tangent searching procedure}. The running
time of the procedure is $O(1+k)$, where $k$ is the number of vertices
of $\alpha(Q_1)$ counterclockwise strictly from $a_1$ to
the new upper tangent point (we say that these vertices are {\em
involved} in the procedure). By the following lemma, the amortized
cost of the procedure is $O(1)$.

\begin{lemma}\label{lem:cost10}
Each point of $L\cup R$ can involve in the insertion-type upper common tangent searching procedure at most once in the entire algorithm.
\end{lemma}
\begin{proof}
Let $v$ be a point involved in the procedure, which is a vertex of $\alpha(Q_1)$. Let $v_1$ and $v_2$ be $v$'s counterclockwise and clockwise neighbors on $\alpha(Q_1)$, respectively. According to our counterclockwise scanning procedure, $cw(v,v_2)$ and $q_j$ form an outer turn, and thus the disk $D(cw(v,v_2))$ does not contain $q_j$, and similarly, $cw(v_1,v)$ and $q_j$ form an outer turn and $D(cw(v_1,v))$ does not contain $q_j$.

We claim that at least one of $v_1$ and $v_2$ are to the right of $v$.
To prove the claim, it is sufficient to show that $v$ is not the
rightmost vertex of $\alpha(Q_1)$. Indeed, since $v$ is involved in
the procedure, $v$ is in $\overline{\partial_{\alpha(Q_1)}[a_1,b_1]}$.
By Observation~\ref{obser:tangentposgeneral}(3), the rightmost vertex
of $\alpha(Q_1)$ is in ${\partial_{\alpha(Q_1)}[a_1,b_1]}$. Therefore,
$v$ is not the rightmost vertex of $\alpha(Q_1)$.
The claim is thus proved.
Without loss of generality, we assume that $v_2$ is to the right of $v$.



We argue that $v$ will not be involved in the same procedure again in future. Assume to the contrary that $v$ is involved in the same procedure again during another insertion of $q_k$, with $k>j$.
Let $Q''_1$,  $Q''_2$, and $Q''$ refer to the corresponding sets right before the insertion.
Since $v$ is involved in the procedure, $v$ has not been deleted
and thus is in $Q''_1$. Since $v_2$ is to the right of $v$, $v_2$ has
also not been deleted and thus is in $Q''_1$ as well. As $cw(v,v_2)$ is an arc of $\alpha(Q_1)$ and $Q_1''\subseteq Q_1$, $cw(v,v_2)$ is also an arc of $\alpha(Q_1'')$.

Let $a_1''$ (resp., $b_1''$) be the tangent point on $\alpha(Q_1'')$ of the upper (resp., lower) common tangent of $\alpha(Q_1'')$ and $\alpha(Q_2'')$.
Since $v$ is involved in the procedure for inserting $q_k$, $v$ must
be a vertex of $\alpha(Q''_1)$ in $\overline{\partial_{\alpha(Q_1'')}[a''_1,b''_1]}$.
As $cw(v,v_2)$ is an arc of $\alpha(Q''_1)$ and $v\in \overline{\partial_{\alpha(Q_1'')}[a''_1,b''_1]}$, $cw(v,v_2)$ must be an arc of $\alpha(Q'')$ and thus
the disk $D(cw(v,v_2))$
must cover $Q''$. Hence, $D(cw(v,v_2))$ covers $Q_2''$. Notice that $q_j$ is in
$Q_2''$, for $j<k$. Therefore, $q_j$ is contained in $D(cw(v,v_2))$. But we have obtained above
that $D(cw(v,v_2))$ does not contain $q_j$. Hence, we obtain contradiction.
\qed
\end{proof}

If $z_1$ satisfies Lemma~\ref{lem:correct}(1) and $z_2$ satisfies Lemma~\ref{lem:correct}(4), then $cw(a_1,a_2)$ is still the upper common tangent but $ccw(b_1,b_2)$ is not the lower common tangent any more.
This is a symmetric case to the above case, and we can apply the clockwise scanning procedure on $\alpha(Q_1)$  (starting from $b_1$) to find the new lower common tangent. We call this  {\em the insertion-type lower common tangent searching procedure}, which takes $O(1)$ amortized time by a similar analysis as Lemma~\ref{lem:cost10}.

If $z_1$ satisfies Lemma~\ref{lem:correct}(2) and $z_2$ satisfies Lemma~\ref{lem:correct}(4), e.g., see Fig.~\ref{fig:insertionlemma20}, then neither $cw(a_1,a_2)$ nor $ccw(b_1,b_2)$ is a common tangent any more.
Indeed, this is the most challenging case.
One reason is that we do not know whether $\alpha(Q)$ exists. Therefore,
our algorithm needs to determine whether $\alpha(Q)$ exists, and if yes, find the new common tangents, which are the tangents from $q_j$ to $\alpha(Q_1)$ by Lemma~\ref{lem:correct}.
Further, if $\alpha(Q)$ does not exist, then our algorithm will find a special vertex $p^*$ on $\alpha(Q_1)$ such that there is no unit disk that can cover $Q_2$ and the points of $Q_1$ to the right of $p^*$ including $p^*$.
As such, before $p^*$ is deleted, $\alpha(Q)$ always does not exist (but $\alpha(Q)$ may still not exist even after $p^*$ is deleted). The following lemma will be useful later.

\begin{lemma}\label{lem:qjtangent}
Assume that $\alpha(Q)$ does not exist.
If $P$ is a subset of $Q_1$ such that $\alpha(P\cup Q_2)$ exists, then there is a unit disk tangent to $\alpha(Q_2)$ at $q_j$ that contains all points of $P\cup Q_2$.
\end{lemma}
\begin{proof}
If $q_j$ is a vertex of $\alpha(P\cup Q_2)$, then by Observation~\ref{obser:basic}(1) there is a disk $D$ with $q_j$ on its boundary and covering $P\cup Q_2$. Since $D$ covers $Q_2$ and has $q_j$ on its boundary, $D$ is tangent to $\alpha(Q_2)$ at $q_j$. This proves the lemma.
Below we show that the case where $q_j$ is not a vertex of $\alpha(P\cup Q_2)$ cannot happen.

Assume to the contrary $q_j$ is not a vertex of $\alpha(P\cup Q_2)$. Then $q_j$ is in the interior of $\alpha(P\cup Q_2)$. Thus, removing $q_j$ from $Q_2$ will not affect $\alpha(P\cup Q_2)$, i.e., $\alpha(P\cup Q_2')=\alpha(P\cup Q_2)$, where $Q_2'=Q_2\setminus\{q_j\}$. Recall that by our algorithm invariant Observation~\ref{obser:invariant}(2), $\alpha(Q_1\cup Q_2')$ exists. Since $P\subseteq Q_1$, $\alpha(P\cup Q_2')\subseteq \alpha(Q_1\cup Q_2')$. Since $q_j$ is in the interior of $\alpha(P\cup Q_2')$, $q_j$ must be in the interior of $\alpha(Q_1\cup Q_2')$, and thus $\alpha(Q_1\cup Q_2')=\alpha(Q_1\cup Q_2'\cup\{q_j\})$. But this implies that
$\alpha(Q)$ exists as $Q=Q_1\cup Q_2'\cup\{q_j\}$, which contradicts with the fact that $\alpha(Q)$ does not exist.
\qed
\end{proof}

We next elaborate on the algorithm.
It is possible that $a_1$ is not in the upper hull or $b_1$ is not in the lower hull of $\alpha(Q_1)$. We first consider the case where $a_1$ is in the upper hull and $b_1$ is in the lower hull; other cases can be handled in a similar (and easier) way and will be discussed later.

\begin{figure}[t]
\begin{minipage}[t]{\textwidth}
\begin{center}
\includegraphics[height=2.0in]{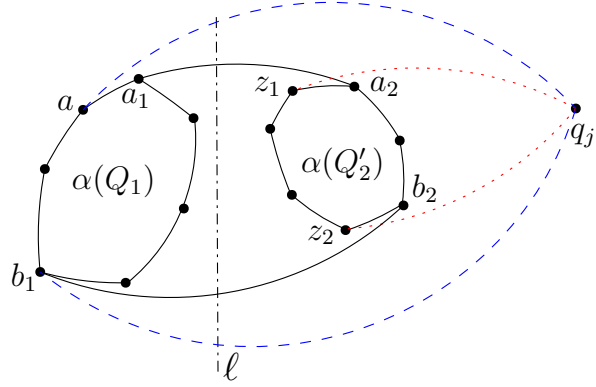}
\caption{\footnotesize Illustrating the case where $z_1$ satisfies Lemma~\ref{lem:correct}(2) and $z_2$ satisfies Lemma~\ref{lem:correct}(4). The two new tangents $cw(a,q_j)$ and $ccw(b,q_j)$ are also shown, with $b=b_1$.}
\label{fig:insertionlemma20}
\end{center}
\end{minipage}
\vspace{-0.15in}
\end{figure}

If $\alpha(Q)$ exists, then as those previous cases, we could find the upper tangent from $q_j$ to $\alpha(Q_1)$ by a counterclockwise scanning procedure on $\alpha(Q_1)$, starting from $a_1$, and similarly, find the lower tangent from $q_j$ to $\alpha(Q_1)$ by a clockwise scanning procedure on $\alpha(Q_1)$, starting from $b_1$.
The two procedures could run independently. However, since we do not know whether $\alpha(Q)$ exists and in the case where $\alpha(Q)$ does not exist we need to find a particular vertex $p^*$,
we will coordinate the two scanning procedures by processing vertices in order of decreasing $x$-coordinate.
Specifically, starting from $p_a=a_1$, we will process $p_a$ and scan $\alpha(Q_1)$ counterclockwise, and simultaneously, starting from $p_b=b_1$, we will process $p_b$ and scan $\alpha(Q_1)$ clockwise, in the same way as the static algorithm. We coordinate the two scanning procedures by the following rule: if $p_a$ is to the right of $p_b$, then we process $p_a$ first; otherwise we process $p_b$ first. In addition, our algorithm maintains the following invariant:
There is a unit disk with $q_j$ on the boundary covering both $z_2$ and $cw(p_a)$, and there is a unit disk with $q_j$ on the boundary covering both $z_1$ and $ccw(p_b)$. For the purpose of describing our algorithm, we temporarily set $cw(a_1)$ to $a_2$ and set $ccw(b_1)$ to $b_2$\footnote{One could consider that we are working on $\alpha(Q')$, and thus $cw(a_1)$ is indeed $a_2$ and $ccw(b_1)$ is indeed $b_2$.}.
The above invariant holds initially when $p_a=a_1$ and $p_b=b_1$, because $cw(p_a)=a_2\in Q_2\subseteq D(cw(q_j,z_2))$ and $ccw(p_b)=b_2\in Q_2\subseteq D(cw(q_j,z_1))$.

%

Without loss of generality, we assume that $p_a$ is to the right of $p_b$.
So we process $p_a$, as follows.
We first check whether there is a unit disk with $q_j$ on the boundary covering both $p_a$ and $z_2$. If not, then we stop the algorithm and return $p^*=p_a$. If yes, we proceed as follows.

We check whether $cw(ccw(p_a),p_a)$ and $q_j$ form an inner turn.
If yes, then $cw(p_a,q_j)$ is the upper tangent from $q_j$ to $\alpha(Q_1)$ and thus is the new upper common tangent by Lemma~\ref{lem:correct}. Then, we proceed to find the lower tangent, which is guaranteed to exist, by running the clockwise scanning procedure.
If it is an outer turn, then we check whether $\alpha(p_a,q_j)$ contains $cw(p_a)$ and $ccw(p_a)$. If yes, then we return $cw(p_a,q_j)$ as the upper common tangent and also return $ccw(p_a,q_j)$ as the lower common tangent.
Otherwise, if $ccw(p_a)$ is to the left of $p_a$ (i.e., $p_a$ is not the leftmost vertex of $\alpha(Q_1)$), then we set $p_a=ccw(p_a)$ and proceed as above; otherwise, we set $p^*=p_a$ and stop the algorithm.

The above describes our algorithm.
For the correctness, in addition to Lemma~\ref{lem:correct}, it is sufficient to show that if the algorithm returns $p^*$, then $p^*$ is correctly computed, as proved in Lemma~\ref{lem:end}.

\begin{lemma}\label{lem:end}
Suppose the algorithm returns $p^*$. Then, there is no unit disk that can cover all points of $Q_2$ and the points of $Q_1$ to the right of $p^*$ including $p^*$.
\end{lemma}
\begin{proof}
Suppose we are processing a vertex $p_a$. There are two ways that $p^*$ is returned: (1) when there is no unit disk with $q_j$ on the boundary covering both $p_a$ and $z_2$; (2) when $p_a$ is the leftmost vertex of $\alpha(Q_1)$ and we still attempt to set $p_a=ccw(p_a)$. In both cases, $p^*=p_a$.
Our goal is to show that $P\cup Q_2$ is not
unit disk coverable, where $P$ is the subset of points of $Q_1$ to the right of $p_a$ including $p_a$.

In the first case,
assume to the contrary that $P\cup Q_2$ are unit disk coverable. Then, by Lemma~\ref{lem:qjtangent}, there is a unit disk with $q_j$ on the boundary covering $P\cup Q_2$. Thus, we obtain contradiction since $p_a\in P$ and $z_2\in Q_2$.

In the second case, we have $P=Q_1$ and $Q=P\cup Q_2$. So it suffices to show that $\alpha(Q)$ does not exist.
Assume to the contrary that $\alpha(Q)$ exists. By Lemma~\ref{lem:correct}, the tangents from $q_j$ to $\alpha(Q_1)$ are the tangents from $q_j$ to $\alpha(Q')$, and $a\in \overline{\partial_{\alpha(Q_1)}(a_1,b_1)}$.
According to our algorithm, $a$ cannot be a vertex of $\alpha(Q_1)$ counterclockwise from $a_1$ to $p_a$.
Thus, $a$ is a vertex of $\alpha(Q_1)$ counterclockwise from $ccw(p_a)$ to $b_1$.
Further, $a$ must be a vertex $\alpha(Q_1)$ counterclockwise from $a_1$ to $b$.
On the other hand, since $p_a$ is the leftmost vertex of $\alpha(Q_1)$ and $p_a$ is currently being processed, it must be the case that $p_b=p_a$ and $v$ has already been processed, where $v=ccw(p_b)$.
This means that $b$ cannot be a vertex of $\alpha(Q_1)$ counterclockwise from $v$ to
$b_1$, and thus $b$ must be a vertex of $\alpha(Q_1)$ counterclockwise from $a_1$ to $p_a$.
Since $a$ is a vertex $\alpha(Q_1)$ counterclockwise from $a_1$ to $b$, we obtain that $a$ must be a vertex of $\alpha(Q_1)$ counterclockwise from $a_1$ to $p_a$.
But this contradicts with that $a$ is a vertex of $\alpha(Q_1)$ counterclockwise from $ccw(p_a)$ to $b_1$.
\qed
\end{proof}

As in the third case, we also call the above algorithm {\em the insertion-type common tangent points searching procedure}, and its runtime is $O(1+k)$ time, where $k$ is the number the vertices of $\alpha(Q_1)$ counterclockwise strictly from $a_1$ to the final position of $p_a$ when the algorithm stops and the number of vertices $\alpha(Q_1)$ clockwise strictly from $b_1$ to the final position of $p_b$ when the algorithm stops (we say that those vertices are involved in the procedure).
We can use literally the same proof as Lemma~\ref{lem:cost10} to show that each point of $L\cup R$ can involve in the procedure at most once in the entire algorithm. In fact, the proof of Lemma~\ref{lem:cost10} shows that each point of $L\cup R$ can involve in the insertion-type common tangent points searching procedure in both this case and the above third case at most once in the entire algorithm. Hence, the amortized cost is $O(1)$.


%
%
%

One of the following cases happens after the above algorithm: (1) the two common tangents of $Q_1$ and
$Q_2$ are found; (2) a vertex $p^*$ (which is either $p_a$ or $p_b$) is returned.
In the first case, we are done with the insertion, and Observation~\ref{obser:invariant}(2) is established.
In the second case, $\alpha(Q)$ does not exist and we
further perform the following processing.
Without loss of generality, we assume that $p^*=p_a$.
According to our algorithm, $p_b$ is either $p_a$ or to the left of $p_a$, and $ccw(p_b)$ must be to the right of $p_b$ because it was processed before $p_a$.

We perform deletions to delete points from $Q_1$ in order from left to right until $p_a$. By the definition of $p^*$, after each deletion except the last deletion of $p_a$, $\alpha(Q)$ does not exist. Note that these deletions actually have not been invoked yet, so we perform them ahead of time in the sense that when they are actually invoked in future we know that $\alpha(Q)$ does not exist.

To process these deletions efficiently, the key idea is that we
process the deletions by pretending $q_j$ has not been inserted yet,
or equivalently, we process the deletions with respect to $Q_2'$.
Because $\alpha(Q')$ exists before any deletion, we know that it
still exists after each deletion. After all these deletions are
completed, we will insert $q_j$ again (by ``resuming'' our previous work on processing the insertion; see below for the details).
This is the reason we temporarily kept the circular hull
$\alpha(Q_2')$ unaltered before.

We again assume that the $Q_2$-dominating case does not happen (with
respect to $Q_2'$) after each deletion, which can be determined by our
$Q_2$-dominating case detection procedure by changing $j^*$ back to
its value before $q_j$ was inserted. Note that we also need to store
the current value $j^*$ in another variable so that when we resume processing
the insertion of $q_j$ again (which will be discuss below) we simply
reset $j^*$ to that value, which only introduces a constant time.

For each deletion, we update the common tangents of $\alpha(Q_1)$ and
$\alpha(Q_2')$ by using the algorithm in Section~\ref{sec:delete}.
Once $p_a$ is deleted, we insert $q_j$ again by ``resuming'' our
previous work of the insertion of $q_j$, as follows. Let $Q_1'$ refer
to the set of $Q_1$ after $p_a$ is deleted.
Let $cw(a_1',a_2')$ and $ccw(b_1',b_2')$ be the common tangents of $\alpha(Q_1')$ and $\alpha(Q_2')$. Let $\beta(a'_2,b'_2)$ denote the set of vertices $\alpha(Q_2')$ clockwise from $a_2'$ to $b_2'$ excluding $a_2'$ and $b_2'$, and $\beta(a_2',b_2')=\emptyset$ if $a_2'=b_2'$. Let $\beta[a_2',b_2']=\beta(a_2',b_2')\cup\{a_2',b_2'\}$.
Recall that $p_a$ and $p_b$ refer to the vertices of $\alpha(Q_1)$ when our earlier algorithm for processing the insertion of $q_j$ stops (and returns $p^*$).
Depending on whether $p_a=a_1$ and whether $p_b=b_1$, there are four cases.

\begin{itemize}
  \item If $p_a\neq a_1$ and $p_b\neq b_1$, then $cw(p_a)$ is to the
  left of or at $a_1$ and $ccw(p_b)$ is to the left of or at $b_1$. In this case, $cw(a_1,a_2)$ and
  $ccw(b_1,b_2)$ are still the common tangents of $\alpha(Q'_1)$ and $\alpha(Q_2')$, i.e., $cw(a_1,a_2)=cw(a'_1,a'_2)$ and $cw(b_1,b_2)=cw(b'_1,b'_2)$. So $\beta(a_2',b_2')=\beta(a_2,b_2)$. If we apply the same algorithm as before for
  processing the insertion of $q_j$, we are still at the fourth case, i.e., $z_1$ satisfies Lemma~\ref{lem:correct}(2) and $z_2$ satisfies Lemma~\ref{lem:correct}(4).
  But the crux of the idea is that instead of starting over the two scanning procedures from $a_1$ and $b_1$, respectively, we ``resume'' the previous work by starting the counterclockwise scanning procedure from $cw(p_a)$ on $\alpha(Q_1')$
  and starting the clockwise scanning procedure from $ccw(p_b)$ on $\alpha(Q_1')$.
  In this way, we avoid processing a vertex twice except $cw(p_a)$ and $ccw(p_b)$,
  for which we can charge the time to the deletion of $p_a$.

  \item If $p_a= a_1$ but $p_b\neq b_1$, then $ccw(p_b)$ is to the left of or at $b_1$ and $cw(b_1,b_2)$ is still the lower common tangent of $\alpha(Q'_1)$ and $\alpha(Q_2')$, i.e., $cw(b_1,b_2)=cw(b'_1,b'_2)$, but the upper one changes, i.e., $cw(a_1,a_2)\neq cw(a'_1,a'_2)$. Consequently, it is possible that $z_1$ now satisfies Lemma~\ref{lem:correct}(1), which can be determined when we process the deletion of $p_a$. We resume the same algorithm as before for the insertion of $q_j$, i.e., regardless of which case happens, when we search the lower common tangent point on $\alpha(Q_1')$ by running the clockwise scanning procedure, we start from $ccw(p_b)$. However, in the counterclockwise scanning procedure for searching the upper common tangent point, we need to start from the new upper tangent point $a_1'$ because $a_1$ has been deleted.

  \item If $p_a\neq a_1$ but $p_b= b_1$, then this is symmetric to the above second case. We start the clockwise scanning procedure from $b_1'$ and start the counterclockwise scanning procedure from $cw(p_a)$.

  \item If $p_a= a_1$ and $p_b= b_1$, then both $a_1$ and $b_1$ have been deleted since $p_a$ is deleted and  $p_b$ is either $p_a$ or to the left of $p_a$. Hence, both upper and lower common tangents get changed, i.e., $cw(a_1,a_2)\neq cw(a'_1,a'_2)$ and $cw(b_1,b_2)\neq cw(b'_1,b'_2)$. We start the new algorithm exactly the same as before, i.e., start the two scanning procedures from $a_1'$ and $b_1'$, respectively.
\end{itemize}

Other than the time for computing the new common tangents after each deletion (whose amortized time is $O(1)$ as shown in Section~\ref{sec:delete}), the amortized cost of processing the insertion of $q_j$ is $O(1)$. After the above processing, if $\alpha(Q_1'\cup Q_2)$ exists, then we are done with the insertion of $q_j$ (and Observation~\ref{obser:invariant}(2) is established). Otherwise, the algorithm will return a new vertex $p^*$ and we will repeat the same algorithm. As more and more points are deleted from $Q_1'$, eventually we will encounter a situation where $\alpha(Q_1'\cup Q_2)$ exists since $\alpha(Q_2)$ exists (e.g., when all points of $Q'_1$ are deleted).

Recall that the above algorithm is for the situation where $a_1$ is on the upper hull and $b_1$ is on the lower hull of $\alpha(Q_1)$. If this is not the case, then $a_1$ and $b_1$ are either both on the upper hull or both on the lower hull. Without loss of generality, assume that they are both on the upper hull. Then, we can change the algorithm in the following way. We only perform the counterclockwise scanning procedure on the upper hull, starting from $p_a=a_1$. The algorithm for processing each vertex is the same as before except the following: if $p_a$ arrives at $b_1$ and we still want to set $p_a=ccw(p_a)$, then we stop the algorithm and return $p^*=p_a$.
If the procedure finds the new upper common tangent, then the lower common tangent exists and we find it by running the clockwise scanning procedure starting from $b_1$. If the procedure returns $p^*$, then we perform deletions as above until $p_a$. Note that the lower common tangent must get changed, i.e., $cw(b_1,b_2)\neq cw(b'_1,b'_2)$, because $b_1$ is to the left of $p^*$ and thus must be deleted. So we run into either the third or the four case as above (i.e., the two cases with $p_b= b_1$). The correctness is still based on Lemma~\ref{lem:correct} and a similar proof for Lemma~\ref{lem:end}.
The amortized cost analysis of Lemmas~\ref{lem:cost10} still applies.

\subsection{Adapting the Algorithm to the Radially Sorted Case}
\label{sec:radial}

The above gives our algorithm in the problem setting where points in $L\cup R$ are sorted from left to right. We show that we can adapt the algorithm easily to the radially sorted case where points in $L\cup R$ are radially sorted around a point $o$ such that $L$ and $R$ are still separated by a line $\ell$ through $o$ (this is actually our original problem setting on $S=S^+\cup S^-$).

Without loss of generality, we assume that $\ell$ is vertical, and
$L=\{q_1,\ldots,q_n\}$ and $R=\{p_1,\ldots,p_n\}$ are sorted
clockwise around $o$ such that all points of $L$ are to the left of
$\ell$ and all points of $R$ are to the right of $\ell$. We first
discuss how to update the circular hull of $Q_2$ under insertions when
$Q_1=\emptyset$ (i.e., extending the static algorithm to the radially
sorted case). We still consider the points of
$R$ following their index order. To handle each insertion of $q_j$, we
still run a counterclockwise scanning procedure to find the upper
tangent from $q_j$ to the current $\alpha(Q_2)$ and a clockwise
scanning procedure to find the lower tangent. Recall that our previous
algorithm starts the two procedure from the rightmost vertex of
$\alpha(Q_2)$. Here, the difference is that we start the two
procedures from the vertex $v$, where $v$ has the largest index among
all vertices of $\alpha(Q_2)$. This will be consistent with our previous
algorithm. Indeed, because the points of $R$ are right of $\ell$ and radially sorted
around $o$, all vertices of $\alpha(Q_2)$ are on one side of the line
$l$ through $o$ and $v$ while $q_j$ is on the other side of $l$.
Based on Observation~\ref{obser:tangentpos}(4), searching the two tangents from
$v$ will be successful, and we can use the similar analysis to prove the correctness of
this adapted algorithm.  Note
that this requires our algorithm to keep track of the vertex
of the largest index of $\alpha(Q_2)$,
which only introduces $O(n)$ overall time for all points of $R$,
just like in the previous algorithm where we need to keep track of the
rightmost vertex (which is actually also the vertex of the largest
index in the previous problem setting where points of $R$ are sorted
from left to right; this means that if we describe the algorithm as maintaining the vertex of $\alpha(Q_2)$ with the largest index then the same algorithm works on both problem settings without any change).

Further, exactly the same as before,
we maintain the leftmost arc of $\alpha(Q_2)$ after each insertion.
This is for handling the case where $Q_1\neq\emptyset$. We still need
the leftmost arc because $Q_1$ and $Q_2$ are still separated by a
vertical line, in the same way as before, so we can use the same
method as before to handle the interactions between $\alpha(Q_1)$
and $\alpha(Q_2)$, such as computing their common tangents, determining
dominating cases, etc. For example,
when the common tangents of $\alpha(Q_1)$ and $\alpha(Q_2)$ exist,
after $q_j$ is inserted, we need to update the common tangents. To
this end, we first compute the two tangent points $z_1$ and $z_2$ from
$q_j$ to $\alpha(Q_2)$ in the way described above, and then we follow exactly
the same algorithm as before, i.e.,, there are four cases depending the
locations of $z_1$ and $z_2$ with respect to Lemma~\ref{lem:correct}.

For computing $\alpha(Q_1)$ initially when $Q_1=L$, we consider the points of $L$ in the inverse index order, in a similar way as the above for $R$, but now we also need to associate stacks with vertices as in the previous algorithm. The rest of the algorithm follows the same as before.

In summary, we can solve the dynamic circular hull problem on $S=S^+\cup S^-$ in $O(n)$ time, and thus Theorem~\ref{theo:decisionserial} is proved.

\section{Computing Common Tangents of Two Circular Hulls in $O(\log n)$ Time}
\label{sec:commontangent}

In this section, we prove Lemma~\ref{lem:commontangent}. Without loss of generality, let $|L|=|R|=n$ and assume that $L$ and $R$ are separated by a vertical line $\ell$ with $L$ on the left side. Let $\alpha_1$ and $\alpha_2$ denote the circular hulls of $L$ and $R$, respectively. Also, we assume that the vertices of $\alpha_1$ in {\em counterclockwise} order starting from the {\em rightmost} vertex $c_1$ of $\alpha_1$ are stored in a balanced binary search tree $T_1$, and each vertex of $\alpha_1$ is associated with its two neighbors (so that given a node of $T_1$ storing a vertex $v$ of $\alpha_1$ we can access $cw(v)$ and $ccw(v)$ in $O(1)$ time). Similarly,
vertices of $\alpha_2$ in {\em clockwise} order starting from the {\em leftmost} vertex $c_2$ of $\alpha_2$ are stored in another balanced binary search tree $T_2$.

In the following, we present an $O(\log n)$ time algorithm for Lemma~\ref{lem:commontangent}, i.e., determine whether $\alpha(L\cup R)$ exists; if yes, then determine whether the $L$-dominating case or the $R$-dominating case happens; if neither dominating case happens, then compute the two common tangents of $\alpha_1$ and $\alpha_2$. Our algorithm is similar in spirit to the binary search algorithm given by Overmars and van Leeuwen~\cite{ref:OvermarsMa81} for finding common tangents of two convex hulls separated by a line, but the technical crux is in finding the criteria on which the binary search is based.

\subsection{A Special Case}
\label{sec:special}

We first consider a special case where $R$ has only one point $q$, but $L$ has $n$ vertices.
We first check whether the $L$-dominating case happens, by checking whether $q$ is in the supporting disk of the rightmost arc of $\alpha_1$. Using $T_1$, the rightmost arc can be found in $O(\log n)$ time. In the following, we assume that $q$ is outside the disk. Next, we will determine whether $\alpha(L\cup\{q\})$ exists, and if yes, find the two tangents from $q$ to $\alpha_1$. To this end, we first assume that the tangents exist and give an algorithm to find them. Later we will show that the algorithm can be slightly modified to determine whether the tangents exist (i.e., whether $\alpha(L\cup\{q\})$ exists).

We only show how to find the upper tangent point $a$, and the lower tangent point can be found in a similar way.
If we order the vertices of $\alpha_1$ counterclockwise starting from $c_1$ as a sequence $\calL_1$, then we partition  the sequence into three subsequences: $A, B, C$, defined as follows. If $c_1\neq a$, then $A$ consists of all vertices from $c_1$ to $cw(a)$; otherwise $A=\emptyset$. $B=\{a\}$. $C$ consists of the rest of vertices.
By Observation~\ref{obser:tangent}, a vertex $v$ of $\alpha_1$ is $a$ if and only if $D(cw(v,q))$ contains both $cw(v)$ and $ccw(v)$.
Lemma~\ref{lem:binarysearch} provides a criteria on which our binary search algorithm is based to find $a$.

\begin{lemma}\label{lem:binarysearch}
Assume that $a\neq c_1$.
Consider a vertex $v\in A\cup C$. If $v=c_1$, then $v\in A$. Otherwise,
$v$ is in $A$ if and only if the four vertices $cw(v),v,c_1,ccw(c_1)$ are all in $D(cw(v,q))$ or all in $D(cw(c_1,q))$.
\end{lemma}
\begin{proof}
If $v=c_1$, then since $c_1\neq a$ and $c_1$ is the first vertex of $\calL_1$, $v$ must be in $A$.

Assume that $v$ is in $A\setminus\{c_1\}$. We show that the four points $cw(v),v,c_1,ccw(c_1)$ are all in $D(cw(v,q))$ or all in $D(cw(c_1,q))$.

We first give an {\em observation}: for any subsequence $F$ of $\calL_1$, $F$ is the cyclic sequence of all vertices on the circular hull $\alpha(F)$ of $F$. To see this, let $w$ be an arc of $\alpha_1$ connecting two adjacent vertices of $F$. Then $D(w)$ contains all vertices of $\alpha_1$, and thus it covers $F$. Therefore, by Observation~\ref{obser:basic}(2), $w$ is also an arc of $\alpha(F)$. Hence, the arc set of $\alpha(F)$ consists of all arcs of $\alpha_1$ connecting all pairs of adjacent vertices of $F$ plus another arc connecting the first vertex and the last vertex of $F$.

Let $F$ be the subsequence of $\calL_1$ from $c_1$ to $v$. By the above observation, $F$ is the vertex set of $\alpha(F)$. Recall our counterclockwise scanning procedure for finding $a$ in our static algorithm in Section~\ref{sec:static}, which starts from $c_1$. When a vertex $v'$ is processed, the result only depends on the two neighbors of $v'$. Hence, if we run our counterclockwise scanning procedure on both $\alpha_1$ and $\alpha(F)$, the result of the algorithm after processing a vertex $v'$ is the same for any $v'\in F\setminus\{c_1,v\}$. However, when $v'$ is $c_1$ or $v$, the result of processing $v'$ may be different as one of its neighbors gets changed from $\alpha_1$ to $\alpha(F)$. As each vertex of $F\setminus\{c_1,v\}$ is not a tangent point from $q$ to $\alpha_1$ (because $v\in A\setminus\{c_1\}$), it is not a tangent point from $p$ to $\alpha(F)$ either. Hence, the upper tangent point from $q$ to $\alpha(F)$ is either $c_1$ or $v$. If it is $c_1$, then $D(cw(c_1,q))$ covers $F$; otherwise, $D(cw(v,q))$ covers $F$.
Notice that all four points $cw(v),v,c_1,ccw(c_1)$ are in $F$. Thus, either $D(cw(c_1,q))$ or $D(cw(v,q))$ contains all the four points.


Now assume that $v$ is in $C$. We show that neither $D(cw(v,q))$ nor $D(cw(c_1,q))$ contains all  four points $cw(v),v,c_1,ccw(c_1)$, which will prove the lemma.

By the definition of $C$, $v\neq a$. Let $F$ be the subsequence of $\calL_1$ from $a$ to $c_1$. According to the above observation, $F$ is the cyclic sequence of vertices of $\alpha(F)$. Thus, $cw(v)$ and $c_1$ are the two neighbors of $v$ in $\alpha(F)$, and $v$ and $ccw(c_1)$ are two neighbors of $c_1$ in $\alpha(F)$.
Assume to the contrary that either $D(cw(v,q))$ or $D(cw(c_1,q))$ contains all four points $cw(v),v,c_1,ccw(c_1)$. We obtain contradiction below for either case.

In the first case (i.e., $D(cw(v,q))$ contains all four points), since $cw(v)$ and $c_1$ are the two neighbors of $v$ in $\alpha(F)$ and both of them are in $D(cw(v,q))$, $D(cw(v,q))$ is tangent to $\alpha(F)$ at $v$. Thus, $cw(v,q)$ is the upper tangent from $q$ to $\alpha(F)$. We claim that $a=v$. Indeed, since $v\in C$, $F$ contains $a$ by the definition of $F$. Because $cw(a,q)$ is the upper tangent from $q$ to $\alpha_1$, $D(cw(a,q))$ contains all vertices of $\alpha_1$ and thus covers $F$. Hence, $cw(a,q)$ is the upper tangent from $q$ to $\alpha(F)$ and $a$ is the tangent point. Thus, it holds that $v=a$. However, this contradicts with that $v\in C$.

In the second case, since $v$ and $ccw(c_1)$ are the two neighbors of $c_1$ in $\alpha(F)$ and both of them are in $D(cw(c_1,q))$, $D(cw(c_1,q))$ is tangent to $\alpha(F)$ at $c_1$. Thus, $cw(c_1,q)$ is the upper tangent from $q$ to $\alpha(F)$. Following the same analysis as above, we can show that $c_1=a$. However, this contradicts with that $a\neq c_1$.
\qed
\end{proof}

In light of Lemma~\ref{lem:binarysearch}, we can compute $a$ in $O(\log n)$ time using the tree $T_1$, as follows. First, we check whether $c_1$ is $a$, which can be done in constant time after $c_1$ is accessed in $O(\log n)$ time from $T_1$. If not, let $v$ be the vertex of $\alpha_1$ at the root of $T_1$. We check whether $v=a$ in $O(1)$ time. If yes, we stop the algorithm. Otherwise, we check whether $v\in A$ using Lemma~\ref{lem:binarysearch}. If yes, then we proceed on the right child; otherwise we proceed on the left child. The running time is $O(\log n)$, which is the height of $T_1$.
The lower tangent from $q$ to $\alpha_1$ can be found likewise.

The above algorithm finds the tangents if they exist. If we do not know whether they exist, then we slightly change the algorithm as follows. Whenever we check whether a vertex $v$ is the tangent point, we also check whether $v$ and $q$ can be covered by a unit disk. If not, then no tangents exist and we stop the algorithm; otherwise we proceed in the same way as before. But if we reach a leaf $v$ and $v$ is still not the tangent point, then no tangents exist. The time of the algorithm is still $O(\log n)$.

\subsection{The General Case}

In the following, we discuss the general case where $L$ and $R$ each have $n$ vertices. Our algorithm begins with checking whether a dominating case happens in the following lemma.

\begin{lemma}
Whether the $L$-dominating case (resp., the $R$-dominating case) happens can be determined in $O(\log n)$ time.
\end{lemma}
\begin{proof}
We only show how to determine whether the $R$-dominating case happens, and the other case is similar.  Recall that the $R$-dominating case refers to the case where $L$ is covered by the supporting disk $D$ of the leftmost arc of $\alpha_2$, which is true if and only if all vertices of $\alpha_1$ are in $D$ by Observation~\ref{obser:basic}(4). We first check whether the leftmost arc of $\alpha_2$ is $null$. If yes, then the case does not happen. Otherwise, we have the disk $D$ and proceed as follows.

Let $v$ be the vertex at the root of $T_1$.
The vertex $v$ and the rightmost vertex $c_1$ of $\alpha_1$ partition the boundary of $\alpha_1$ into two chains with a roughly equal number of vertices. We check whether both $v$ and $c_1$ are in $D$. If not, then the $R$-dominating case does not happen and we stop the algorithm. Otherwise, by Lemma 4.6 of~\cite{ref:HershbergerFi91}, one of the chains of $\alpha_1$ partitioned by $v$ and $c_1$ is entirely in $D$, and that chain can be determined in $O(1)$ time by knowing the neighbors of $v$ and $c_1$. If the chain counterclockwise from $c_1$ to $v$ is in $D$, then we go to the right child of $v$, i.e., working on the other chain recursively; otherwise, we go to the left child of $v$. If we reach a leaf $v$, then the $R$-dominating case happens if and only if $v\in D$. Clearly, the runtime of the algorithm is $O(\log n)$. \qed
\end{proof}

In the following, we assume that neither dominating case happens. Our goal is to determine whether $\alpha(L\cup R)$ exists, and if yes, compute the two common tangents of $\alpha_1$ and $\alpha_2$. We first show how to find the common tangents by assuming that $\alpha(L\cup R)$ exists.
We follow the binary search scheme of Overmars and van Leeuwen~\cite{ref:OvermarsMa81} for convex hulls but resort to the criteria in Lemma~\ref{lem:binarysearch}.

With respect to any vertex $q$ of $\alpha_2$, we define three sets of vertices of $\alpha_1$: $A,B,C$ in the same way as in Section~\ref{sec:special}. We further partition $C$ into two subsets: $C_1$ and $C_2$ as follows. A vertex $v\in C$ is in $C_1$ if $v$ is on $\alpha_1$ counterclockwise from $a$ to $b$, where $a$ and $b$ are the upper and lower tangent points from $q$ to $\alpha_1$, respectively. Let $C_2=C\setminus C_1$. Note that $C_1=\emptyset$ if $a=b$, for $a\not\in C$. By Observation~\ref{obser:tangent}, a vertex $v\in C$ is in $C_1$ if and only if there is a unit disk $D$ tangent to $\alpha_1$ at $v$ containing $q$, which can be determined in $O(1)$ time given the two neighbors of $v$.
A vertex $p$ of $\alpha_1$ is called an {\em $E$-vertex with respect to $q$} if $p\in E$ for any $E\in \{A,B,C,C_1,C_2\}$.

Symmetrically, with respect to a vertex $p\in \alpha_1$, we also define $E$-vertices of $\alpha_2$ following the {\em clockwise} order from the {\em leftmost} vertex $c_2$ of $\alpha_2$, for $E\in \{A,B,C,C_1,C_2\}$. For a pair of vertices $(p,q)$ with $p\in \alpha_1$ and $q\in \alpha_2$, we say that the pair is an {\em $(E,F)$ case} if $p$ is an $E$-vertex of $\alpha_1$ with respect to $q$ and $q$ is an $F$-vertex of $\alpha_2$ with respect to $p$, with $E,F\in \{A,B,C,C_1,C_2\}$.

We describe an algorithm to compute the upper common tangent $cw(a_1,b_1)$ with $a_1$ and $b_1$ as the tangent points on $\alpha_1$ and $\alpha_2$, respectively. Suppose $p$ and $q$ are vertices at the roots of $T_1$ and $T_2$, respectively.
Depending on whether $(p,q)$ is an $(E,F)$ case, for $E,F\in \{A,B,C\}$, there are nine cases (several subcases arise for the case $(C,C)$). We show below that in each case we can disregard half of the remaining vertices of either $\alpha_1$ or $\alpha_2$. Let $\calL_1$ be the list of vertices of $\alpha_1$ following their order in $T_1$, i.e.,  counterclockwise from $c_1$. Let $\calL_2$ be the list of vertices of $\alpha_2$ following their order in $T_2$, i.e., clockwise from $c_2$.
We discuss the nine cases in order corresponding to those in~\cite{ref:OvermarsMa81}, as follows.

\begin{enumerate}
\item
Case $(B,B)$, which corresponds to Case a. in~\cite{ref:OvermarsMa81}; e.g., see Fig.~\ref{fig:caseBB}. In this case, $a_1=p$ and $b_1=q$. We can stop the algorithm.

\begin{figure}[h]
\begin{minipage}[t]{0.49\textwidth}
\begin{center}
\includegraphics[height=1.8in]{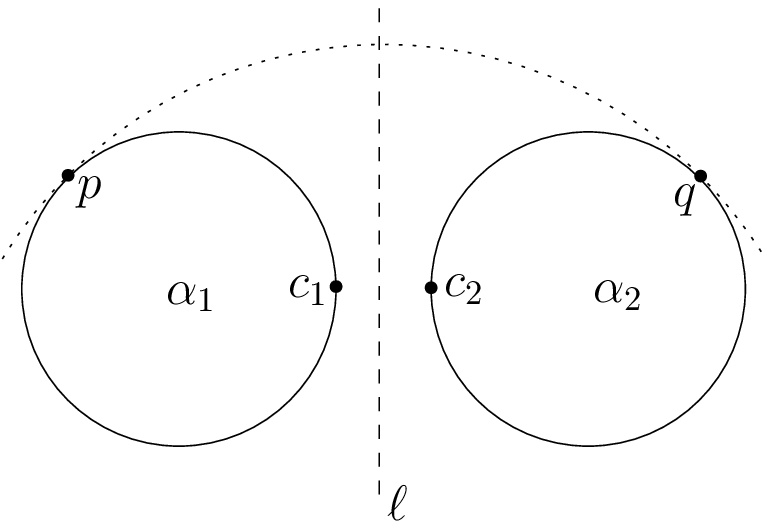}
\caption{\footnotesize Illustrating the case $(B,B)$.}
\label{fig:caseBB}
\end{center}
\end{minipage}
\hspace{0.05in}
\begin{minipage}[t]{0.49\textwidth}
\begin{center}
\includegraphics[height=1.8in]{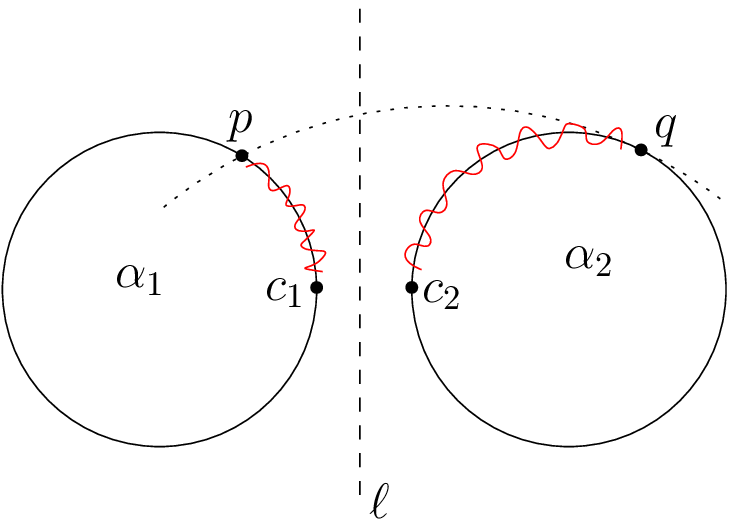}
\caption{\footnotesize Illustrating the case $(A,B)$.}
\label{fig:caseAB}
\end{center}
\end{minipage}
\vspace{-0.15in}
\end{figure}

\item
Case $(A,B)$, which corresponds to Case b. in~\cite{ref:OvermarsMa81} (with the notation $p$ and $q$ switched; the same applies below); e.g., see Fig.~\ref{fig:caseAB}. In this case, the part of $\calL_1$ before $p$ and the part of $\calL_2$ before $q$ can be disregarded, i.e., we move $p$ to its right child and move $q$ to its right child.

\item
Case $(C,B)$, which corresponds to Case c. in~\cite{ref:OvermarsMa81}; e.g., see Fig.~\ref{fig:caseCB}. In this case, the part of $\calL_1$ after $p$ and the part of $\calL_2$ before $q$ can be disregarded,
i.e., we move $p$ to its left child and move $q$ to its right child.

\begin{figure}[h]
\begin{minipage}[t]{0.49\textwidth}
\begin{center}
\includegraphics[height=1.8in]{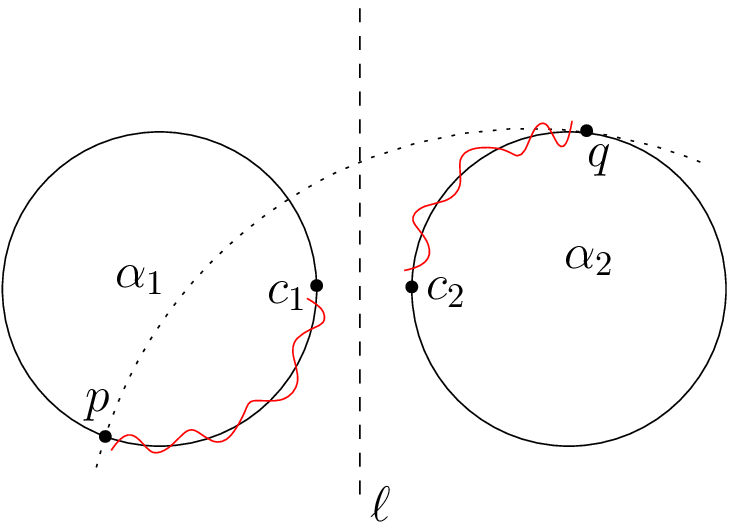}
\caption{\footnotesize Illustrating the case $(C,B)$.}
\label{fig:caseCB}
\end{center}
\end{minipage}
\hspace{0.05in}
\begin{minipage}[t]{0.49\textwidth}
\begin{center}
\includegraphics[height=1.8in]{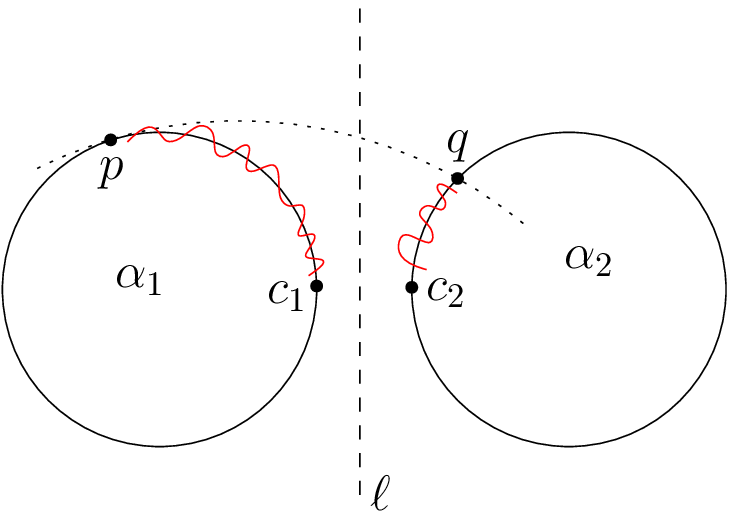}
\caption{\footnotesize Illustrating the case  $(B,A)$.}
\label{fig:caseBA}
\end{center}
\end{minipage}
\vspace{-0.15in}
\end{figure}

\item
Case $(B,A)$, which corresponds to Case d. in~\cite{ref:OvermarsMa81}; e.g., see Fig.~\ref{fig:caseBA}. In this case, the part of $\calL_1$ before $p$ and the part of $\calL_2$ before $q$ can be disregarded.

\item
Case $(B,C)$, which corresponds to Case e. in~\cite{ref:OvermarsMa81}; e.g., see Fig.~\ref{fig:caseBC}. In this case, the part of $\calL_1$ before $p$ and the part of $\calL_2$ after $q$ can be disregarded.

\begin{figure}[h]
\begin{minipage}[t]{0.49\textwidth}
\begin{center}
\includegraphics[height=1.8in]{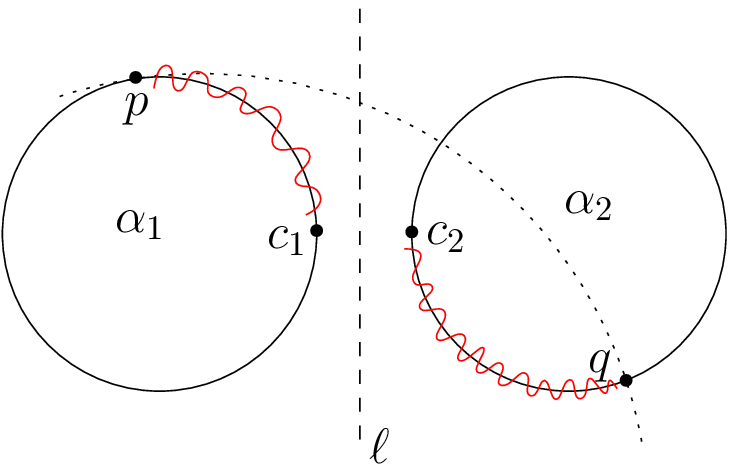}
\caption{\footnotesize Illustrating the case $(B,C)$.}
\label{fig:caseBC}
\end{center}
\end{minipage}
\hspace{0.05in}
\begin{minipage}[t]{0.49\textwidth}
\begin{center}
\includegraphics[height=1.8in]{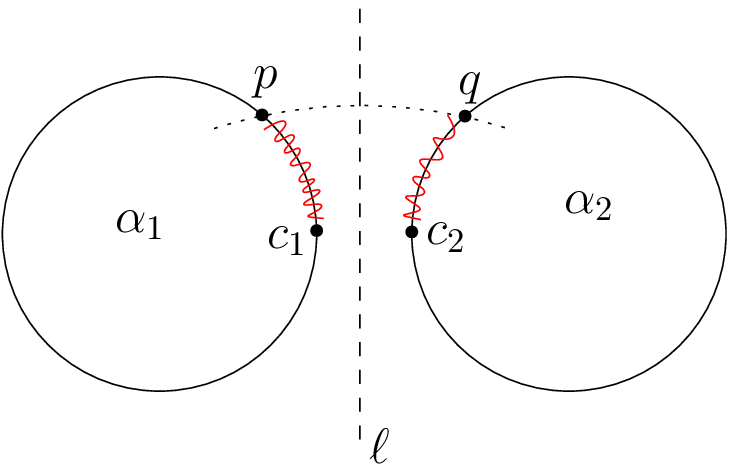}
\caption{\footnotesize Illustrating the case $(A,A)$.}
\label{fig:caseAA}
\end{center}
\end{minipage}
\vspace{-0.15in}
\end{figure}

\item
Case $(A,A)$, which corresponds to Case f. in~\cite{ref:OvermarsMa81}; e.g., see Fig.~\ref{fig:caseAA}. In this case, the part of $\calL_1$ before $p$ and the part of $\calL_2$ before $q$ can be disregarded.

\item
Case $(A,C)$, which corresponds to Case g. in~\cite{ref:OvermarsMa81}; e.g., see Fig.~\ref{fig:caseAC}. In this case, the part of $\calL_1$ before $p$ can be disregarded.

\begin{figure}[h]
\begin{minipage}[t]{0.49\textwidth}
\begin{center}
\includegraphics[height=1.8in]{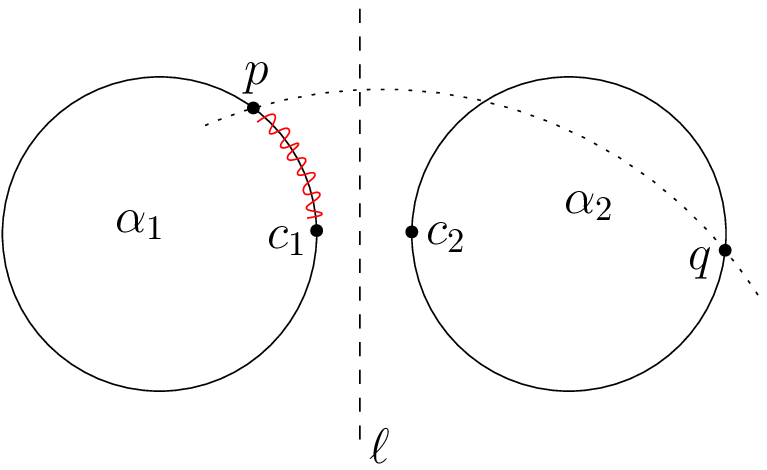}
\caption{\footnotesize Illustrating the case $(A,C)$.}
\label{fig:caseAC}
\end{center}
\end{minipage}
\hspace{0.05in}
\begin{minipage}[t]{0.49\textwidth}
\begin{center}
\includegraphics[height=1.8in]{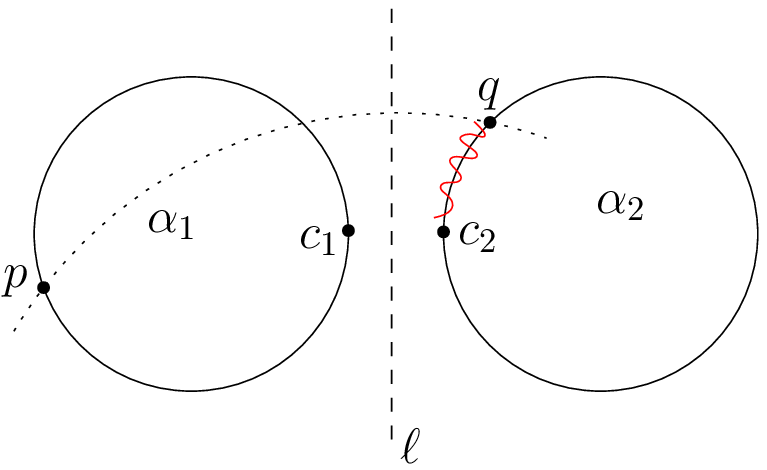}
\caption{\footnotesize Illustrating the case $(C,A)$.}
\label{fig:caseCA}
\end{center}
\end{minipage}
\vspace{-0.15in}
\end{figure}

\item
Case $(C,A)$, which corresponds to Case h. in~\cite{ref:OvermarsMa81}; e.g., see Fig.~\ref{fig:caseCA}. In this case, the part of $\calL_2$ before $q$ can be disregarded.

\item
Case $(C,C)$, which corresponds to Case i. in~\cite{ref:OvermarsMa81}.
In this case, two subcases are further considered in~\cite{ref:OvermarsMa81}. Here, however, we need more subcases.
Depending on whether $(p,q)$ is an $(E,F)$ case, for $E,F\in \{C_1,C_2\}$, there are four subcases.

\begin{enumerate}
\item
Case $(C_1,C_2)$; e.g., see Fig.~\ref{fig:caseC1C2}. In this case, the part of $\calL_2$ after $q$ can be disregarded. Indeed, for each vertex $q'$ in that part, $q'$ is in $C_2$ of $\calL_2$ with respect to $p$. By the definition of $C_2$, there is no unit disk tangent to $\alpha_2$ at $q'$ that covers $p$ (and thus $L$). Therefore, $q'$ cannot be the upper common tangent point, and thus can be disregarded.

\begin{figure}[h]
\begin{minipage}[t]{0.49\textwidth}
\begin{center}
\includegraphics[height=1.8in]{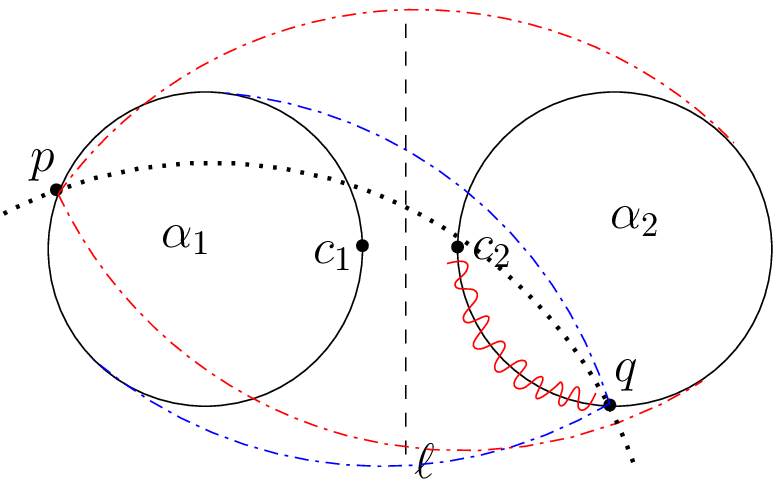}
\caption{\footnotesize Illustrating the case $(C_1,C_2)$. Also shown are the two tangents from $p$ to $\alpha_2$ (red dash-dotted arcs) and the two tangents from $q$ to $\alpha_1$ (blue dash-dotted arcs).}
\label{fig:caseC1C2}
\end{center}
\end{minipage}
\hspace{0.05in}
\begin{minipage}[t]{0.49\textwidth}
\begin{center}
\includegraphics[height=1.8in]{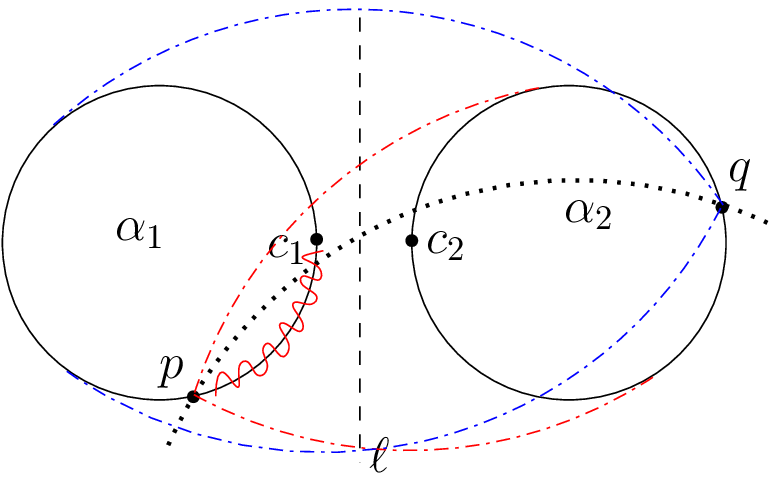}
\caption{\footnotesize Illustrating the case $(C_2,C_1)$.}
\label{fig:caseC2C1}
\end{center}
\end{minipage}
\vspace{-0.15in}
\end{figure}

\item
Case $(C_2,C_1)$; e.g., see Fig.~\ref{fig:caseC2C1}. In this case, the part of $\calL_1$ after $p$ can be disregarded, for the similar reason discussed above.

\item
Case $(C_2,C_2)$. In this case, the part of $\calL_1$ after $p$ and the part of $\calL_2$ after $q$ can be disregarded.

\item
Case $(C_1,C_1)$. In this case, we can find a unit disk $D_1$ that is tangent to $\alpha_1$ at $p$ and covers $q$ and a unit disk $D_2$ that is tangent to $\alpha_2$ at $q$ and covers $p$. Clearly, $D_1$ intersects $D_2$, because each of them contains both $p$ and $q$.

If $D_1=D_2$, then we claim that $ccw(p,q)$ is the lower common tangent. Indeed, since $D_1=D_2$, $D_1$ is tangent to $\alpha_1$ at $p$ and also tangent to $\alpha_2$ at $q$. Thus, either $cw(p,q)$ is the upper common tangent or $ccw(p,q)$ is the lower common tangent. As we know that $p$ is a $C$-vertex of $\calL_1$ with respect to $q$, $p$ cannot be the upper common tangent point and thus $cw(p,q)$ cannot be the upper common tangent. Hence, $ccw(p,q)$ is the lower common tangent.

The claim implies that $a_1$ cannot be after $p$ in $\calL_1$ and $a_2$ cannot be after $q$ in $\calL_2$. Therefore, in this case, the part of $\calL_1$ after $p$ and the part of $\calL_2$ after $q$ can be disregarded.

If $D_1\neq D_2$, then their boundaries intersect at two points. Let $s$ be the intersection point such that if we move from $p$ around $\partial D_1$ clockwise, we will encounter $s$ before the other insertion. Depending on whether $s$ is to the left or right of $\ell$, there are two subcases, which correspond to the two subcases of Case i. in~\cite{ref:OvermarsMa81}.

\begin{figure}[h]
\begin{minipage}[t]{0.49\textwidth}
\begin{center}
\includegraphics[height=2.8in]{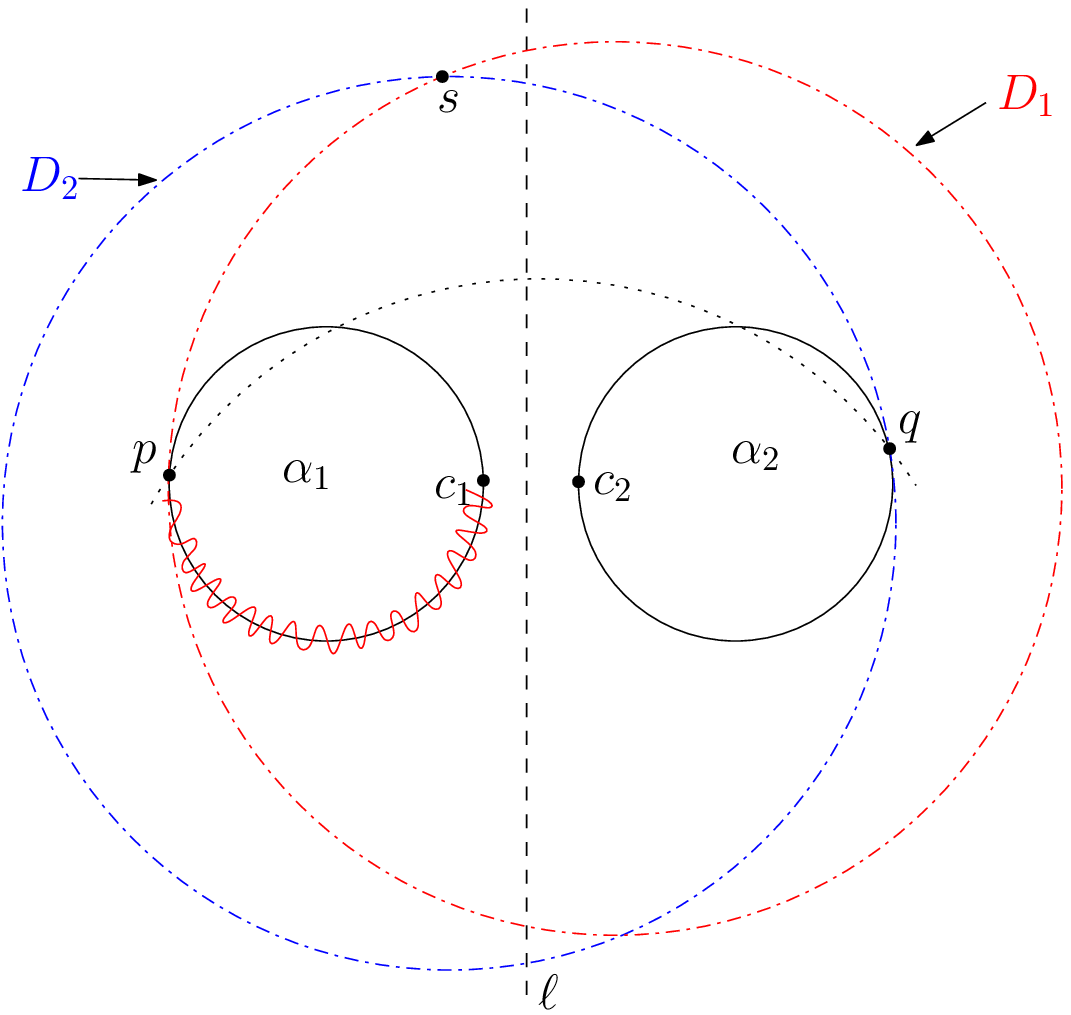}
\caption{\footnotesize Illustrating the case $(C_1,C_1)$, and the intersection $s$ is to the left of $\ell$.}
\label{fig:caseC1C1-1}
\end{center}
\end{minipage}
\hspace{0.05in}
\begin{minipage}[t]{0.49\textwidth}
\begin{center}
\includegraphics[height=2.8in]{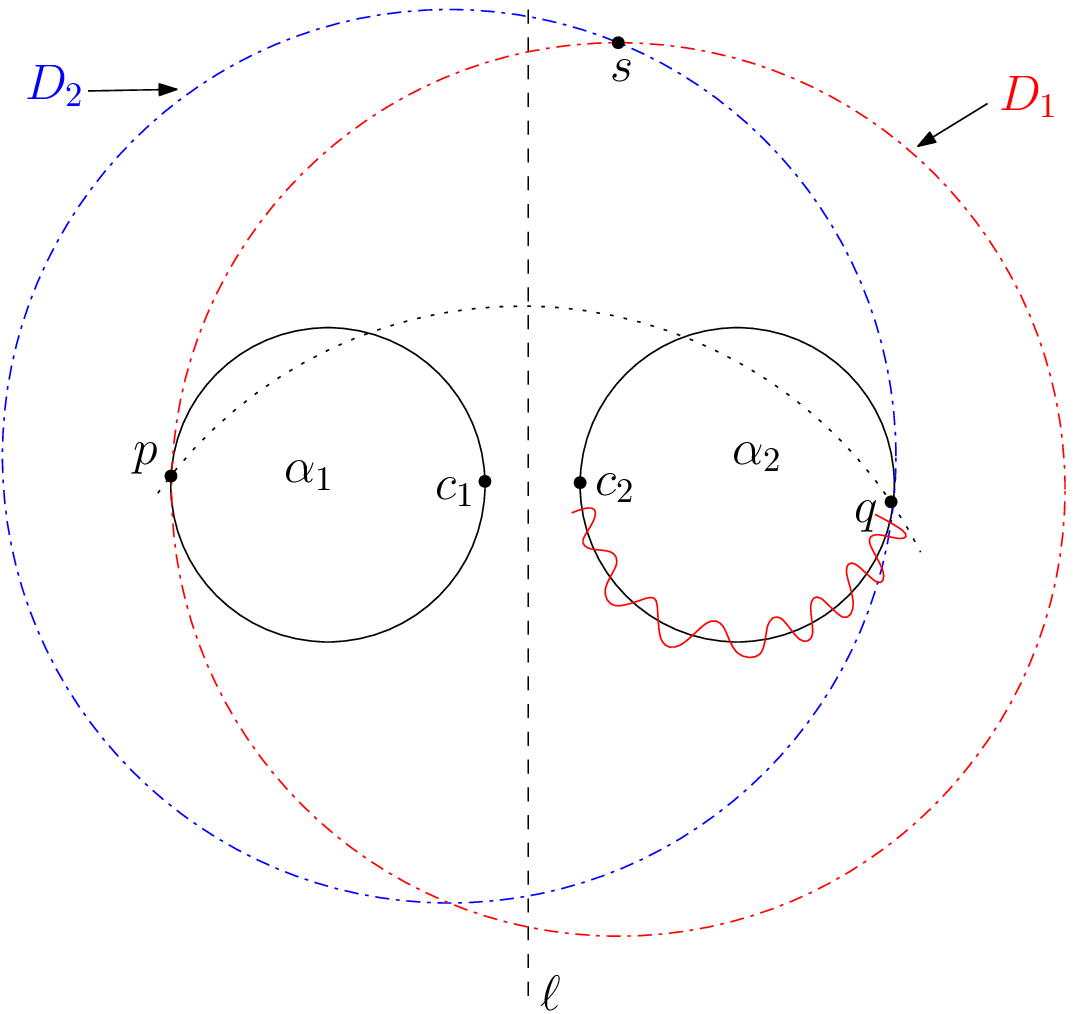}
\caption{\footnotesize Illustrating the case $(C_1,C_1)$, and the intersection $s$ is to the right of $\ell$}
\label{fig:caseC1C1-2}
\end{center}
\end{minipage}
\vspace{-0.15in}
\end{figure}

\begin{enumerate}
  \item If $s$ is to the left of $\ell$, e.g., see Fig.~\ref{fig:caseC1C1-1}, which corresponds to Case i1.in~\cite{ref:OvermarsMa81}, then the part of $\calL_1$ after $p$ can be disregarded.

  \item If $s$ is to the right of $\ell$, e.g., see Fig.~\ref{fig:caseC1C1-2}, which corresponds to Case i2.in~\cite{ref:OvermarsMa81}, then the part of $\calL_2$ after $q$ can be disregarded.
\end{enumerate}

\end{enumerate}

\end{enumerate}

By Lemma~\ref{lem:binarysearch}, with the two neighbors of $p$ and the two neighbors of $q$, each of the above nine cases can be determined in constant time. For the subcases in Case $(C,C)$, recall that given the two neighbors of $p$ in $\alpha_1$, whether $p$ is a $C_1$-vertex with respect to $q$ can be determined in constant time. Similarly, given the two neighbors of $q$ in $\alpha_2$, whether $q$ is a $C_1$-vertex with respect to $p$ can also be determined in constant time. Hence, determining all cases and subcases can be done in constant time.
Therefore, the upper common tangent can be found in $O(\log n)$ time. By a symmetric algorithm, we can compute the lower common tangent in $O(\log n)$ time.

The above algorithm is based on the assumption that $\alpha(L\cup R)$ exists (and thus the common tangents of $\alpha_1$ and $\alpha_2$ exist). If we do not know whether this is true, then we slightly change the algorithm as follows. Suppose we are considering a pair of vertices $(p,q)$ as above. Then, we first check whether $\{p,q\}$ is unit disk coverable. If not, then $\alpha(L\cup R)$ does not exist and we stop the algorithm. Otherwise, we proceed in the same way as before. In addition, if one of $p$ and $q$ is a leaf in its tree and the algorithm still wants to go to a child of that leaf, then we know that the common tangents do not exist and we stop the algorithm.
The runtime of the algorithm is still $O(\log n)$. This proves Lemma~\ref{lem:commontangent}.




%

\bibliographystyle{plain}
\bibliography{reference}

\begin{thebibliography}{10}

\bibitem{ref:AgarwalAn08}
P.K. Agarwal and J.M. Phillips.
\newblock An efficient algorithm for {2D Euclidean} 2-center with outliers.
\newblock In {\em Proceedings of the 16th Annual European Symposium on
  Algorithms (ESA)}, pages 64--75, 2008.

\bibitem{ref:AgarwalPl94}
P.K. Agarwal and M.~Sharir.
\newblock Planar geometric location problems.
\newblock {\em Algorithmica}, 11:185--195, 1994.

\bibitem{ref:AgarwalTh98}
P.K. Agarwal, M.~Sharir, and E.~Welzl.
\newblock The discrete 2-center problem.
\newblock {\em Discrete and Computational Geometry}, 20:287--305, 1998.

\bibitem{ref:ArkinBi15}
E.M. Arkin, J.M. D\'{\i}az-B\'{a}{\~{n}}ez, F.~Hurtado, P.~Kumar, J.S.B.
  Mitchell, B.~Palop, P.~P{\'e}rez-Lantero, M.~Saumell, and R.I. Silveira.
\newblock Bichromatic 2-center of pairs of points.
\newblock {\em Computational Geometry: Theory and Applications}, 48:94--107,
  2015.

\bibitem{ref:BlumTi73}
M.~Blum, R.W. Floyd, V.~Pratt, R.L. Rivest, and R.E. Tarjan.
\newblock Time bounds for selection.
\newblock {\em Journal of Computer and System Sciences}, 7:448--461, 1973.

\bibitem{ref:ChanMo99}
T.M. Chan.
\newblock More planar two-center algorithms.
\newblock {\em Computational Geometry: Theory and Applications}, 13:189--198,
  1999.

\bibitem{ref:ChanA16}
T.M. Chan.
\newblock A simpler linear-time algorithm for intersecting two convex polyhedra
  in three dimensions.
\newblock {\em Discrete and Computational Geometry}, 56:860--865, 2016.

\bibitem{ref:ChazelleAn92}
B.~Chazelle.
\newblock An optimal algorithm for intersecting three-dimensional convex
  polyhedra.
\newblock {\em SIAM Journal on Computing}, 21(4):671--696, 1992.

\bibitem{ref:ChazelleOn96}
B.~Chazelle and J.~Matou\v{s}ek.
\newblock On linear-time deterministic algorithms for optimization problems in
  fixed dimension.
\newblock {\em Journal of Algorithms}, 21:579--597, 1996.

\bibitem{ref:ColeSl87}
R.~Cole.
\newblock Slowing down sorting networks to obtain faster sorting algorithms.
\newblock {\em Journal of the ACM}, 34(1):200--208, 1987.

\bibitem{ref:DobkinDe90}
D.P. Dobkin and D.G. Kirkpatrick.
\newblock Determining the separation of preprocessed polyhedra -- {A} unified
  approach.
\newblock In {\em Proc. of the 17th International Colloquium on Automata,
  Languages and Programming}, volume 443 of {\em Lecture Notes in Computer
  Science}, pages 400--413. Springer, 1990.

\bibitem{ref:DriscollMa89}
J.~Driscoll, N.~Sarnak, D.~Sleator, and R.E. Tarjan.
\newblock Making data structures persistent.
\newblock {\em Journal of Computer and System Sciences}, 38(1):86--124, 1989.

\bibitem{ref:DyerOn86}
M.E. Dyer.
\newblock On a multidimensional search technique and its application to the
  {Euclidean} one centre problem.
\newblock {\em SIAM Journal on Computing}, 15(3):725--738, 1986.

\bibitem{ref:EdelsbrunnerOn83}
H.~Edelsbrunner, D.~Kirkpatrick, and R.~Seidel.
\newblock On the shape of a set of points in the plane.
\newblock {\em IEEE Transactions on Information Theory}, 29:551--559, 1983.

\bibitem{ref:EppsteinDy92}
D.~Eppstein.
\newblock Dynamic three-dimensional linear programming.
\newblock {\em ORSA Journal on Computing}, 4:360--368, 1992.

\bibitem{ref:EppsteinFa97}
D.~Eppstein.
\newblock Faster construction of planar two-centers.
\newblock In {\em Proc. of the 8th Annual ACM-SIAM Symposium on Discrete
  Algorithms (SODA)}, pages 131--138, 1997.

\bibitem{ref:HershbergerA93}
J.~Hershberger.
\newblock A faster algorithm for the two-center decision problem.
\newblock {\em Information Processing Letters}, 1:23--29, 1993.

\bibitem{ref:HershbergerFi91}
J.~Hershberger and S.~Suri.
\newblock Finding tailored partitions.
\newblock {\em Journal of Algorithms}, 3:431--463, 1991.

\bibitem{ref:JaromczykAn94}
J.~Jaromczyk and M.~Kowaluk.
\newblock An efficient algorithm for the {Euclidean} two-center problem.
\newblock In {\em Proceedings of the 10th Annual Symposium on Computational
  Geometry (SoCG)}, pages 303--311, 1994.

\bibitem{ref:KatzAn97}
M.~Katz and M.~Sharir.
\newblock An expander-based approach to geometric optimization.
\newblock {\em SIAM Journal on Computing}, 26(5):1384--1408, 1997.

\bibitem{ref:KimEf00}
S.K. Kim and C.-S. Shin.
\newblock Efficient algorithms for two-center problems for a convex polygon.
\newblock In {\em Proceedings of the 6th International Computing and
  Combinatorics Conference (COCOON)}, pages 299--309, 2000.

\bibitem{ref:MegiddoAp83}
N.~Megiddo.
\newblock Applying parallel computation algorithms in the design of serial
  algorithms.
\newblock {\em Journal of the ACM}, 30(4):852--865, 1983.

\bibitem{ref:MegiddoLi83}
N.~Megiddo.
\newblock Linear-time algorithms for linear programming in {$R^3$} and related
  problems.
\newblock {\em SIAM Journal on Computing}, 12(4):759--776, 1983.

\bibitem{ref:MegiddoOn84}
N.~Megiddo and K.J. Supowit.
\newblock On the complexity of some common geometric location problems.
\newblock {\em SIAM Journal on Computing}, 13:182--196, 1984.

\bibitem{ref:OvermarsMa81}
M.~Overmars and J.~van Leeuwen.
\newblock Maintenance of configurations in the plane.
\newblock {\em Journal of Computer System Sciences}, 23(2):166--204, 1981.

\bibitem{ref:SarnakPl86}
N.~Sarnak and R.E. Tarjan.
\newblock Planar point location using persistent search trees.
\newblock {\em Communications of the ACM}, 29:669--679, 1986.

\bibitem{ref:SharirA97}
M.~Sharir.
\newblock A near-linear algorithm for the planar 2-center problem.
\newblock {\em Discrete and Computational Geometry}, 18:125--134, 1997.

\bibitem{ref:TanSi17}
X.~Tan and B.~Jiang.
\newblock Simple {$O(n\log^2n)$} algorithms for the planar 2-center problem.
\newblock In {\em Proceedings of the 23rd International Computing and
  Combinatorics Conference (COCOON)}, pages 481--491, 2017.

\bibitem{ref:WangIm19}
H.~Wang and J.~Xue.
\newblock Improved algorithms for the bichromatic two-center problem for pairs
  of points.
\newblock In {\em Proceedings of the 16th Algorithms and Data Structures
  Symposium (WADS)}, pages 578--591, 2019.

\end{thebibliography}

\appendix

\appendix
\section*{Appendix}

We provide a counterexample to show that Tan and Jiang's algorithm~\cite{ref:TanSi17} is not correct. We follow the same notation as in~\cite{ref:TanSi17} without further explanations. The authors first gave an algorithm for the convex position case where $S$ is in convex position, and then use it to solve the general case. Their algorithm uses binary search that relies on a monotonicity property given in Theorem 1. The argument of the proof does not stand. For example, because $r_1^*$ is adjustable, the authors claim that $r_1^*\geq r_2^*$ due to Lemma 3. But Lemma 3 does not imply that at all. Nevertheless, we provide a counterexample to demonstrate that the monotonicity property claimed in Theorem 1 does not hold.

\begin{figure}[h]
\begin{minipage}[t]{\textwidth}
\begin{center}
\includegraphics[height=3.0in]{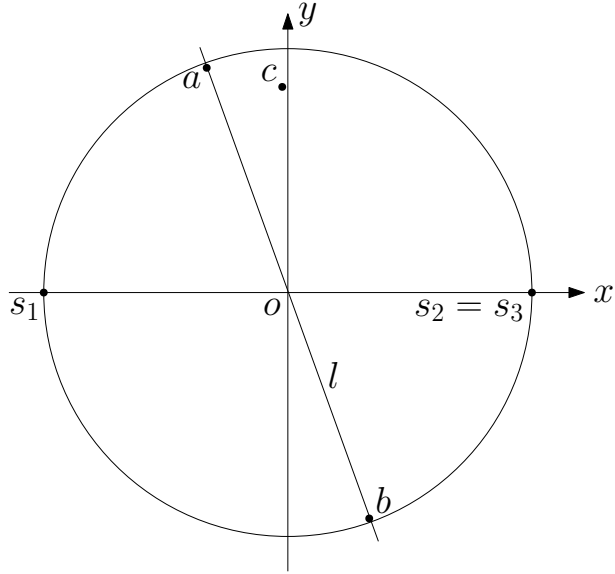}
\caption{\footnotesize Illustrating a counterexample for Theorem 1 in~\cite{ref:TanSi17}.}
\label{fig:counterex}
\end{center}
\end{minipage}
\vspace{-0.15in}
\end{figure}

Refer to Fig.~\ref{fig:counterex}. $S=\{s_1,a,b,c,s_2\}$. A circle $C$ centered at the origin $o$ contains all five points.  $s_1$ and $s_2$ are the two intersections of $x$-axis and $C$. $a,b,c$ are all in the interior of $C$. Hence, $C$ is the smallest enclosing circle of $S$. By definition, we have $s_2=s_3$. $a$ and $b$ are on a line $l$ through $o$ such that $a$ is in the second quadrant and $b$ is in the fourth quadrant. $l$ and $y$-axis form a relatively small angle. Both $a$ and $b$ are arbitrarily close to the boundary of $C$ so that any circle enclosing both $a$ and $b$ has a radius very close to $r$ or larger than $r$.

For any two points $p$ and $q$, let $|pq|$ denote their Euclidean distance.

We can pick the points $a,b,c$ to guarantee the following properties (although we do not provide their exact coordinates, one can verify that the example in Fig.~\ref{fig:counterex} satisfies these properties): (1) $|oa|=|ob|$ (and thus $|s_1b|=|s_2a|$ and $|s_2b|=|s_1a|$); (2) $|s_1a|<|s_1c|<|s_1b|<|bc|$; (3) $r(\{s_1,a,c\})=|s_1c|/2$; (4) $r(\{c,s_2,b\})=|bc|/2$; (5) $r(\{a,c,s_2\})=|as_2|/2$.

With the above properties, one can verify that the following holds (again, refer to~\cite{ref:TanSi17} for the definitions of the notation).
$r_1^*=\max\{r(\{s_1,b\}),r(\{a,c,s_2\})\}=\max\{|s_1b|/2,|as_2|/2\}$=$|s_1b|/2$. $r_2^*=\max\{r(\{s_1,a\}),r(\{c,s_2,b\})\}=\max\{|s_1a|/2,|bc|/2\}=|bc|/2$.
$r_3^*=\max\{r(\{s_1,a,c\}),r(\{s_2,b\})\}=\max\{|s_1c|/2,|s_2b|/2\}=|s_1c|/2$.
Due to that $|s_1c|<|s_1b|<|bc|$, we obtain $r_3^*<r_1^*<r_2^*$. Therefore, $r^*=r_3^*$, and according to Theorem 1 of~\cite{ref:TanSi17}, $r_1^*\geq r_2^*\geq r_3^*$ should hold, which contradicts with $r_3^*<r_1^*<r_2^*$. Hence, Theorem 1 is not correct.


\end{document}